\documentclass[leqno]{article}

\usepackage[frenchb,english]{babel}
\usepackage[utf8]{inputenc}
\usepackage{amsmath}
\usepackage{amssymb}
\usepackage{amsfonts}
\usepackage{enumerate}
\usepackage{vmargin}
\usepackage[all]{xy}
\usepackage{mathrsfs}
\usepackage{mathtools}
\usepackage{lmodern}
\usepackage{slashed}
\usepackage[colorlinks=true,linkcolor=blue,pagebackref=true]{hyperref}%
\setmarginsrb{3cm}{3cm}{3.5cm}{3cm}{0cm}{0cm}{1.5cm}{3cm}
\usepackage{comment}

\usepackage{accents}
\newlength{\dhatheight}
\newcommand{\doublehat}[1]{%
    \settoheight{\dhatheight}{\ensuremath{\hat{#1}}}%
    \addtolength{\dhatheight}{-0.2ex}%
    \hat{\vphantom{\rule{1pt}{\dhatheight}}%
    \smash{\hat{#1}}}}

\newcommand{\C}{\mathrm{C}}

\newcommand{\E}{\ensuremath{\mathbb{E}}}

\newcommand{\N}{\ensuremath{\mathbb{N}}}
\newcommand{\B}{\mathrm{B}} %Bounded operators
\let\H\relax %\H ne fait plus rien !
\newcommand{\H}{\mathrm{H}}
\newcommand{\I}{\mathrm{I}} %matrice identit\'e
\newcommand{\II}{\mathrm{II}} %type II

\let\L\relax %\L ne fait plus rien !
\newcommand{\L}{\mathrm{L}}
\newcommand{\Q}{\mathrm{Q}}
\newcommand{\Mult}{\mathrm{Mult}}
\newcommand{\scr}{\mathscr}

\newcommand{\Aut}{\mathrm{Aut}}

\newcommand{\M}{\mathrm{M}}
\newcommand{\CB}{\mathrm{CB}}

\newcommand{\HS}{\mathrm{HS}}
\newcommand{\NP}{\mathrm{NP}}

\let\cal\relax
\newcommand{\cal}{\mathcal}
\newcommand{\Z}{\ensuremath{\mathbb{Z}}}
\newcommand{\R}{\ensuremath{\mathbb{R}}}
\newcommand{\T}{\ensuremath{\mathbb{T}}}

\newcommand{\Id}{\mathrm{Id}}
\newcommand{\VN}{\mathrm{VN}}
\newcommand{\la}{\langle}
\newcommand{\ra}{\rangle}
\newcommand{\EA}{\mathrm{EA}} %entanglement assisted

\newcommand{\flip}{\mathrm{flip}}

\renewcommand{\leq}{\ensuremath{\leqslant}}
\renewcommand{\geq}{\ensuremath{\geqslant}}
\newcommand{\qed}{\hfill \vrule height6pt  width6pt depth0pt}
\newcommand{\bnorm}[1]{ \big\| #1  \big\|}

\newcommand{\norm}[1]{\left\Vert#1\right\Vert}
\newcommand{\co}{\colon}
\newcommand{\otp}{\widehat{\ot}}

\newcommand{\ot}{\otimes}
\newcommand{\ovl}{\overline}
\newcommand{\otvn}{\ovl\ot}

\newcommand{\dec}{\mathrm{dec}}

\newcommand{\cb}{\mathrm{cb}}
\newcommand{\op}{\mathrm{op}} %opposed structure
\let\i\relax %\i ne fait plus rien !
\newcommand{\i}{\mathrm{i}}

\newcommand{\ov}{\overset}
\newcommand{\sa}{\mathrm{sa}}

\newcommand{\JBW}{\mathrm{JBW}}

\newcommand{\epsi}{\varepsilon}
\renewcommand{\d}{\mathop{}\mathopen{}\mathrm{d}} %op\'erateur diff\'erentiel
\newcommand{\e}{\mathrm{e}} %constante e
\renewcommand{\d}{\mathop{}\mathopen{}\mathrm{d}}

\DeclareMathOperator{\Span}{span} %sev engendre
\DeclareMathOperator{\tr}{Tr} %trace
\let\Re\relax %\Re ne fait plus rien
\DeclareMathOperator{\Re}{Re} %partie r\'eelle
\newtheorem{thm}{Theorem}[section]

\newtheorem{prop}[thm]{Proposition}

\newtheorem{cor}[thm]{Corollary}
\newtheorem{lemma}[thm]{Lemma}

\newtheorem{remark}[thm]{Remark}
\newtheorem{example}[thm]{Example}

\newenvironment{proof}[1][]{\noindent {\it Proof #1} : }{\hbox{~}\qed
\smallskip
}

\usepackage{tocloft}
\numberwithin{equation}{section}
\usepackage[nottoc,notlot,notlof]{tocbibind}%biblio dans sommaire
\let\OLDthebibliography\thebibliography
\renewcommand\thebibliography[1]{
  \OLDthebibliography{#1}
  \setlength{\parskip}{0pt}
  \setlength{\itemsep}{0pt plus 0.3ex}
}

\newcommand\reallywidehat[1]{\arraycolsep=0pt\relax%
\begin{array}{c}
\stretchto{
  \scaleto{
    \scalerel*[\widthof{\ensuremath{#1}}]{\kern-.5pt\bigwedge\kern-.5pt}
    {\rule[-\textheight/2]{1ex}{\textheight}} %WIDTH-LIMITED BIG WEDGE
  }{\textheight} % 
}{0.5ex}\\           % THIS SQUEEZES THE WEDGE TO 0.5ex HEIGHT
#1\\                 % THIS STACKS THE WEDGE ATOP THE ARGUMENT
\rule{-1ex}{0ex}
\end{array}
}

\begin{document}
\selectlanguage{english}
\title{\bfseries{Entanglement-assisted classical capacities of some channels acting as radial multipliers on fermion algebras}}
\date{}
\author{\bfseries{C\'edric Arhancet}}
%Coulhon-Varopulos dimension, curvature,and functional calculus
%On the spectral dimension of spectral triples
%associated to Markov semigroups of operators
\maketitle

%%%%%%%%%%%%%%%%%%%%%%%%%%%%%%%%%%%%%%%%%%%%%%%%%%%%%%%%%%%%%%%%
%%%%%%%%%%%%%%%%%%%%%%%%%%%%%%%%%%%%%%%%%%%%%%%%%%%%%%%%%%%%%%%%

\begin{abstract}
We investigate a new class of unital quantum channels on $\mathrm{M}_{2^k}$, acting as radial multipliers when we identify the matrix algebra $\mathrm{M}_{2^k}$ with a finite-dimensional fermion algebra. Our primary contribution lies in the precise computation of the (optimal) rate at which classical information can be transmitted through these channels from a sender to a receiver when they share an unlimited amount of entanglement. Our approach relies on new connections between fermion algebras with the $n$-dimensional discrete hypercube $\{-1,1\}^n$. Significantly, our calculations yield exact values applicable to the operators of the fermionic Ornstein-Uhlenbeck semigroup. This advancement not only provides deeper insights into the structure and behaviour of these channels but also enhances our understanding of Quantum Information Theory in a dimension-independent context.
%\textbf{A FAIRE : changer les titres avec Cantor group}
\end{abstract}%to any such quantum channels in any dimension and can be applied to
% and more generally the Cantor group

%\vspace{0.2cm}
%
%-calcul de de l'entropy pour le semigroup fermionic
%
%\emph{\textbf{-lemme L1Lp general}}
%
%-mieux definir fermion algebras

%%%%%%%%%%%%%%%%%%%%%%%%%%%%%%%%%%%%%%%%%%%%%%%%%%%%%%%%%%%%%%%%
%%%%%%%%%%%%%%%%%%%%%%%%%%%%%%%%%%%%%%%%%%%%%%%%%%%%%%%%%%%%%%%%

\makeatletter
 \renewcommand{\@makefntext}[1]{#1}
 \makeatother
 \footnotetext{
 %The authors are supported by the research program ANR-18-CE40-0021 (project HASCON).\\
 2020 {\it Mathematics subject classification:}
 46L51, 94A40, 46L07, 81P45. %  46L51  Noncommutative measure and integration
%  46M35       Abstract interpolation of topological vector spaces [See also 46B70]
% 46L07       Operator spaces and completely bounded maps [See also 47L25]
% 43A22 Homomorphisms and multipliers of function spaces on groups, semigroups, etc.
% 43A15 Lp-spaces and other function spaces on groups, semigroups, etc.
% 47D03 Groups and semigroups of linear operators For nonlinear operators.
% 58B34 	Noncommutative geometry (à la Connes)
% 81P45   	Quantum information, communication, networks (quantum-theoretic aspects)
% 94A40   	Channel models (including quantum) in information and communication theory
\\
{\it Key words}: noncommutative $\L^p$-spaces, multipliers, fermion algebras, entropies, quantum channels, discrete hypercube, quantum information theory.}
%noncommutative $L^p$-spaces, operator spaces

{
  \hypersetup{linkcolor=blue}
 \tableofcontents
}

%%%%%%%%%%%%%%%%%%%%%%%%%%%%%%%%%%%%%%%%%%%%%%%%%%%%%%%%%%%%%%%%%%%%%%%%%%%%%%%%%%%%%%%%%%%%%%%%%%%%%%%%%%%%%%%%%%%%%%%%%%%%%%%%%%%%%%%%%%%%%%%%%%%%%%
%%%%%%%%%%%%%%%%%%%%%%%%%%%%%%%%%%%%%%%%%%%%%%%%%%%%%%%%%%%%%%%%%%%%%%%%%%%%%%%%%%%%%%%%%%%%%%%%%%%%%%%%%%%%%%%%%%%%%%%%%%%%%%%%%%%%%%%%%%%%%%%%%%%%%%
%%%%%%%%%%%%%%%%%%%%%%%%%%%%%%%%%%%%%%%%%%%%%%%%%%%%%%%%%%%%%%%%%%%%%%%%%%%%%%%%%%%%%%%%%%%%%%%%%%%%%%%%%%%%%%%%%%%%%%%%%%%%%%%%%%%%%%%%%%%%%%%%%%%%%%

\section{Introduction}
\label{sec:Introduction}

One of the primary objectives of Quantum Information Theory is to compute the optimal rates of information transmission (be it classical or quantum) through quantum channels. Quantum channels characterize the evolution of quantum states during transmission via noisy physical mediums. We refer to the books \cite{Wat18} and \cite{Wil17} for more information. To achieve this, numerous capacities and entropies have been introduced to describe the ability of a channel to convey information from a sender to a receiver. Determining these values is crucial in defining the limits of the capacity of a quantum channel for reliable information transmission. For a comprehensive overview of the various capacities associated with a quantum channel, we recommend the survey article \cite{GIN18}.

In the Schrödinger picture, a quantum channel is understood as a trace preserving completely positive map 
$T \co S^1_n \to S^1_n$. The concept of minimum output entropy for a quantum channel is defined by
\begin{equation}
\label{Def-minimum-output-entropy-ini}
\H_{\min}(T) \overset{\mathrm{def}}{=} \min_{x \geq 0, \tr x = 1} \H(T(x)).
\end{equation}
In this context, $\H(\rho) \ov{\mathrm{def}}{=} -\tr(\rho\log_2 \rho)$ is used to denote the von Neumann entropy, which is essentially the quantum counterpart of Shannon entropy, acting as a measure of uncertainty for $\rho$. For an introductory perspective on this topic, one might refer to \cite{Pet01}. It is noteworthy that $\H(\rho)$ maintains a non-negative value, achieving its maximum when $\rho$ is maximally mixed and descending to zero in the case of a pure state. Therefore, the minimum output entropy can be interpreted as a measure of how much the purity of states degrades through the channel. 

Note that  Hastings proved in \cite{Has09} that the minimum output entropy is not additive, i.e.~we have $\H_{\min}(T_1 \ot T_2) \not= \H_{\min}(T_1)+\H_{\min}(T_2)$ in general for quantum channels $T_1$ and $T_2$. Its proof relies on probabilistic constructions, and finding a concrete counterexample is one of the most important open problems in quantum information theory. Finally, observe that for a unital qubit channel $T \co S^1_2 \to S^1_2$, a formula of $\H_{\min}(T)$ is known, see \cite[Example 8.10 p.~153]{{Hol19}}. Moreover, King has proved in \cite[Theorem 1]{Kin02} that the minimum output entropy $\H_{\min}$ is additive if one of the two channels $T_1$ or $T_2$ is a unital qubit channel. The same thing is true if $T_1$ or $T_2$ is entangling breaking, see \cite[Exercise 8.2 p.~556]{Wat18}.

We can generalize the definition of the minimum output entropy to quantum channels $T \co \L^1(\mathcal{M})\to \L^1(\cal{M})$ acting on a noncommutative $\L^1$-space $\L^1(\cal{M})$ associated with a finite-dimensional von Neumann algebra equipped with a finite faithful trace. In this broader context, it is well-known \cite[Remark 9.5]{Arh24b} that we can express the minimum output entropy as a derivative of suitable operators norms:
\begin{equation}
\label{Smin-as-derivative-intro}
\H_{\min,\tau}(T)
=-\frac{1}{\log 2}\frac{\d}{\d p} \big[\norm{T}_{\L^1(\mathcal{M}) \to \L^p(\mathcal{M})}\big]|_{p=1}. 
\end{equation}
In this notation, we explicitly denote the dependency of this quantity on the chosen trace $\tau$. If the von Neumann $\cal{M}$ is the matrix algebra $\M_n$ equipped with its non-normalized trace $\tr$, $\H_{\min,\tr}(T)$ is equal to the minimum output entropy $\H_{\min}(T)$ defined in \eqref{Def-minimum-output-entropy-ini}.

Calculating the minimum output entropy for a specific channel is generally a challenging task. Indeed, it is announced in \cite{BeS08} that the problem of deciding whether the minimum entropy of any channel is less than a constant is $\NP$-complete. However, in the case of some specific quantum channels, it is occasionally possible to compute this quantity. Notably, in \cite{Arh24b}, the minimum output entropy of all Fourier multipliers on co-amenable compact quantum groups of Kac type is precisely computed and even for convolution operators on even more abstract structures arising from Jones's subfactor theory. Furthermore, in addition to the minimum output entropy, \cite{Arh24b} also provides a precise computation of its completely bounded version, which is defined by
\begin{equation}
\label{Def-intro-Scb-min}
\H_{\cb,\min,\tau}(T)
\ov{\mathrm{def}}{=} -\frac{1}{\log 2}\frac{\d}{\d p} \big[\norm{T}_{\cb,\L^1(\mathcal{M}) \to \L^p(\mathcal{M})} \big]|_{p=1}. 
\end{equation} 
Here, the subscript $\cb$ means <<completely bounded>> and is natural in operator space theory, see the books \cite{BLM04}, \cite{EfR00}, \cite{Pau02}, \cite{Pis98} and \cite{Pis2}. This concept was initially introduced for matrix algebras in \cite{DJKRB06} and later independently identified in \cite{GPLS09}. See also \cite{YHW19}. For a channel $T \co S^1_n \to S^1_n$, we have the following variational expression
$$
\H_{\cb,\min,\tr}(T)
=\inf_{\rho} \Big\{\H\big[(T \ot \Id)(|\psi\ra \la\psi|)\big]- \H(\rho) \Big\},
$$
where the infimum is taken on all states $\rho$ in $S^1_n$ and where $|\psi\ra$ is a purification of $\rho$. We can restrict the infimum to the pure states. We always have $-\log_2 n \leq \H_{\cb,\min,\tr}(T) \leq \log_2 n$. The lower bound is attained for the identity channel $\Id_{S^1_n}$ and the upper bound for the completely noisy channel. Moreover, if $T \co S^1_n \to S^1_n$ is entanglement breaking, we have $\H_{\cb,\min,\tr}(T) \geq 0$ by \cite{DJKRB06}. Finally, if the trace $\tau$ is \textit{normalized}, it is observed in \cite{Arh24b} that 
\begin{equation}
\label{cbminless0}
\H_{\cb,\min,\tau}(T) \leq 0
\end{equation}
for any quantum channel $T \co \L^1(\mathcal{M})\to \L^1(\cal{M})$.

A significant difference between the minimum output entropy and its completely bounded version is that the second quantity is additive, i.e.~for any quantum channels $T_1$ and $T_2$ we have
\begin{equation}
\label{additivity-cb}
\H_{\cb,\min,\tau}(T_1 \ot T_2)
=\H_{\cb,\min,\tau}(T_1)+\H_{\cb,\min,\tau}(T_2).
\end{equation}
We refer to \cite{GuW15} for a nice application of this additivity to the strong converse of entanglement-assisted capacity. It is worth noting that it is elementary to check that for a \textit{unital} quantum channel $T \co S^1_n \to S^1_n$, the value $-\H_{\cb,\min,\tr}(T)$ gives a lower bound of the channel coherent information $\Q^{(1)}(T)$ defined in \cite[Definition 8.14 p.~474]{Wat18} by 
\begin{equation}
\label{def-Q1}
\Q^{(1)}(T)
\ov{\mathrm{def}}{=} \max_{\rho} \big\{ \H(T(\rho)) -\H((T \ot \Id)(|\psi\ra \la\psi|)) \big\},
\end{equation}
where the maximum is taken on all states $\rho$ in $S^1_n$ and where $|\psi\ra$ is a purification of $\rho$.

In \cite{Arh24b}, we also compute the entanglement-assisted classical capacity $\C_{\EA}(T)$ of Fourier multipliers. This is the highest rate at which classical information can be transmitted from a sender to a receiver when they share an unlimited amount of entanglement. We refer to \cite[Definition 8.6 p.~469]{Wat18} and \cite[p.~574]{Wil17} for a precise definition. Note that for a quantum channel $T \co S^1_n \to S^1_n$, a classical theorem \cite{BSST02} of Bennett, Shor, Smolin and Thapliyal gives the following  remarkably simple variational expression  
$$
\C_{\EA}(T)
=\max_{\rho} \big\{ \H(\rho)+\H(T(\rho))-\H((T \ot \Id)(|\psi \ra \la\psi)) \big\}
$$
for this quantity, where the maximum is taken on all states $\rho$ in $S^1_n$ and where $|\psi\ra$ is a purification of $\rho$. We refer also to \cite[Chapter 9]{Hol19} and \cite[Theorem 8.41 p.~510]{Wat18} for another proof of this important result. It is worth mentioning that for a quantum channel $T \co S^1_n \to S^1_n$ which is <<covariant>> with respect to a compact group and some suitable representations, the entanglement-assisted classical capacity $\C_{\EA}(T)$ is related to the completely bounded minimum output entropy by the formula $\C_{\EA}(T)=\log_2 n-\H_{\cb,\min,\tr}(T)$ of \cite[p.~385]{JuP15}.
%\begin{equation}
%\label{formula-CEA-covariant}
%\end{equation}

In summary, the paper \cite{Arh24b} reveals that noncommutative harmonic analysis on quantum groups provides a huge class of quantum channels, which offer the potential for calculating \textit{exact} values of various quantities. Additionally, it is demonstrated in \cite{Arh24b} how the theory of unital qubit channels $T \co S^1_2 \to S^1_2$ is encompassed within the theory of convolution operators on the noncommutative group algebra $\VN(\mathbb{Q}_8)$ of the quaternion group $\mathbb{Q}_8$ (the case of channels  $T \co S^1_n \to S^1_n$, with $n \geq 3$ has not yet been explored). 

This paper continues in the same vein by exploring other types of multipliers stemming from noncommutative harmonic analysis on fermionic analysis. We introduce a notion of radial multiplier on the fermion algebra $\scr{C}(H)$, where $H$ is a separable real Hilbert space with dimension $n \in \{2,\ldots,\infty\}$. The fermion algebras are also called complex Clifford algebras and arise from the study of systems of  fermions, particles (like electrons) obeying the Pauli exclusion principle. In this paper, we focus on fermion algebras associated with a finite-dimensional real Hilbert space $H$. In this setting, these algebras are finite-dimensional. Furthermore, we will retain the conventional notation $\scr{C}(H)$, even though the real Hilbert space $H$ is no longer explicitly present.

These von Neumann algebras are in particular generated by a <<spin system>>\footnote{\thefootnote. We warn the reader that the converse is false if the real Hilbert space $H$ is infinite-dimensional. Indeed, the family of the $s_i$'s can generate a von Neumann algebra of type III. We refer to \cite[Remark 1.2.11 p.~35]{ARU97} for more information. If the real Hilbert space $H$ is infinite-dimensional, we have to assume that that the generated von Neumann algebra is $*$-isomorphic to the unique approximately finite-dimensional factor of type $\II_1$ with separable predual. }, i.e.~a family $(s_i)_{1 \leq i \leq n}$ of selfadjoint unitary operators $\not=\pm \Id$ acting on a complex Hilbert space $\cal{H}$ which satisfy the anticommutation rule
\begin{equation}
%\label{anticommute}
s_is_j
=-s_js_i, \quad i \not= j.
\end{equation}

Moreover, this von Neumann algebra $\scr{C}(H)$ is equipped with a \textit{normalized} normal finite faithful trace $\tau$ which we will describe later, allowing to introduce the noncommutative $\L^p$-space $\L^p(\scr{C}(H))$ as the completion of the vector space $\scr{C}(H)$ with respect to the norm
\begin{equation*}
\norm{x}_{\L^p(\scr{C}(H))}
\ov{\mathrm{def}}{=} \tau(|x|^p)^{\frac{1}{p}}, \quad x \in \scr{C}(H),
\end{equation*} 
where $1 \leq p <\infty$ and where $|x|\ov{\mathrm{def}}{=} (x^*x)^{\frac{1}{2}}$. There exists a canonical orthonormal basis of the complex Hilbert space $\L^2(\scr{C}(H))$ that can be introduced as follows. For any non-empty finite subset $A=\{i_1,\ldots,i_k\}$ of the set $\{1,\ldots,n\}$ with $i_1 < \cdots < i_k$, we set
\begin{equation*}
s_A 
\ov{\mathrm{def}}{=} s_{i_1} \cdots s_{i_k}.
\end{equation*}
We equally let $s_\emptyset \ov{\mathrm{def}}{=} 1$. It is known that the $s_A$'s define an orthonormal basis of the Hilbert space $\L^2(\scr{C}(H))$. If the real Hilbert space $H$ is finite-dimensional, it is worth observing that the $s_A$'s form a basis of the algebra $\scr{C}(H)$. If $H$ is infinite-dimensional then the subspace $\Span \{s_A : A \text{ finite subset of } \N \}$ is weak* dense in the von Neumann algebra $\scr{C}(H)$.

Now, we are able to describe the quantum channels investigated in this work, referred to as radial multipliers. We say that a (weak* continuous) linear map $T \co \scr{C}(H) \to \scr{C}(H)$ is a <<radial multiplier>> if there exists a complex function $\phi \co \{0,\ldots,n\} \to \mathbb{C}$ such that
\begin{equation*}
%\label{def-rad-mult}
T(s_A)
= \phi(\vert A\vert) s_A, \quad A \subset \{1,\ldots,n\}.
\end{equation*}
The quantum channels considered in our paper are the trace preserving completely positive maps $T \co \L^1(\scr{C}(H)) \to \L^1(\scr{C}(H))$ induced by radial multipliers. For each of these channels, we compute in Corollary \ref{Cor-ent-cb} the completely bounded minimum output entropy $\H_{\cb,\min,\tau}(T)$ defined in \eqref{Def-intro-Scb-min}. We also perform the computation of the entanglement-assisted classical capacity $\C_{\EA}(T)$ in Theorem \ref{th-capacity-assisted}. %We also characterize radial multipliers with zero quantum capacity in Corollary \ref{Cor-breaking}.

Note that if $\cal{S}$ denote the weak closure of the \textit{real} linear space spanned by the $s_i$'s then $\R1+\cal{S}$ is a concrete realization of a <<spin factor>>. This Jordan subalgebra of $\B(\cal{H})$ is an important example of $\JBW$-factor, studied by Topping in \cite{Top65} and \cite{Top66}. Indeed, by \cite[Proposition 6.1.8 p.~138]{HOS84} these $\JBW$-algebras are exactly the $\JBW$-factors of type $\I_2$. Finally, we refer also \cite[Appendix]{GKR24} for more information on spin factors. 

Several concrete realizations of these operators exist. For example, we can consider $s_{i} \ov{\mathrm{def}}{=} a_i+a_i^*$ where the $a_i$'s are <<creation operators>> over fermionic Fock spaces. Another construction is possible using tensors of $2 \times 2$ complex matrices, see e.g.~\cite[p.~140]{HOS84}. Actually, if the dimension $n=2k$ of the Hilbert space $H$ is even, we have a $*$-isomorphism between the fermion algebra $\scr{C}(H)$ and the matrix algebra $\M_{2^k}$. So the previous <<radial multipliers>> can be identified with quantum channels on matrix algebras. In the lowest dimension $n=2$, we observe that we can recover in particular the dephasing channels on the matrix algebras $\M_2$,  see Example \ref{Ex-multiplier-Fermions}. Finally, if $n=2k-1$ is odd, the fermion algebra $\scr{C}(H)$ identifies to a sum $\M_{2^{k-1}} \oplus \M_{2^{k-1}}$ of matrix algebras.

We will observe in  Example \ref{Ex-Ornstein} that the operators of the fermionic Ornstein-Uhlenbeck semigroup $(T_t)_{t \geq 0}$ are radial multipliers. This Markovian semigroup was introduced by Gross in \cite{Gro72} and is a fermionic analogue of the classical Ornstein-Uhlenbeck semigroup. For a better  understanding of this semigroup, we suggest consulting the references \cite{CaL93}, \cite{Hu93} for hypercontractivity, \cite{CaM14} for the fermionic Fokker-Planck equation, and \cite[Section 9.B]{JMX06} \cite{ArK24} for the functional calculus of its generator.

\paragraph{Approach of the Paper}
We begin to identify the completely bounded norm of each radial multiplier from the space $\L^1(\scr{C}(H))$ into the Banach space $\L^p(\scr{C}(H))$ if $p > 1$. This computation relies on new connections with some multipliers on the $n$-dimensional discrete hypercube $\{-1,1\}^n$ %(or the Cantor group $\Omega$ if the real Hilbert space $H$ is infinite-dimensional)
relying in part on a group action of this group on the fermion algebra $\scr{C}(H)$. This allows us to deduce the completely bounded minimum output entropy by differentiating with the formula \eqref{Def-intro-Scb-min}.

Our calculation of the entanglement-assisted classical capacity is based on the following formula \cite[Theorem 1.1 p.~355]{JuP15} of Junge and Palazuelos
\begin{equation}
\label{CEA-formula}
\C_{\EA}(T)
=\frac{1}{\log 2} \frac{\d}{\d p}\big[\norm{T^*}_{\pi_{p^*}^\circ, \M_n \to \M_n}\big]|_{p=1},
\end{equation}
which describes this capacity as a derivative of completely $p$-summing norms of the adjoint map $T^* \co \M_n \to \M_n$ of the quantum channel $T \co S^1_n \to S^1_n$. Here $p^* \ov{\mathrm{def}}{=} \frac{p}{p-1}$. For any $1 \leq p <\infty$, the notion of completely $p$-summing map is the noncommutative version of the classical concept of $p$-summing map between Banach spaces, initially introduced by Pietsch in \cite{Pie66}. If $X$ and $Y$ are Banach spaces, recall that a linear operator $T \co X \to Y$ is $p$-summing if it induces a bounded operator $\Id \ot T \co \ell^p \ot_{\epsi} X \to \ell^p(Y)$, where $\ot_{\epsi}$ denotes the injective tensor product. We refer to the books \cite{DeF93} and \cite{DJT95} for a comprehensive presentation of these operators and their important applications in Banach space theory. Pisier introduced a noncommutative theory of $p$-summing maps between operator spaces (=noncommutative Banach spaces), introducing the concept of completely $p$-summing maps in \cite{Pis98}. If $E$ and $F$ are operator spaces, recall that a linear map $T \co E \to F$ is said to be completely $p$-summing \cite[p.~51]{Pis98} if it induces a bounded map $\Id_{S^p} \ot T \co S^p \ot_{\min} E \to S^p(F)$. In this case, the completely $p$-summing norm is defined by
\begin{equation}
\label{def-norrm-cb-p-summing}
\norm{T}_{\pi_{p}^\circ,E \to F}
\ov{\mathrm{def}}{=} \norm{\Id_{S^p} \ot T}_{S^p \ot_{\min} E \to S^p(F)}.
\end{equation}
The primary focus of the paper is to establish that there exists a simple expression for the completely $p$-summing norm of the adjoint $T^*$ when the quantum channel $T \co \L^1(\scr{C}(H)) \to \L^1(\scr{C}(H))$ is a radial multiplier, which allows us to compute the entanglement-assisted classical capacity $\C_{\EA}(T)$ with the formula \eqref{CEA-formula}.

Our main tools are the previous group action of the $n$-dimensional discrete hypercube $\{-1,1\}^n$ on the associated fermion algebra and Stein interpolation. Moreover, we will also use our  computation of the completely bounded norm $\norm{T}_{\cb, \L^1(\scr{C}(H)) \to \L^p(\scr{C}(H))}$ and some results of \cite{Arh24b}.
% similar thing happen for deciding if the Holevo capacity is bigger than a constant.

\paragraph{Structure of the Paper}
This paper is structured as follows. Section \ref{sec-preliminaries} provides the necessary background and revisits some notations. It also reviews key results that are necessary to our paper. In Section \ref{Sec-multipliers}, we introduce some notions of radials multipliers on fermion algebras and on discrete hypercubes. In Section \ref{sec-ergodic-action}, we describe an ergodic action of the discrete hypercube on the corresponding fermion algebra which is central in our approach. Section \ref{sec-completely-bounded-multipliers} is dedicated to the computation of some completely bounded norms of our multipliers.  We also characterize the complete positivity of the radial multipliers. In Section \ref{entropy-cb-radial-multipliers}, we calculate the completely bounded minimum output entropy. In Section \ref{summing-radial-multipliers}, we describe the completely $p$-summing norm of the radial multipliers. We deduce the computation of the entanglement-assisted classical capacity in Theorem \ref{th-capacity-assisted}, which is our main result. Finally, in Section \ref{sec-future}, we discuss other future developments.

%%%%%%%%%%%%%%%%%%%%%%%%%%%%%%%%%%%%%%%%%%%%%%%%%%%%%%%%%%%%%%%%%%%%%%%%%%%%%%%%%%%%%%%%%%%%%%%%%
\section{Preliminaries}
\label{sec-preliminaries}

\paragraph{Operator spaces} We refer to the books \cite{BLM04}, \cite{EfR00}, \cite{Pau02} and \cite{Pis2} for background on operators spaces and completely bounded maps. An operator space $E$ is a complex Banach space together with a matrix norm $\norm{\cdot}_{\M_n(E)}$ on each matrix space $\M_n(E)$, where $n \geq 1$, satisfying the conditions
$$
\norm{x \oplus y}_{\M_{k+l}(E)} 
= \max \big\{\norm{x}_{\M_k(E)},\norm{y}_{\M_l(E)} \big\}
\quad \text{and} \quad
\norm{\alpha y \beta}_{\M_k(E)} 
\leq \norm{\alpha}_{\M_{k,l}} \norm{y}_{\M_l(E)} \norm{\beta}_{\M_{l,k}}
$$
for any $x \in \M_k(E)$, $y \in \M_l(E)$, $\alpha \in \M_{k,l}$, $\beta \in \M_{l,k}$ and any integers $k,l \geq 1$. When $E$ and $F$ are two operator spaces, a linear map $T \co E \to F$ is completely bounded if
$$
\norm{T}_{\cb,E \to F}
\ov{\mathrm{def}}{=} \sup_{n\geq 1} \norm{\Id_{\M_n} \ot u}_{\M_n(E) \to \M_n(F)}
$$
is finite. We denote by $\CB(E,F)$ the space of all completely bounded maps. It is known that $\CB(E,F)$ is an operator space for the structure corresponding to the isometric identifications $\M_n(\CB(E,F))=\CB(E,\M_n(F))$ for each integer $n \geq 1$. The dual operator space of $E$ is $E^* \ov{\mathrm{def}}{=} \CB(E,\mathbb{C})$. If $E$ and $F$ are operator spaces then the adjoint map $\CB(E,F) \to \CB(F^*,E^*)$, $T \mapsto T^*$ is a complete isometry. Moreover, we have by \cite[Corollary 7.1.5 p.~128]{EfR00} a completely isometric isomorphism
\begin{equation}
\label{ot-proj-CB}
(E \otp F)^* 
=\CB(E,F^*),
\end{equation}
where $\otp$ denotes the operator space projective tensor product. If the operator spaces $E$ and $F$ are in addition \textit{finite-dimensional}, we have a completely isometric isomorphism
\begin{equation}
\label{ot-proj-dual}
(E \otp F)^*
=E^* \ot_{\min} F^*,
\end{equation}
where $\ot_{\min}$ denotes the minimal tensor product. For an operator space $E$, note that its opposite operator space $E^\op$ is defined by the same space $E$ but by the norms $\norm{\cdot}_{\M_n(E^\op)}$ defined by $\norm{[x_{ij}]}_{\M_n(E^\op)} \ov{\mathrm{def}}{=} \norm{[x_{ji}]}_{\M_n(E)}$ for any matrix $[x_{ij}]$ in $\M_n(E)$. If $\cal{A}$ is a $\mathrm{C}^*$-algebra then \cite[1.2.3 p.~5]{BLM04} the $*$-algebra $\M_n(\cal{A})$ has a unique norm with respect to which it is a $\mathrm{C}^*$-algebra. With respect to these matrix norms, $\cal{A}$ is an operator space. It is known that the matrix norms on $\cal{A}^\op$ coincide with the canonical matrix norms on the <<opposite>> $\mathrm{C}^*$-algebra, which is $\cal{A}$ endowed with its reversed multiplication.

%A dual operator space $E$ is said to have the dual slice mapping property \cite[p.~48]{BLM04} if 
%\begin{equation}
%\label{normal-minimal}
%E \otvn F 
%= \CB(F_*,E)
%\end{equation}
%holds for any dual operator space $F$, where $E \otvn F$ is the normal minimal tensor product, i.e.~the weak* closure of $E \ot_{\min} F$ in the dual space $\CB(F_*,E)$. By \cite[Theorem 11.2.5 p.~201]{EfR00}, this is equivalent to $E_*$ possessing the <<operator space approximation property>>.

%\begin{equation}
%\label{ot-min-CB}
%X \ot_{\min} Y 
%= \CB(Y^*, X)
%\end{equation}

\paragraph{Noncommutative $\L^p$-spaces} We direct readers to the survey \cite{PiX03} and references therein for information on noncommutative $\L^p$-spaces. Let $\cal{M}$ be a finite von Neumann algebra equipped with a normal finite faithful trace $\tau$. For any $1 \leq p < \infty$, the noncommutative $\L^p$-space $\L^p(\cal{M})$ is defined as the completion of the von Neumann algebra $\cal{M}$ for the norm
\begin{equation}
\label{norm-Lp}
\norm{x}_{\L^p(\cal{M})}
\ov{\mathrm{def}}{=} \tau(|x|^p)^{\frac{1}{p}}, \quad x \in \cal{M},
\end{equation}
where $|x| \ov{\mathrm{def}}{=} (x^*x)^{^{\frac{1}{2}}}$. We sometimes write $\L^\infty(\cal{M}) \ov{\mathrm{def}}{=} \cal{M}$. The trace $\tau$ induces a canonical continuous embedding $j \co \cal{M} \to \cal{M}_*$  of the von Neumann algebra $\cal{M}$ into its predual $\cal{M}_*$ defined by 
$$
\la j(x), y \ra_{\cal{M}_*,\cal{M}} 
\ov{\mathrm{def}}{=} \tau(xy), \quad x,y \in \cal{M}.
$$ 
Moreover, the previous map $j\co \cal{M} \to \cal{M}_*$ extends to a canonical isometry $j_1 \co \L^1(\cal{M})\to \cal{M}_*$, which allows us to identify the Banach space $\L^1(\cal{M})$ with the predual $\cal{M}_*$ of the von Neumann algebra $\cal{M}$.

With the embedding $j \co \cal{M} \to \cal{M}_*$, $(\cal{M},\cal{M}_*)$ is a compatible couple of Banach spaces in the context of complex interpolation theory \cite{BeL76}. For any $1 \leq p \leq \infty$, we have by \cite[p.~139]{Pis2} the isometric interpolation formula
\begin{equation}
\label{Lp-interpolation}
\L^p(\cal{M})
= (\cal{M},\cal{M}_*)_{\frac{1}{p}}.
\end{equation}
We equip the Banach space $\L^1(\cal{M})$ with the operator space structure induced by the one of $\cal{M}_*^\op$ and the isometric isomorphism $j_1 \co \L^1(\cal{M}) \to \cal{M}_*$. This means that $j_1 \co \L^1(\cal{M}) \to \cal{M}_*^\op$ is a complete isometry.

Suppose that $1 < p < \infty$. There exists a canonical operator space structure on each noncommutative $\L^p$-space $\L^p(\cal{M})$, that can be equally introduced using complex interpolation theory. Recall that this structure is defined by
\begin{equation*}
\label{}
\M_n(\L^p(\cal{M}))
\ov{\mathrm{def}}{=} (\M_n(\cal{M}),\M_n(\cal{M}_*^\op))_{\frac{1}{p}}, \quad n \geq 1,
\end{equation*}
where we use the opposed operator space structure on the predual $\cal{M}_*$. 

Finally, we recall H\"older's inequality. If $1 \leq p,q,r \leq \infty$ satisfy $\frac{1}{r}=\frac{1}{p}+\frac{1}{q}$ then
\begin{equation}
\label{Holder}
\norm{xy} _{\L^r(\cal{M})}
\leq \norm{x}_{\L^p(\cal{M})} \norm{y}_{\L^q(\cal{M})},\qquad x\in \L^p(\cal{M}), y \in \L^q(\cal{M}).
\end{equation}
Conversely, for any $z \in \L^r(\cal{M})$, there exist $x \in \L^p(\cal{M})$ and $y \in \L^q(\cal{M})$ such that 
\begin{equation}
\label{inverse-Holder}
z
=xy
\quad \text{with} \quad 
\norm{z} _{\L^r(\cal{M})}
=\norm{x}_{\L^p(\cal{M})} \norm{y}_{\L^q(\cal{M})}.
\end{equation}

\paragraph{Conditional expectations} We refer to the book \cite[Chapter II]{Str81} for more information. Recall that a positive map $T \co A \to A$ on a $\mathrm{C}^*$-algebra $A$ is said faithful if $T(x)=0$ for some positive element $x \in A$ implies $x=0$. Let $B$ be a $\mathrm{C}^*$-subalgebra of a $\mathrm{C}^*$-algebra $A$. A linear map $\E \co A \to A$ is called a conditional expectation on $B$ if it is a positive projection \cite[p.~116]{Str81} of range $B$ which is $B$-bimodular, that is 
\begin{equation}
\label{bimodular}
\E(xyz)
=x\E(y)z, \quad  y \in A, x,z \in B.
\end{equation}
Such a map is completely positive \cite[Proposition p.~118]{Str81}. The following existence result is \cite[Corollary p.~134]{Str81}.

\begin{prop}
\label{prop-existence-conditional-expectation}
Let $\cal{M}$ be a von Neumann algebra equipped with a normal finite faithful trace $\tau$. For any unital von Neumann subalgebra $\cal{N}$ of $\cal{M}$, there exists a unique faithful normal conditional expectation $\E \co \cal{M} \to \cal{M}$ on $\cal{N}$ such that $\tau \circ \E = \tau$.
\end{prop}
In this situation, it is well-known \cite[Theorem 7 p.~151]{Kad04} that 
\begin{equation}
\label{Esp-dual}
\tau(x\E(y))
=\tau(xy), \quad x \in \cal{N}, y \in \cal{M}.
\end{equation}

\paragraph{Vector-valued noncommutative $\L^p$-spaces} 
Pisier pioneered the study of vector-valued noncommutative $\L^p$-spaces, focusing on the case where the underlying von Neumann algebra $\cal{M}$ is approximately finite-dimensional and comes with a normal semifinite faithful trace $\tau$, as outlined in \cite{Pis98}. Recall that a von Neumann algebra $\cal{M}$ is approximately finite-dimensional if it admits an increasing net of finite-dimensional $*$-subalgebras whose union is weak* dense. Within this framework, Pisier demonstrated that for any given operator space $E$, the Banach spaces underlying the operator spaces
\begin{equation}
\label{def-Lp-vec-1}
\L^\infty_0(\cal{M},E) \ov{\mathrm{def}}{=}\cal{M} \ot_{\min} 
E
\quad \text{and} \quad
\L^1(\cal{M},E) \ov{\mathrm{def}}{=} \cal{M}_*^{\op} \widehat{\ot} E
\end{equation}
can be embedded into a common topological vector space, as detailed in \cite[pp.~37-38]{Pis98}. 
This result, relying crucially on the approximate finite-dimensionality of the von Neumann algebra $\cal{M}$, allows everyone to use complex interpolation theory in order to define vector-valued noncommutative $\L^p$-spaces by letting
$$
\L^p(\cal{M},E)
\ov{\mathrm{def}}{=} (\L^\infty_0(\cal{M},E),\L^1(\cal{M},E))_{\frac{1}{p}}
$$
for any $1 < p < \infty$. The notation $\L^\infty(\cal{M},E)$ for the space $\L^\infty_0(\cal{M},E)$ is a convenient notation sometimes used to simplify the statement of some results. 

\begin{example} \normalfont
If $\cal{M}$ is the von Neumann algebra $\B(\ell^2)$ equipped with its canonical normal semifinite faithful trace, the space $\L^p(\cal{M},E)$ is the vector-valued Schatten space $S^p(E)$.
\end{example}

Let $z \in \cal{M} \ot E$. There exist elements $(c_1,\ldots,c_n)$ in $\cal{M}$ and elements $(x_1,\ldots,x_n)$ in $E$ such that $y=\sum_{k=1}^{n} c_k \ot x_k$. Then for any $a \in \L^{2p}(\cal{M})$ and any $b \in \L^{2p}(\cal{M})$, we will write $a \cdot z\cdot b$ for the element belonging to the space $\L^p(\cal{M}) \ot E$ defined by
\begin{equation}
\label{product-cdot}
a \cdot z \cdot b
\ov{\mathrm{def}}{=}\sum_{k=1}^n a c_k b \ot x_k.
\end{equation}
By \cite[Theorem 3.8 p.~43]{Pis98}, if the trace $\tau$ is finite then the norm of the Banach space $\L^p(\cal{M},E)$ for any $1 \leq p \leq \infty$ admits the following description. For any $x \in \cal{M} \ot E$, we have the variational formula
\begin{equation}
\label{LpME-inf}
\norm{x}_{\L^p(\cal{M},E)}
=\inf_{x=a\cdot z \cdot b} \norm{a}_{\L^{2p}(\cal{M})}\norm{z}_{\cal{M} \ot_{\min} E}\norm{b}_{\L^{2p}(\cal{M})}.
\end{equation}
Moreover, the subspace $\cal{M} \ot E$ is dense in the Banach space $\L^p(\cal{M},E)$ if the trace $\tau$ is finite and equal to the space $\L^p(\cal{M},E)$ if the von Neumann algebra $\cal{M}$ is finite-dimensional.

By \cite[Theorem 4.1 p.~46]{Pis98}, for any $1 < p < \infty$ we have a completely isometric inclusion
\begin{equation}
\label{inclusion-dual}
\L^{p^*}(\cal{M},E^*) 
\subset \L^{p}(\cal{M},E)^*,
\end{equation}
where $p^* \ov{\mathrm{def}}{=} \frac{p}{p-1}$. Moreover, if $\cal{N}$ is another approximately finite-dimensional von Neumann algebra equipped with a normal semifinite faithful trace, recall Fubini's theorem \cite[(3.6) p.~40]{Pis98}, which states that for any $1 \leq p<\infty$, there exists a canonical completely isometric isomorphism
\begin{equation}
\label{Fubini}
\L^p(\cal{M},\L^p(\cal{N}))
=\L^p(\cal{M} \otvn \cal{N}),
\end{equation}
where the von Neumann algebra $\cal{M} \otvn \cal{N}$ is endowed with the product of traces.

We need the folklore formula \eqref{LpME-sup} which is stated without proof in \cite[p.~3]{JuP10} for approximately finite-dimensional von Neumann algebras. We begin with a straightforward intermediate result which is a variant of H\"older's inequality.

\begin{prop}
\label{prop-ineg}
%Soient $1 \leq p \leq \infty$, $n$ un entier strictement positif et $E$ un espace d'opérateurs. Si $ x \in \L^{p}(M_{n},E) $ et $a,b \in M_{n}$ alors
%$$
%\norm{a \cdot x \cdot b}_{\L^{p}(\cal{M},E)} 
%\leq \norm{a}_{\L^{\infty}(M_{n})}\norm{x}_{\L^{p}(M_{n},E)}\norm{b}_{\L^{\infty}(M_{n})}
%$$
%Plus généralement, si  
Let $E$ be an operator space. Suppose that $\cal{M}$ is a finite-dimensional von Neumann algebra equipped with a finite faithful trace. If $1 \leq r \leq p \leq \infty$ and $1 \leq q \leq \infty$ satisfy $\frac{1}{r} = \frac{1}{p} + \frac{1}{q}$ then for any $x \in \L^{p}(\cal{M},E)$ and any $a,b \in \L^{2q}(\cal{M})$ we have
\begin{equation}
\label{ine-vector}
\norm{a \cdot x \cdot b}_{\L^{r}(\cal{M},E)}
\leq
\norm{a}_{\L^{2q}(\cal{M})} \norm{x}_{\L^{p}(\cal{M},E)} \norm{b}_{\L^{2q}(\cal{M})}.
\end{equation}
\end{prop}

\begin{proof}
%First, we suppose that $p<\infty$. 
Let $x = c \cdot v \cdot d$ be a factorization of $x$ with $c \in \L^{2p}(\cal{M})$, $v \in \cal{M} \ot_{\min} E$ and $d \in \L^{2p}(\cal{M})$. We have $a \cdot x \cdot b \ov{\eqref{product-cdot}}{=} (a c) \cdot v \cdot (d  b)$. So we have a factorization of the element $a \cdot x \cdot b$ with $a c \in \L^{2r}(\cal{M})$, $v \in \cal{M} \ot_{\min} E$ and $d b \in \L^{2r}(\cal{M})$ since we have $\frac{1}{2r} = \frac{1}{2p} + \frac{1}{2q}$. With H\"older's inequality, we deduce that
\begin{align*}
\MoveEqLeft
\norm{a \cdot x \cdot b}_{\L^{r}(\cal{M},E)} 
\ov{\eqref{LpME-inf}}{\leq} 
\norm{ac}_{\L^{2r}(\cal{M})} \norm{v}_{\cal{M} \ot_{\min} E} \norm{db}_{\L^{2r}(\cal{M})} \\
&\ov{\eqref{Holder}}{\leq}  \norm{a}_{\L^{2q}(\cal{M})} \norm{c}_{\L^{2p}(\cal{M})} \norm{v}_{\cal{M} \ot_{\min} E}
\norm{d}_{\L^{2p}(\cal{M})} \norm{b}_{\L^{2q}(\cal{M})}.
\end{align*}
Taking the infimum, we obtain \eqref{ine-vector}.
%Now, if $p=\infty$ and $r < \infty$, we have $q=r$. In this case, we have
%$$
%\norm{a \cdot x \cdot b}_{\L^{r}(\cal{M},E)} 
%\ov{\eqref{LpME-inf}}{\leq} \norm{a}_{\L^{2r}(\cal{M})}\norm{x}_{_{\cal{M} \ot_{\min} E}} \norm{b}_{\L^{2r}(\cal{M})}.
%$$
%The case $p=r=\infty$ is not more difficult since the map $\cal{M} \to \cal{M}$, $x\mapsto axb$ is completely bounded.
\end{proof}

Now, we prove the desired formula.

\begin{prop}
\label{prop-facto-sup}
Let $\cal{M}$ and $\cal{N}$ be two finite-dimensional von Neumann algebras equipped with finite faithful traces. Suppose that $1 \leq q \leq p \leq \infty$ and $1 \leq r \leq \infty$ such that $\frac{1}{q}=\frac{1}{p}+\frac{1}{r}$. For any $y \in \L^{p}(\cal{M},\L^{q}(\cal{N}))$, we have
\begin{equation}
\label{LpME-sup}
\norm{y}_{\L^{p}(\cal{M},\L^{q}(\cal{N}))}
= \sup_{\norm{a}_{\L^{2r}(\cal{M})},\norm{b}_{\L^{2r}(\cal{M})} \leq 1} \norm{(a \ot 1) y (b\ot 1)}_{\L^{q}(\cal{M},\L^{q}(\cal{N}))}.
\end{equation}
\end{prop}

\begin{proof}
If $a,b$ are elements of $\L^{2r}(\cal{M})$ satisfying $\norm{a}_{\L^{2r}(\cal{M})}$ and $\norm{b}_{\L^{2r}(\cal{M})} \leq 1$, we have
\begin{align*}
\MoveEqLeft
\norm{(a\ot 1) y (b \ot 1)}_{\L^{q}(\cal{M},\L^{q}(\cal{N}))}
 \ov{\eqref{ine-vector}}{\leq}   \norm{a}_{\L^{2r}(\cal{M})} \norm{y}_{\L^{p}(\cal{M},\L^{q}(\cal{N}))} \norm{b}_{\L^{2r}(\cal{M})} 
\leq \norm{y}_{\L^{p}(\cal{M},\L^{q}(\cal{N}))}. 
\end{align*}
Taking the supremum, we infer that
$$
\sup_{\norm{a}_{\L^{2r}(\cal{M})},\norm{b}_{\L^{2r}(\cal{M})} \leq 1}  \norm{(a \ot 1) y (b \ot 1)}_{\L^{q}(\cal{M},\L^{q}(\cal{N}))} 
\leq \norm{y}_{\L^{p}(\cal{M},\L^{q}(\cal{N}))}.
$$
Now, we will show the reverse inequality. Since the spaces are finite-dimensional, the inclusion of \eqref{inclusion-dual} for $E=\L^{q}(\cal{N})$ becomes the isometric identity for any $1 \leq p,q \leq \infty$
\begin{equation}
\label{equality-inter-9}
\L^{p}(\cal{M},\L^{q}(\cal{N}))^*  
= \L^{p^*}(\cal{M},\L^{q^*}(\cal{N})).
\end{equation}
(the case $p=1$ and $p=\infty$ are consequences of \eqref{ot-proj-dual} and \eqref{def-Lp-vec-1}). In particular, we have
$$
\L^{1}(\cal{M},\L^{q}(\cal{N}))^*
%\ov{\eqref{def-Lp-vec-1}}{=} (\cal{M}_*^\op \otp \L^{q}(\cal{N}))^*
\ov{\eqref{def-Lp-vec-1}}{=} \cal{M} \ot_{\min} \L^{q^*}(\cal{N}).
$$
%In all cases, we have
%$$
%(\L^{p}(\cal{M},\L^{q}(\cal{N}))^*
%=\L^{p^*}(\cal{M},\L^{q^*}(\cal{N}) ).
%$$
Let $y \in \L^{p}(\cal{M},\L^{q}(\cal{N}))$. By duality, there exists $x \in \L^{p^*}(\cal{M},\L^{q^*}(\cal{N}) )$ such that $\norm{y}_{\L^{p}(\cal{M},\L^{q}(\cal{N}))} = |\la x , y \ra|$ with $ \norm{x}_{\L^{p^*}(\cal{M},\L^{q^*}(\cal{N}) )} \leq 1$. Let $\epsi >0$. According to \eqref{LpME-inf}, we can write $x=(a \ot 1) z (b \ot 1)$ with $a , b \in B_{\L^{2p^*}(\cal{M})}$, $z \in \cal{M} \ot_{\min} \L^{q^*}(\cal{N})$ such that 
\begin{equation}
\label{div-120}
\norm{z}_{\cal{M} \ot_{\min} \L^{q^*}(\cal{N})} 
\leq \norm{x}_{\L^{p^*}(\cal{M},\L^{q^*}(\cal{N}) )}+\epsi 
\leq 1+\epsi.
\end{equation}
We obtain
\begin{align}
\MoveEqLeft
\label{inter-765}
\norm{y}_{\L^{p}(\cal{M},\L^{q}(\cal{N}))} 
= |\la x , y \ra|
= \big|\big\la (a \ot 1) z (b \ot 1) , y \big\ra_{\L^{p^*}(\cal{M},\L^{q^*}(\cal{N}) ),\L^{p}(\cal{M},\L^{q}(\cal{N}))}\big| \nonumber\\
%&= \sup \Big\{  |\la  v  , a^{T} \cdot y \cdot b^{T} \ra| : a , b \in B_{\L^{2p^*}(\cal{M})}, v \in B_{\cal{M} \ot_{\min} E^{*}}  \Big\} \\
&=  \big|\big\la  z  , (a \ot 1) y (b \ot 1) \big\ra_{\cal{M} \ot_{\min} \L^{q^*}(\cal{N}),\L^{1}(\cal{M},\L^{q}(\cal{N}))} \big|  \nonumber\\
&\leq \norm{z}_{\cal{M} \ot_{\min} \L^{q^*}(\cal{N})} \norm{(a \ot 1) y (b \ot 1)}_{\L^{1}(\cal{M},\L^{q}(\cal{N}))}\nonumber \\
&\ov{\eqref{div-120}}{\leq} (1+\epsi)  \norm{(a \ot 1) y (b \ot 1)}_{\L^{1}(\cal{M},\L^{q}(\cal{N}))}. \nonumber
\end{align}
%Consequently, we obtain by duality
%\begin{align}
%\MoveEqLeft
%\label{inter-765}
%\norm{y}_{\L^{p}(\cal{M},\L^{q}(\cal{N}))} 
%= \sup_{\norm{x}_{\L^{p^*}(\cal{M},\L^{q^*}(\cal{N}) )} \leq 1}   |\la x , y \ra| \\
%&\ov{\eqref{LpME-inf}}{=} \sup \Big\{ \big|\big\la (a \ot 1) v (b \ot 1) , y \big\ra\big| : a , b \in B_{\L^{2p^*}(\cal{M})},\  v \in B_{\cal{M} \ot_{\min} \L^{q^*}(\cal{N})}  \Big\} \nonumber\\
%%&= \sup \Big\{  |\la  v  , a^{T} \cdot y \cdot b^{T} \ra| : a , b \in B_{\L^{2p^*}(\cal{M})}, v \in B_{\cal{M} \ot_{\min} E^{*}}  \Big\} \\
%&= \sup \Big\{  \big|\big\la  v  , (a \ot 1) y (b \ot 1) \big\ra\big| : a , b \in B_{\L^{2p^*}(\cal{M})}, v \in B_{\cal{M} \ot_{\min} \L^{q^*}(\cal{N})}  \Big\} \nonumber\\
%&= \sup_{\norm{a}_{\L^{2p^*}(\cal{M})},\norm{b}_{\L^{2p^*}(\cal{M})} \leq 1}  \norm{(a \ot 1) y (b \ot 1)}_{\L^{1}(\cal{M},\L^{q}(\cal{N}))}. \nonumber
%\end{align}
Since $\frac{1}{p^*}=\frac{1}{q^*}+\frac{1}{r}$, with \eqref{inverse-Holder} we can write $a=a'a''$ et $b=b''b'$ with $a',b' \in B_{\L^{2q^{*}}(\cal{M})}$ et $a'',b'' \in B_{\L^{2r}(\cal{M})}$. We obtain
\begin{align*}
\MoveEqLeft
\norm{y}_{\L^{p}(\cal{M},\L^{q}(\cal{N}))}
%\ov{\eqref{inter-765}}{=} \sup_{\norm{a}_{\L^{2p^*}(\cal{M})},\norm{b}_{\L^{2p^*}(\cal{M})} \leq 1} \norm{(a \ot 1) y (b \ot 1)}_{\L^{1}(\cal{M},\L^{q}(\cal{N}))}   \\
\leq  (1+\epsi)\norm{(a'a'' \ot 1) y (b''b' \ot 1)}_{\L^{1}(\cal{M},\L^{q}(\cal{N}))}  \\
%&  &  \text{car $\frac{1}{r}=\frac{1}{p}-\frac{1}{q}$}\\
&\ov{\eqref{ine-vector}}{\leq}  (1+\epsi) \norm{a'}_{\L^{2q^*}(\cal{M})}  \norm{(a'' \ot 1) y (b'' \ot 1)}_{\L^{q}(\cal{M},\L^{q}(\cal{N}))} \norm{b'}_{\L^{2q^*}(\cal{M})} \\
&\leq (1+\epsi)   \norm{(a'' \ot 1) y (b'' \ot 1)}_{\L^{q}(\cal{M},\L^{q}(\cal{N}))}.
\end{align*} 
The conclusion is obvious.
\end{proof}

In particular, for any $1 \leq q \leq \infty$, we have
\begin{equation}
\label{LinftyLp-norms}
\norm{x}_{\L^\infty(\cal{M},\L^q(\cal{N}))}
=\sup_{\norm{a}_{\L^{2q}(\cal{M})},\norm{b}_{\L^{2q}(\cal{M})} \leq 1} \bnorm{(a \ot 1) x (b \ot 1)}_{\L^q(\cal{M} \otvn \cal{N})}.
\end{equation}
We will use the following result which is proved in a more general framework in \cite{Arh24b}. In this result, the used positive cone of the space $\L^p(\cal{M},\L^1(\cal{N}))$ is the same that the one of the algebra $\cal{M} \ot \cal{N}$.

\begin{lemma}%normal semifinite
Let $\cal{M}$ and $\cal{N}$ be finite-dimensional von Neumann algebras equipped with finite faithful traces. Suppose that $1 \leq p \leq \infty$. For any positive element $x$ in the space $\L^p(\cal{M},\L^1(\cal{N}))$, we have
\begin{equation}
\label{Norm-LpL1-pos}
\norm{x}_{\L^p(\cal{M},\L^1(\cal{N}))}
=\norm{(\Id \ot \tau)(x)}_{\L^p(\cal{M})},
\end{equation}
where $\tau$ is the trace of $\cal{N}$.
\end{lemma}

We finish with a useful variational formula which is written without proof in \cite[Lemma 3.11 p.~3434]{GJL20} in a slightly more general context. For the sake of completeness, we give the proof.

\begin{prop}
\label{prop-Gao}
Let $\cal{M}$ and $\cal{N}$ be finite-dimensional von Neumann algebras equipped with finite faithful traces $\tau_{\cal{M}}$ and $\tau_{\cal{N}}$. If $\E \ov{\mathrm{def}}{=} \Id \ot \tau_{\cal{N}} \co \cal{M} \otvn \cal{N} \to \cal{M}$ is the canonical trace preserving faithful conditional expectation, we have for any element $x$ in the space $\L^\infty(\cal{M},\L^1(\cal{N}))$
\begin{equation}
\label{formula-Junge}
\norm{x}_{\L^\infty(\cal{M},\L^1(\cal{N}))}
=\inf_{x=yz} \bnorm{\E(yy^*)}_{\L^\infty(\cal{M})}^{\frac{1}{2}}\bnorm{\E(z^*z)}_{\L^\infty(\cal{M})}^{\frac{1}{2}}.
\end{equation}
\end{prop}

\begin{proof}
Let $x \in \L^\infty(\cal{M},\L^1(\cal{N}))$. Suppose that $x=y z$. Then for any elements $a,b \in \L^{2}(\cal{M})$ with $\norm{a}_{\L^{2}(\cal{M})} \leq 1$ and $\norm{b}_{\L^{2}(\cal{M})} \leq 1$, we have
\begin{align*}
\MoveEqLeft
\bnorm{(a \ot 1) x (b \ot 1)}_{\L^1(\cal{M} \otvn \cal{N})}         
=\bnorm{(a \ot 1) yz (b \ot 1)}_{\L^1(\cal{M} \otvn \cal{N})} \\
&\ov{\eqref{Holder}}{\leq} \bnorm{(a \ot 1) y}_{\L^2(\cal{M} \otvn \cal{N})} \bnorm{z (b \ot 1)}_{\L^2(\cal{M} \otvn \cal{N})} \\
&= \bnorm{[(a \ot 1) y]^*}_{\L^2(\cal{M} \otvn \cal{N})} \bnorm{z (b \ot 1)}_{\L^2(\cal{M} \otvn \cal{N})} \\
&\ov{\eqref{norm-Lp}}{=} (\tau_{\cal{M}} \ot \tau_{\cal{N}}) \big[(a \ot 1)(yy^*)(a^* \ot 1)\big]^{\frac{1}{2}} \, (\tau_{\cal{M}} \ot \tau_{\cal{N}})\big[(b^* \ot 1)(z^*z)(b \ot 1)\big]^{\frac{1}{2}} \\
&= \tau_{\cal{M}}\big[\E((a \ot 1)(yy^*)(a^* \ot 1))\big]^{\frac{1}{2}} \tau_{\cal{M}}\big[\E((b^* \ot 1)(z^*z)(b \ot 1))\big]^{\frac{1}{2}}.
\end{align*}
Using the property of the conditional expectation described in \eqref{bimodular}, we obtain that
\begin{align*}
\MoveEqLeft
\bnorm{(a \ot 1) x (b \ot 1)}_{\L^1(\cal{M} \otvn \cal{N})}  
\leq \tau_{\cal{M}}\big[a\E(yy^*) a^*\big]^{\frac{1}{2}} \tau\big[b^*\E(z^*z)b\big]^{\frac{1}{2}} \\
&\ov{\eqref{norm-Lp}}{=}  \bnorm{a\E(yy^*)a^*}_{\L^1(\cal{M})}^{\frac{1}{2}} \bnorm{b^* \E(z^*z)b}_{\L^1(\cal{M})}^{\frac{1}{2}} \\
&\ov{\eqref{Holder}}{\leq} \norm{a}_{\L^2(\cal{M})}^{\frac{1}{2}} \bnorm{\E(yy^*)}_{\L^\infty(\cal{M})}^{\frac{1}{2}} \bnorm{a^*}_{\L^2(\cal{M})}^{\frac{1}{2}}\bnorm{b^*}_{\L^2(\cal{M})}^{\frac{1}{2}}\bnorm{\E(z^*z)}_{\L^\infty(\cal{M})}^{\frac{1}{2}}\bnorm{b}_{\L^2(\cal{M})}^{\frac{1}{2}} \\ 
&\leq \bnorm{\E(yy^*)}_{\L^\infty(\cal{M})}^{\frac{1}{2}} \bnorm{\E(z^*z)}_{\L^\infty(\cal{M})}^{\frac{1}{2}}.
\end{align*}
Taking the infimum, we obtain with the case $q=1$ of \eqref{LinftyLp-norms} that the left-hand member of the equality \eqref{formula-Junge} is less than the right-hand member of \eqref{formula-Junge}. The converse inequality is \cite[Lemma 4.9 p.~76]{JuP10}. 
\end{proof}

\paragraph{Tensorizations of linear maps} 
Let $E$ and $F$ be operator spaces. Suppose $1 \leq p \leq \infty$. By \cite[Corollary 1.2 p.~19]{Pis98}, a linear map $T \co E \to F$ is completely bounded if and only if the linear map $\Id_{S^p} \ot T$ extends to a bounded operator $\Id_{S^p} \ot T \co S^p(E) \to S^p(F)$. In this case, the completely bounded norm $\norm{T}_{\cb,E \to F}$ is given by
\begin{equation}
\label{defnormecb}
\norm{T}_{\cb,E \to F}
=\norm{\Id_{S^p} \ot T}_{S^p(E) \to S^p(F)}.
\end{equation}
Let $\cal{M}$ be an approximately finite-dimensional von Neumann algebra equipped with a normal semifinite faithful trace. If a linear map $T \co E \to F$ is completely bounded then by \cite[(3.1) p.~39]{Pis98} it induces a completely bounded map $\Id_{\L^p(\cal{M})} \ot T \co \L^p(\cal{M},E) \to \L^p(\cal{M},F)$ and we have 
\begin{equation}
\label{ine-tensorisation-os}
\norm{\Id_{\L^p(\cal{M})} \ot T}_{\cb, \L^p(\cal{M},E) \to \L^p(\cal{M},F)} 
\leq \norm{T}_{\cb,E \to F}.
\end{equation}

\paragraph{Entropy}
We refer to \cite{LoW22}, \cite{LuP17}, \cite{LuP19}, \cite{LPS17}, \cite{NaU61}, \cite{OcS78} and \cite{Rus73} for more information on Segal entropy. Let $g \co \R^+ \to \R$ be the function defined by $g(0) \ov{\mathrm{def}}{=} 0$ and $g(t) \ov{\mathrm{def}}{=} t\log_2(t)$ if $t > 0$. Let $\cal{M}$ be a von Neumann algebra equipped with a normal faithful \textit{finite} trace $\tau$. If $x$ is a positive element in the Banach space $\L^1(\cal{M})$, by a slight abuse, we use the notation $x\log_2 x \ov{\mathrm{def}}{=} g(x)$. In this case, if in addition $\norm{x}_{\L^1(\cal{M})}=1$ we can define the Segal entropy 
\begin{equation}
\label{Segal-entropy}
\H(x) \ov{\mathrm{def}}{=} -\tau(x \log_2 x),
\end{equation}
introduced in \cite{Seg60}. If we need to specify the trace $\tau$, we will use the notation $\H_\tau(x)$. Note that if the trace $\tau$ is \textit{normalized} and if $x$ is a positive element in the space $\L^1(\cal{M})$ with $\tau(x)=1$ then by \cite[Proposition 2.2]{LoW22} the entropy $\H(x)$ belongs to $[-\infty,0]$ with
\begin{equation}
\label{carac-entropy-max}
\H(x)=0 
\text{ if and only if } x=1.
\end{equation}
Recall that by \cite[p.~78]{OcS78}, if $x_1$ and $x_2$ are positive elements in the Banach space $\L^1(\cal{M})$ with norm 1, we have the additivity formula
\begin{equation}
\label{entropy-tensor-product}
\H(x_1 \ot x_2)
= \H(x_1)+\H(x_2).
\end{equation}
In this paper, we may need to use various traces (though proportional) on the same algebra. While entropy is dependent on the trace, there exists a straightforward method to switch from one to another. If $\tau$ is again the normalized trace, recall that for any $t > 0$ and any positive elements $x \in \L^1(\cal{M},\tau)$ and $y \in \L^1(\cal{M},t\tau)$ with norms equal to 1, we have
\begin{equation}
\label{changement-de-trace}
\H_\tau(x)
=\H_{t\tau}\bigg(\frac{1}{t}x\bigg)-\log_2 t
\quad \text{and} \quad
\H_{\tau}(t y)
=\H_{t\tau}(y)-\log_2 t.
%\quad 
%\H(\rho,\tau)
%=\H\bigg(\frac{1}{t}\rho,\tr\bigg)-\log t.
\end{equation}
%\begin{lemma}
%\label{Lemma-useful}
%Suppose $1 \leq p \leq \infty$. Let $\QG$ be a compact quantum group of Kac type. The coproduct $\Delta$ map induces an isometry $\Delta \co \L^p(\QG) \to \L^\infty(\QG,\L^p(\QG))$ and a bounded map $\Delta \co \L^1(\QG) \to \L^p(\QG,\L^1(\QG))$ with norm less than $\tau(1)^{\frac{1}{p}}$ where $\tau$ is a Haar trace.
%\end{lemma}
We finish with the following well-known fundamental result, which connects the noncommutative $\L^p$-norms with the Segal entropy, see e.g.~\cite[Remark 9.5]{Arh24b} or \cite[Theorem 3.4 p.~371]{JuP15} (see also \cite[p.~625]{Seg60}).

\begin{prop}
\label{prop-deriv-norm-p}
Let $\cal{M}$ be a finite-dimensional von Neumann algebra equipped with a finite faithful trace $\tau$. For any positive element $x \in \cal{M}$ with $\norm{x}_{\L^1(\cal{M})}=1$, we have
\begin{equation}
\label{deriv-norm-p}
\frac{1}{\log 2}\frac{\d}{\d p} \norm{x}_{\L^p(\cal{M})}|_{p=1}
=\tau(x \log_2 x)
\ov{\eqref{Segal-entropy}}{=} -\H(x).
\end{equation}
%Moreover, the monotonically decreasing family of functions $x \mapsto \frac{1-\norm{x}_{\L^p(\cal{M})}^p}{p-1}$ converge uniformly to the function $x \mapsto -\tau(x \log x)=\H(x)$ when $p \to 1^{+}$ on the subset of positive element $x \in \cal{M}$ with $\norm{x}_1=1$.
\end{prop}

%\begin{equation}
%\label{}
%\H_{\cb,\min,\tau}(T)
%=\H_{\cb,\min,\tr}(T) - \log_2(n), \quad \H_{\cb,\min,\tau}(T) +\log_2(n)
%=\H_{\cb,\min,\tr}(T) 
%\end{equation}

\paragraph{Harmonic analysis on locally compact abelian groups} We refer to the books \cite{Bou19} and \cite{Fre13} for more information on locally compact abelian groups. Let $G$ be a locally compact abelian group endowed with a Haar measure $\mu$. Recall that the Pontryagin dual $\hat{G}$ of $G$ is the locally compact abelian group consisting of continuous group homomorphisms from the group $G$ into the circle group $\T$, with pointwise multiplication as the group operation and the topology of uniform convergence on compact sets. The Pontryagin duality theorem \cite[Th\'eor\`eme 2, II.220]{Bou19} says that any locally compact abelian group is canonically isomorphic with its bidual (the Pontryagin dual of its Pontryagin dual). More precisely the map $G \to \doublehat{G}$, $s \mapsto (\chi \mapsto \chi(s))$ is a topological isomorphism from $G$ onto $\doublehat{G}$. 

Recall that the convolution of a bounded Borel regular complex measure $\nu$ on $G$ and a locally $\mu$-integrable function $g \co G \to \mathbb{C}$, which are convolvable, is defined by
\begin{equation}
\label{conv-mes-funct}
(\nu *g)(t)
\ov{\mathrm{def}}{=} \int_G g(s^{-1}t) \d\nu(s), \quad \text{a.e. } t \in G.
\end{equation}
Consider two locally $\mu$-integrable functions $f,g \co G \to \mathbb{C}$. If the measure $f\cdot\mu$ and the function $g$ are convolvable then we say that $f$ and $g$ are convolvable and we can introduce the convolution product
\begin{equation}
\label{Convolution-formulas}
(f*g)(t)
\ov{\mathrm{def}}{=} \int_G f(s)g(s^{-1}t) \d\mu(s)
=\int_G f(ts^{-1})g(s) \d\mu(s),\quad \text{a.e. } t \in G.
\end{equation}
In this case, we have the equality $(f\cdot\mu)*g= f*g$ almost everywhere.

\begin{example} \normalfont
\label{Example-finite-group}
Let $G$ be a finite abelian group equipped with its Haar probability measure $\mu$. The measure $\mu$ is atomic and we have $\mu(\{s\})=\frac{1}{|G|}$ for any $s \in G$, where $|G|$ is the cardinal of $G$. Let $f,g \co G \to \mathbb{C}$ be some functions. Note that $\norm{f}_{\L^1(G)}=\frac{1}{|G|} \sum_{s \in G} f(s)$. For any $t \in G$, we have
\begin{equation}
\label{convol-finite-group-1}
(f*g)(t)
\ov{\eqref{Convolution-formulas}}{=} %\int_G f(s)g(s^{-1}t) \d\mu(s)
=\frac{1}{|G|} \sum_{s \in G} f(s)g(s^{-1}t).
\end{equation}
Consider the measure $\nu$ on the group $G$ defined by $\nu(\{s\})=f(s)$ for any $s \in G$. We warn the reader that
\begin{equation}
\label{convol-finite-group-2}
(\nu*g)(t)
\ov{\eqref{conv-mes-funct}}{=} \sum_{s \in G} g(s^{-1}t)f(s),\quad t \in G,
\end{equation}
which is slightly different from \eqref{convol-finite-group-1}.
\end{example}

The Fourier transform of a bounded Borel regular complex measure $\mu$ on $G$ is defined in \cite[D\'efinition 3, II.206]{Bou19} or in \cite[p.~3338]{Fre13} (without the conjugate bar) by
$$
\hat{\mu}(\chi)
=\int_G \ovl{\la \chi,s\ra_{\hat{G},G}} \d\mu(s), \quad \chi \in \hat{G}.
$$
Recall that a $\mu_G$-integrable function $f \co G \to \mathbb{C}$ defines a bounded Borel regular complex measure $f \cdot \mu$. The Fourier transform of such a function is defined in \cite[(14), II.208]{Bou19} or in \cite[p.~338]{Fre13} (without the conjugate bar) by
\begin{equation}
\label{Fourier-transform-I}
\hat{f}(\chi)
=\int_G \ovl{\la \chi,s\ra_{\hat{G},G}}\, f(s)\d\mu(s), \quad \chi \in \hat{G}.
\end{equation}
We have $\widehat{f \cdot \mu}=\hat{f}$. For any $\mu_G$-integrable function $f,g \co G \to \mathbb{C}$, we have $\widehat{f*g}=\hat{f}\hat{g}$ by \cite[Proposition 445G p.~339]{Fre13}. Recall that a function $f \co G \to \mathbb{C}$ is positive definite if $\sum_{j,k=1}^{n} \lambda_j \ovl{\lambda_k} f(s_k^{-1}s_j) \geq 0$ for any $ \lambda_1,\ldots,\lambda_n \in \mathbb{C}$ and any $s_1,\ldots,s_n \in G$. By \cite[Theorem 445P p.~349]{Fre13}, there exists a unique measure $\hat{\mu}$ on the Pontryagin dual $\hat{G}$ such that whenever $f \co G \to \mathbb{C}$ is continuous positive-definite $\mu_G$-integrable function then the Fourier transform $\hat{f} \co \hat{G} \to \mathbb{C}$ is $\hat{\mu}$-integrable on $\hat{G}$ and we have the <<Fourier inversion formula>>
\begin{equation}
\label{Fourier-inversion-formula}
f(s)
=\int_{\hat{G}} \la \chi,s\ra_{\hat{G},G} \hat{f}(\chi) \d \hat{\mu}(\chi), \quad s \in G.
\end{equation}
The measure $\hat{\mu}$ is sometimes called the dual measure of $\mu$.

\begin{example} \normalfont
\label{dual-measure-compact-group}
If $G$ is a compact abelian group equipped with its Haar probability measure $\mu$ then by \cite[Proposition 18 II.233]{Bou19} the Pontryagin dual $\hat{G}$ is discrete and the dual measure $\hat{\mu}$ on the Pontryagin dual $\hat{G}$ is the counting measure.
\end{example}

\paragraph{Fourier multipliers} We refer to the book \cite{Lar71} for more information on Fourier multipliers on locally compact abelian groups. Let $G$ be a locally compact abelian group equipped with a Haar measure $\mu$. Suppose that $1 \leq p,q \leq \infty$. Following \cite[p.~66]{Lar71}, we say that a bounded operator $T \co \L^q(G) \to \L^p(G)$ is a bounded Fourier multiplier if it commutes with all translations, i.e.~for all $s \in G$ we have $T\tau_{s,q}=\tau_{s,p} T$ where $\tau_{s,p} \co \L^p(G) \to \L^p(G)$ is the translation operator defined by $(\tau_{s,p}f)(t)=f(ts^{-1})$ for any $s,t \in G$ and any $f \in \L^p(G)$. We say that a bounded Fourier multiplier $T$ is positive if for any almost everywhere positive function $f \in \L^q(G)$ the function $T(f)$ is also almost everywhere positive.

The following characterization of bounded Fourier multipliers from $\L^1$ into $\L^p$ is a particular case of results of the papers \cite{Arh24b} and \cite{Edw55}  (see also \cite[Theorem 3.1.1 p.~68]{Lar71}). Note the compactness assumption on the group $G$ in the two next statements. 

\begin{thm}
\label{Thm-description-multipliers-2}
Let $G$ be a compact abelian group. Suppose that $1 < p \leq \infty$.  Let $T \co \L^1(G) \to \L^p(G)$ be a linear map. The following conditions are equivalent.
\begin{enumerate}
\item $T \co \L^1(G) \to \L^p(G)$ is a bounded Fourier multiplier.
\item There exists a function $f \in \L^p(G)$ such that 
\begin{equation}
\label{T-as-convolution-QG}
T(g)
=f*g,\quad g \in \L^1(G).
\end{equation}
\item There exists a function $\varphi \in \L^\infty(\hat{G})$ such that
\begin{equation}
\label{Fourier-multiplier-and-transforms}
\widehat{T(g)}
=\varphi \hat{g}, \quad g \in \L^1(G).
\end{equation}
\end{enumerate} 
In this situation, the functions $f$ and $\varphi$ are unique and $\hat{f}
=\varphi$.  
%\begin{equation}
%%\label{hatf=varphi}
%.
%\end{equation}
Moreover, the map $T \co \L^1(G) \to \L^p(G)$ is completely bounded and we have
\begin{equation}
\label{norm-equality-QG}
\norm{f}_{\L^p(G)}
=\norm{T}_{\L^1(G) \to \L^p(G)}
=\norm{T}_{\cb,\L^1(G) \to \L^p(G)}.
\end{equation}
\end{thm}

Finally, recall the characterization of bounded Fourier multipliers on $\L^\infty(G)$ and positive Fourier multipliers on the same space. We refer to \cite[Theorem 3.4.2 p.~76]{Lar71} and \cite[Theorem 3.6.1 p.~82]{Lar71} (and duality).

\begin{thm}
\label{thm-positivity}
Let $G$ be a compact Abelian group. Suppose that $T \co \L^\infty(G) \to \L^\infty(G)$ is a linear map. Then the following statements are equivalent.
\begin{enumerate}
\item The map $T \co \L^\infty(G) \to \L^\infty(G)$ is a bounded Fourier multiplier.

%\item There exists a function $\varphi \in \L^\infty(\hat{G})$ on $G$ such that $\widehat{T(g)}=\varphi \hat{g}$ for each $f \in \L^\infty(G)$.

\item There exists a bounded Borel regular complex measure $\nu$ on $G$  such that $T(g) = \nu * g$ for any $f \in \L^\infty(G)$.

\item There exists a bounded Borel regular complex measure $\nu$ on $G$ such that 
\begin{equation}
\label{Fourier-fin}
\widehat{T(g)}=\hat{\nu} \hat{g}, \quad f \in \L^\infty(G).
\end{equation}

\end{enumerate}
In this situation, the measure $\nu$ is unique. Moreover, the equivalences remain true if we replace <<bounded>> by <<positive>> in the three statements.
\end{thm}

%\begin{thm}
%\label{Th-Larsen}
%Let $G$ be a locally compact abelian group. Suppose that $1 \leq p \leq \infty$. A bounded linear map $T \co \L^p(G) \to \L^p(G)$ is a positive Fourier multiplier if and only if there exists a positive bounded Borel measure such that $T(g)=\mu * g$ for any $g \in \L^p(G)$. In this case, the measure $\mu$ is unique.
%\end{thm}

%In this case, the Fourier transform $\hat{\mu}$ is the unique function $\varphi \in \L^\infty(\hat{G})$ such that
%\begin{equation}
%\label{multiplier-56}
%\widehat{T(g)}
%=\varphi \hat{g}, \quad g \in \L^2(G) \cap \L^p(G).
%\end{equation} 
%holds.

 %A continuous function is positive definite if and only if there exists a bounded positive measure such that $f=\hat{\mu}$. In this case the measure is unique. Bochner's theorem

%$\norm{f}_{\L^1(G,\mu)}=\frac{1}{|G|}\sum_{s \in G} f(s)=\frac{1}{|G|}\sum_{s \in G} \nu(s)=1$.
%
%
%that the convolution operator $T \co \ell^1_G \to \ell^1_G$, $g \mapsto \nu*g$ where 
%$$
%(\nu*g)(t)
%\ov{\mathrm{def}}{=} \sum_{s \in G} g\big(s^{-1}t\big)\nu(s)
%,\quad t \in G,
%$$
%
%Indeed, first we consider the function $f \co G \to \mathbb{C}$, $s \mapsto f(s)=|G|\nu(s)$. Note that $\norm{f}_{\L^1(G,\mu_G)}=\frac{1}{|G|}\sum_{s \in G} f(s)=\frac{1}{|G|}\sum_{s \in G} \nu(s)=1$.

\paragraph{Stein's interpolation theorem}
We recall an <<abstract version>> of Stein's interpolation theorem for compatible couples of Banach spaces. See \cite[Theorem 1]{CwJ84}, \cite[Theorem 2.7 p.~52]{Lun18} and \cite[Theorem 2.1]{Voi92} for the proof. Here, we use the open strip $S \ov{\mathrm{def}}{=} \{z \in \mathbb{C} : 0 < \Re z < 1 \}$.

\begin{thm}
\label{thm-Stein}
Let $(X_0,X_1)$ and $(Y_0,Y_1)$ be two compatible couples of Banach spaces such that $X_0$ is separable. Let $(T_z)_{z \in \ovl{S}}$ be a family of bounded operators from $X_0 \cap X_1$ into $Y_0+Y_1$ such that 
\begin{enumerate}
	
\item for any $x \in X_0 \cap X_1$ and any $y \in (Y_0+Y_1)^*$ the function $z \mapsto \la T_z(x),y\ra_{Y_0+Y_1,(Y_0+Y_1)^*}$ is continuous and bounded on $\ovl{S}$ and analytic on $S$,
	
\item for any $x \in X_0 \cap X_1$, the functions $t \to T_{\i t}(x)$ and $t \mapsto T_{1+\i t}(x)$ take values in $Y_0$ and $Y_1$,
	
\item there exist positive constants $M_0$ and $M_1$ such that for any $x \in X_0 \cap X_1$
$$
\sup_{t \in \R} \norm{T_{\i t}(x)}_{Y_0}
\leq M_0 \norm{x}_{X_0}
\quad \text{and} \quad
\sup_{t \in \R} \norm{T_{1+\i t}(x)}_{Y_1}
\leq M_0 \norm{x}_{X_1}.
$$
\end{enumerate}
%$$
%\{T_{\i t} : t\in \mathbb{R}\}\subset \B(X_0,Y_0) ,\quad \{T_{1+\i t} : t\in \mathbb{R}\}\subset  \B(X_1,Y_1).
%$$
%Suppose $M_0=\sup_{t \in \R} \norm{T_{\i t}}_{X_0 \to Y_0}$ and $M_1=\sup_{t \in \R} \norm{T_{1+\i t}}_{X_1 \to Y_1}$ are both finite, 
Then for any $ 0< \theta < 1$, $T_\theta$ admits a unique extension as a bounded linear map from the space $(X_0,X_1)_\theta$ into the space $(Y_0,Y_1)_\theta$ and
$$
\norm{T_\theta}_{(X_0,X_1)_\theta \to (Y_0,Y_1)_\theta}
\leq M_0^{1-\theta}M_1^{\theta}.
$$
%\noindent In particular, when $T$ is a constant map, the above theorem implies
%\begin{align}
%\norm{T}_{\B((X_0,X_1)_\theta ,(Y_0,Y_1)_\theta)} 
%\leq \norm{T}_{\B(X_0 ,Y_0)}^{1-\theta}\norm{T}_{\B(X_1 ,Y_1)}^{\theta}. \label{interpolation1}
%\end{align}
\end{thm}

\paragraph{Extension of linear maps}

We will use the following very useful folklore result. We sketch a proof. %is a particular case of \cite[Lemma 1 p.~390]{JuX07}.

\begin{lemma}
\label{lemma-trace-preserving} 
Let $\cal{M}$ and $\cal{N}$ be von Neumann algebras equipped with normal finite faithful traces. Suppose that $1 \leq p < \infty$. Let $T \co \cal{M} \to \cal{N}$ be a trace preserving unital $*$-homomorphism then $T$ induces a complete isometry $T$ from the space $\L^p(\cal{M})$ into $\L^p(\cal{N})$.
\end{lemma}

\begin{proof}
We denote by $\tau_{\cal{M}}$ and $\tau_{\cal{N}}$ the traces on $\cal{M}$ and $\cal{N}$. For any positive element $x$ in $\cal{M}$, we have
$$
\norm{T(x)}_{\L^1(\cal{N})}
=\tau_\cal{N}(T(x))
=\tau_\cal{M}(x)
=\norm{x}_{\L^1(\cal{M})}.
$$
Using the density \cite[Lemma 2.2]{ArK23} of the cone $\cal{M}_+$ in the cone $\L^1(\cal{M})_+$ for $\norm{\cdot}_1$, we obtain a continuous map  $T_+ \co \L^1(\cal{M})_+ \to \L^1(\cal{M})_+$, which is clearly additive. By \cite[Lemma 1.26 p.~24]{AlT07}, the map $T_+$ extends uniquely to a real linear positive map $T_{\sa} \co \L^1(\cal{M})_\sa \to \L^1(\cal{N})_\sa$, where $\L^1(\cal{M})_\sa$ is the space of selfadjoint elements in the space $\L^1(\cal{M})$ (the space $\L^1(\cal{N})_\sa$ is defined similarly). Let $x \in \L^1(\cal{M})_{\sa}$. We can write $x = x_1-x_2$ for some positive elements $x_1$ and $x_2$ with $\norm{x_1}_{\L^1(\cal{M})} \leq \norm{x}_{\L^1(\cal{M})}$ and $\norm{x_2}_{\L^1(\cal{M})} \leq \norm{x}_{\L^1(\cal{M})}$. Then
\begin{align*}
\MoveEqLeft
  \norm{T_\sa(x)}_{\L^1(\cal{M})}
=\norm{T_\sa(x_1-x_2)}_{\L^1(\cal{M})}
=\norm{T_\sa(x_1)-T_\sa(x_2)}_{\L^1(\cal{M})}\\
&\leq \norm{T_\sa(x_1)}_{\L^1(\cal{M})}+\norm{T_\sa(x_2)}_{\L^1(\cal{M})}
\leq \norm{x_1}_{\L^1(\cal{M})}+\norm{x_2}_{\L^1(\cal{M})}
\leq 2\norm{x}_{\L^1(\cal{M})}.
\end{align*}
If $x \in \L^1(\cal{M})$ we can write $x=x_1+\i x_2$ for some selfadjoint elements $x_1$ and $x_2$ with $\norm{x_1}_{\L^1(\cal{M})} \leq \norm{x}_{\L^1(\cal{M})}$ and $\norm{x_2}_{\L^1(\cal{M})} \leq \norm{x}_{\L^1(\cal{M})}$. We obtain
\begin{align*}
\MoveEqLeft
\norm{T_\sa(x_1)+\i T_\sa(x_2)}_{\L^1(\cal{M})}         
\leq \norm{T_\sa(x_1)}_{\L^1(\cal{M})} +\norm{T_\sa(x_2)}_{\L^1(\cal{M})}  \\
&\leq 2\norm{x_1}_{\L^1(\cal{M})} +2\norm{x_2}_{\L^1(\cal{M})} 
\leq 4 \norm{x}_{\L^1(\cal{M})}.
\end{align*}
By a complexification argument, we obtain a positive linear map $T_1 \co \L^1(\cal{M}) \to \L^1(\cal{M})$ with $\norm{T_1}_{\L^1(\cal{M}) \to \L^1(\cal{M})} \leq 4$, which extends $T$.

By duality, $R \ov{\mathrm{def}}{=} T_1^* \co \cal{N} \to \cal{M}$ is a bounded positive map. Consequently, by \cite[Corollary 2.9 p.~15]{Pau02} we have $\norm{R}_{\cal{N} \to \cal{M}} = \norm{R(1)}_\cal{M}$. Moreover, we have %\textbf{(if $\cal{M}_+$ ?is? a positively norming set, second equality)}
\begin{align*}
\MoveEqLeft
    \norm{R(1)}_\cal{M} 
=\norm{T_1^*(1)}_\cal{M}
=\sup_{x \in \cal{M}_+, \norm{x}_1 \leq 1} \tau(T_1^*(1)x)
=\sup_{x \in \cal{M}_+, \norm{x}_1 \leq 1} \tau(1T_1(x))\\
&=\sup_{x \in \cal{M}_+, \norm{x}_1 \leq 1} \tau(T(x))
= \sup_{x \in \cal{M}_+, \norm{x}_1 \leq 1} \tau(x)
= 1.
\end{align*}
Thus $R \co \cal{N} \to \cal{M}$ is contractive, and consequently $T_1 \co \L^1(\cal{M}) \to \L^1(\cal{N})$ is also contractive. Then by complex interpolation, $T$ extends to a contraction on the Banach space $\L^p(\cal{M})$ for all $1 \leq p \leq \infty$. 

%By \cite[Lemma 1 p.~390]{JuX07}, the map $T$ is normal and induces a contraction $T_p \co \L^p(\cal{M}) \to \L^p(\cal{N})$. 
Replacing $T$ by $\Id \ot T$, we obtain a complete contraction $T$. By \cite[Lemma 2.1 p.~2283]{Arh19}, the map $T$ is injective. So we can consider the $*$-homomorphism $T^{-1} \co T(\cal{M}) \to \cal{M}$, which extends to a complete contraction from the space $\L^p(T(\cal{M}))$ into $\L^p(\cal{M})$. The conclusion is transparent with \cite[(1.2.7) p.~6]{BLM04}.
\end{proof}

%\begin{remark} \normalfont
%Actually, it is known that such a map $T \co \cal{M} \to \cal{N}$ is normal by \cite[p.~249]{AcC82}.
%\end{remark}

%%%%%%%%%%%%%%%%%%%%%%%%%%%%%%%%%%%%%%%%%%%%%%%%%%%%%%%%%%%%%%%%%%%%%%%%%%%%%%%%%%%%%%%%%%%%%
\section{Radial multipliers on discrete hypercubes and fermion algebras}
\label{Sec-multipliers}

In this section, we introduce a notion of multipliers acting on fermion algebras. Some of these operators identify to well-known quantum channels (See Example \ref{Ex-multiplier-Fermions}).

\paragraph{The $n$-dimensional discrete hypercube and the Cantor group} 
If $n \in \{1,2,\ldots,\infty\}$, we equip the compact abelian group $\Omega_n \ov{\mathrm{def}}{=} \{-1,1\}^n$ with its canonical normalized Haar measure $\mu$. If $n$ is finite then the group $\Omega_n$ is sometimes called the $n$-dimensional discrete hypercube. For any complex function $f \co \Omega_{n} \to \mathbb{C}$, the integral over $\Omega_{n}$ is given by 
\begin{equation}
\label{measure-hypercube}
\int_{\Omega_{n}} f
=\frac{1}{2^n}\sum_{\epsi \in \{-1,1\}^n} f(\epsi),
\end{equation}
i.e.~the average over uniformly chosen signs. The group $\Omega_\infty$ is the Cantor group. For any integer $1 \leq i \leq n$, we denote by $\epsi_i \co \{-1,1\}^n \to \{-1,1\}$ the $i$-th coordinate on the product $\{ -1, 1 \}^n$. We can see the $\epsi_i$'s as independent Rademacher variables. We refer to \cite{EdW82} and \cite{SWS90} for more information on the Cantor group and to \cite{SBSW15a} and \cite{SBSW15b} for the historical aspect of the associated <<dyadic analysis>>. For a non-empty finite subset $A=\{i_1,\ldots,i_k\}$ of the set $\{1,\ldots,n\}$ with $i_1 < \cdots < i_n$ we define the Walsh function $w_A \co \Omega_n \to \R$ by
$$
w_A 
\ov{\mathrm{def}}{=} \epsi_{i_1} \cdots \epsi_{i_k}. 
$$
We also set $w_\emptyset \ov{\mathrm{def}}{=} 1$. In other words, for any finite subset $A$ of $\{1,\ldots,n\}$, we have
\[ 
w_A(\epsi_1,\ldots,\epsi_n)
=\prod_{j \in A} \epsi_j, \quad \epsi_1,\ldots,\epsi_n \in \{-1,1\}.
\]
If $n=\infty$, the family $(w_A)_{A \subset \{1,\ldots,n\}, A \textrm{ finite}}$ is the Walsh system and it is well-known that it can be identified with the Pontryagin dual $\hat{\Omega}_\infty$ of the Cantor group $\Omega_\infty$. It is worth noting that the Walsh system is an orthonormal basis of the Hilbert space $\L^2(\Omega_\infty)$. Finally, note that for any finite subsets $A$ and $B$ of $\{1,\ldots,n\}$ we have
\begin{equation}
\label{prod-walsh}
w_A w_B
=w_{A\Delta B},
\end{equation}
where we use the symmetric difference $A \Delta B \ov{\mathrm{def}}{=} \left(A \setminus B\right) \cup \left(B \setminus A\right)$, and
\begin{equation}
\label{int-walsh}
\int_{\Omega_n} w_A
= 0 \quad \text{if } A \not= \emptyset, \quad \text{i.e.} \quad \int_{\Omega_n} w_A=\delta_{A,\emptyset},
\end{equation}
where $\delta$ is the Kronecker symbol.

It is standard to write $\hat{f}(A)$ for the Fourier coefficient $\hat{f}(w_A)$. Every function $f \co \{-1,1\}^n \to \mathbb{C}$ can be uniquely written as
\begin{equation}
\label{Fourier-decompo}
f(\epsi_1,\ldots,\epsi_n)
\ov{\eqref{Fourier-inversion-formula}}{=} \sum_{A \subset \{1,\ldots,n\}} \hat f(A) w_A(\epsi_1,\ldots,\epsi_n),
\end{equation}
where the Fourier coefficient $\hat f(A)$ is given by
$$
\hat f(A)
\ov{\eqref{Fourier-transform-I}}{=} \int_{\Omega_n} w_A(\epsi)f(\epsi) \d\mu(\epsi).
$$
In particular, we have 
\begin{equation}
\label{Fourier-Walsh}
\hat{w}_A(B)
=\delta_{A,B}, \quad A,B \subset \{1,\ldots,n\}.
\end{equation}
%what follows, whenever  we use
%the probabilistic notation.
%$$
%\widehat f(A)
%=\E_\epsi \Bigl(f(\epsi)w_A(\epsi) \Bigr).
%$$

%\begin{equation}
%\label{radial-as-multiplier}
%T_g
%=C_{\sum_A g(|A|)w_A}
%\end{equation}

\paragraph{Radial multipliers on the $n$-dimensional discrete hypercube}
We introduce some <<radial multipliers>> on the finite group $\Omega_n$. Consider a complex function $\phi \co \{0,\ldots,n\} \to \mathbb{C}$ and the complex function $f_\phi \co \Omega_n \to \mathbb{C}$ defined by 
\begin{equation}
\label{def-hg}
f_\phi
\ov{\mathrm{def}}{=} \sum_A \phi(|A|)w_A.
\end{equation}
By the uniqueness of the decomposition \eqref{Fourier-decompo}, we see that that the Fourier coefficients of this function are given by
\begin{equation}
\label{coeff-1}
\widehat{f_\phi}(A)       
=\phi(|A|), \quad A \subset \{1,\ldots,n\}.
\end{equation}
Now, we introduce the convolution operator $T_\phi \co \L^\infty(\Omega_n) \to \L^\infty(\Omega_n)$, $g \mapsto f_\phi*g$. By \eqref{convol-finite-group-2}, this operator is the convolution operator by the measure $\nu \ov{\mathrm{def}}{=} f_\phi \cdot \mu$ on the group $\Omega_n$ defined by $\nu(\{s\})=\frac{1}{|\Omega_n|}f_\phi(s)$ for any $s \in \Omega_n$. By Theorem \ref{thm-positivity}, such a linear operator is a Fourier multiplier satisfying the formula 
\begin{equation}
\label{}
\widehat{T_\phi(w_A)}(B)
\ov{\eqref{Fourier-fin}}{=} \widehat{f_\phi \cdot \mu}(B) \hat{w}_A(B)
\ov{\eqref{Fourier-Walsh}}{=} \hat{f}_\phi(B)\delta_{A,B}
=\hat{f}_\phi(A)\delta_{A,B}
\ov{\eqref{coeff-1}}{=} \phi(|A|)\delta_{A,B}
\end{equation}
for any finite subsets $A$ and $B$ of $\{1,\ldots,n\}$. The uniqueness of the Fourier decomposition \eqref{Fourier-decompo} implies that is it equivalent to
\begin{equation}
\label{def-mult-rad}
T_\phi(w_A)
=\phi(|A|) w_A, \quad A \subset \{1,\ldots,n\}.
\end{equation}

\paragraph{Fermion algebras}
We refer to \cite{DeG13}, \cite[Section 9.B]{JMX06}, and \cite[Section 9.3]{Pis2} and \cite{PlR94} for information, other points of view or variants of these algebras. Let $H$ be a  real Hilbert space with dimension $n \in \{2,\ldots,\infty\}$. Here, we limit ourselves to defining $\scr{C}(H)$ when $H$ is finite-dimensional and we refer to previous references for the more involved infinite-dimensional case. In this case, the fermion algebra $\scr{C}(H)$ is the algebra generated by a sequence $(s_1,\ldots,s_n)$ of selfadjoint unitary operators $\not=\pm \Id$ acting on a complex Hilbert space $\cal{H}$ which satisfy the anticommutation rule
\begin{equation}
\label{anticommute-bis}
s_is_j
=-s_js_i, \quad i \not= j.
\end{equation}
%\begin{equation}
%\label{anticommute}
%s_is_j
%=-s_js_i, \quad i \not= j.
%\end{equation}operators on a finite-dimensional Hilbert space satisfying the canonical anti-commutation relations (CAR):
%QjQk + QkQj = 2δjk1 .
The element $s_A$ belonging to $\scr{C}(H)$ is defined for any non-empty finite subset $A=\{i_1,\ldots,i_k\}$ of the set $\{1,\ldots,n\}$ with $i_1 < \cdots < i_k$ by
\begin{equation}
\label{def-sA}
s_A 
\ov{\mathrm{def}}{=} s_{i_1} \cdots s_{i_k}.
\end{equation}
We also let $s_\emptyset \ov{\mathrm{def}}{=} 1$. Whether $H$ is finite-dimensional or not, recall that this algebra is equipped with a normalized normal finite faithful trace $\tau$ satisfying
\begin{equation}
\label{trace-fer}
\tau(s_A) 
= 0 \quad \text{if } A \not= \emptyset, 
\quad \text{i.e.} \quad 
\tau(s_A)=\delta_{A,\emptyset}.
\end{equation}
Moreover, it is well-known that the involutive algebra generated by the $s_i$'s is equal to the subspace $\mathrm{span}\{ s_A : A \subset \{1,\ldots,n\}\}$. The weak* closure of this algebra is the von Neumann algebra $\scr{C}(H)$. Moreover, for any finite subsets $A$ and $B$ of $\{1,\ldots,n\}$ we have 
\begin{equation}
\label{product-sA}
s_As_B
=(-1)^{n(A,B)}s_{A \Delta B},
\end{equation}
where $n(A,B)$ is an integer depending only on $A$ and $B$. Finally, we have 
\begin{equation}
\label{adjoint-sA}
(s_A)^*
=(-1)^{n_A} s_A
\end{equation}
for some integer $n_A$.

%There is a normalized trace $\tau$ on this algebra such that $\tau(s_A) = 0$ if $A \not= \emptyset$.  

%It is known that the $s_A$'s define an orthonormal basis of the Hilbert space $\L^2(\scr{C}(H))$. 
%Let $R$ be the $w*$-closure of $R_0$. $R$ is again a hyperfinite factor, so $R$ is isomorphic to $N$. 
%Note that the von Neumann subalgebra generated by ${s_1 , \ldots , s_{2n}}$ is isomorphic to $\M_{2n}$.

%Assume now that $H$ is infinite dimensional, and let $(e_i)_{i \geq 1}$ be an orthonormal family. We let $s_i=w(e_i)$ for any $i \geq 1$. 

%It is well-known that these operators form a `spin system'. Namely they are hermitian unitaries on $\Lambda(\H)$ and
%$$
%s_is_j 
%= -s_js_i,
%\qquad\hbox{ if}\,\ i\not= j.
%$$
%We let $I$ be the set of all increasing finite sequences $\{i_1 < i_2 <\cdots < i_m\}$ of positive integers. 
%If $F$ is such a sequence we let
%$$
%s_A 
%= s_{i_1}\cdots s_{i_m}.
%$$
%By convention, the empty set belongs to $\I$, and we let $s_{\emptyset} = 1$. Also we write $\vert F\vert$ for the cardinal of $A \in \I$. 
%Since the $s_i$'s form a spin system, the $*$-algebra they generate is equal to
%$$
%\P 
%= \textrm{span} \big\{ s_A :  A \in I \big\}.
%$$
%Thus $\P$ is weak* dense in $\scr{C}(H)$, and it is dense in $\L^p(\scr{C}(H))$ for any $1 \leq p <\infty$. 

\paragraph{Radial multipliers on fermion algebras}
We now present a notion of <<radial multipliers>> on fermion algebras which parallels the previous notion of radial multipliers on discrete hypercubes. Such a multiplier is a weak* continuous map $R_\phi \co \L^\infty(\scr{C}(H)) \to \L^\infty(\scr{C}(H))$ defined by a complex function $\phi \co \{0,\ldots,n\} \to \mathbb{C}$ by the formula
\begin{equation}
\label{def-rad-mult}
R_\phi(s_A)
= \phi(\vert A\vert) s_A, \quad A \subset \{1,\ldots,n\}.
\end{equation}
In this situation, we say that the complex function $\phi \co \{0,\ldots,n\} \to \mathbb{C}$ is the symbol of the radial multiplier.

%Now, we introduce some multipliers inspired by \cite{HoR11}.

%\begin{defi}
%Let $n \in \mathbb{N} \cup \{\infty\}$. For any function $g \co \{0,\ldots,N\} \to \mathbb{C}$, 
%Then we say that a weak* continuous linear map $T \co \scr{C}(H) \to \scr{C}(H)$ is a <<radial multiplier>> if there exists a complex function $g \co \{0,\ldots,N\} \to \mathbb{C}$ such that
%\begin{equation}
%\label{def-rad-mult}
%T(s_A)
%= g(\vert A\vert) s_A, \quad A \subset \{1,\ldots,N\}.
%\end{equation}
%
%
%
%the linear map
%\begin{equation}
%\label{def-rad-mult}
%M_g(s_A)
%= g(\vert A\vert) s_A
%\end{equation}
%is called radial multiplier if it induces a weak* continuous linear map on the algebra $\scr{C}(H)$. We say that $g$ is the symbol of the multiplier.
%\end{defi}

\begin{remark} \normalfont
\label{remark-trace-preserving}
It is easy to see that such a radial multiplier $R_\phi$ is unital if and only if $\phi(0)=1$. This is also equivalent to say that the map $R_\phi$ is trace preserving.
\end{remark}

Let us now explore this concept when applied to the fermion algebra of the smallest dimension.

\begin{example} \normalfont
\label{Ex-multiplier-Fermions}
Consider the case $n=2$ and the Pauli matrices $s_1
\ov{\mathrm{def}}{=} \begin{bmatrix}
      0&1\\
      1&0
    \end{bmatrix} $ and $s_2 \ov{\mathrm{def}}{=}
    \begin{bmatrix}
      0& -\i \\
      \i&0
    \end{bmatrix} $.  
		%$s_3=
    %\begin{bmatrix}
      %1&0\\
      %0&-1
    %\end{bmatrix}
%$. 
These matrices are selfadjoint unitaries and anticommute. We have $s_{\{1,2\}}\ov{\eqref{def-sA}}{=} s_1s_2=\begin{bmatrix}
      0&1\\
      1&0
    \end{bmatrix}    \begin{bmatrix}
      0& -\i \\
      \i&0
    \end{bmatrix} 
		=    \begin{bmatrix}
    \i  & 0 \\
    0  & -\i
    \end{bmatrix}=\i \sigma_3$ where $\sigma_3 \ov{\mathrm{def}}{=} \begin{bmatrix}
    1 & 0 \\
    0 & -1
\end{bmatrix}
$. 
Note that a 2x2 complex matrix $\begin{bmatrix}
    a & b \\
    c & d
\end{bmatrix}$ can be written %(\href{https://physics.stackexchange.com/questions/292102/how-do-one-show-that-the-pauli-matrices-together-with-the-unit-matrix-form-a-bas}{mathstack})
\begin{equation}
\label{decompo}
\begin{bmatrix}
    a & b \\
    c & d
\end{bmatrix}
=\frac{a+d}{2}\I +\frac{b+c}{2}\sigma_1+\frac{\i(b-c)}{2}\sigma_{2}+\frac{a-d}{2\i}\i\sigma_3.
\end{equation}
Consider a complex function $\phi \co \{0,1,2\} \to \mathbb{C}$. Let us concretely describe the associated radial multiplier $R_\phi \co \M_2 \to \M_2$, $\alpha \I_2+\beta s_1+\gamma s_2+\zeta s_{\{1,2\}} \mapsto \phi(0)\alpha \I_2+\phi(1)\beta s_1+\phi(1)\gamma s_2+\phi(2)\zeta s_{\{1,2\}}$. For any matrix $\begin{bmatrix}
    a & b \\
    c & d
\end{bmatrix}$ of $\M_2$, we have
\begin{align*}
\MoveEqLeft
R_\phi\left(\begin{bmatrix}
    a & b \\
    c & d
\end{bmatrix}\right)         
\ov{\eqref{decompo}}{=} R_\phi\left(\frac{a+d}{2}\I +\frac{b+c}{2}s_1+\frac{\i(b-c)}{2}s_{2}+\frac{a-d}{2\i}\i\sigma_3\right)  \\
& \ov{\eqref{def-rad-mult}}{=} \phi(0)\frac{a+d}{2}\I +\phi(1)\frac{b+c}{2}s_1+\phi(1)\frac{\i(b-c)}{2}s_{2}+\phi(2)\frac{a-d}{2\i}\i \sigma_3\\
&=\begin{bmatrix}
  \phi(0)\frac{a+d}{2}+\phi(2)\frac{a-d}{2}  & \phi(1)\frac{b+c}{2}+\phi(1)\frac{b-c}{2} \\
  \phi(1)\frac{b+c}{2}-\phi(1)\frac{b-c}{2}   & \phi(0)\frac{a+d}{2}-\phi(2)\frac{a-d}{2}
\end{bmatrix} 
=\begin{bmatrix}
  \phi(0)\frac{a+d}{2}+\phi(2)\frac{a-d}{2}  & \phi(1)b \\
  \phi(1)c   & \phi(0)\frac{a+d}{2}-\phi(2)\frac{a-d}{2}
\end{bmatrix}.
\end{align*}
In particular, if $\phi(0)=1$, $\phi(1)=1-2t$ and $\phi(2)=1$ for $0 \leq t \leq 1$, we obtain
\begin{equation}
\label{}
R_\phi\left(\begin{bmatrix}
    a & b \\
    c & d
\end{bmatrix}\right) 
=\begin{bmatrix}
 1    & 1-2t \\
 1-2t    & 1
\end{bmatrix}*\begin{bmatrix}
    a & b \\
    c & d
\end{bmatrix},
\end{equation}
where $*$ is the Schur product. Consequently, this radial multiplier identifies to the dephasing channel $T_t \co \M_2 \to \M_2$, $x \mapsto (1-t)x+tZxZ$ of \cite[p.~155]{Wil17} (see also \cite[Example~2.14 p.~19]{Pet08}) defined with the Pauli matrix $Z
\ov{\mathrm{def}}{=} \begin{bmatrix}
   1  &  0 \\
   0  &  -1 \\
\end{bmatrix}$. Indeed, for any complex matrix $\begin{bmatrix}
    a & b \\
    c & d
\end{bmatrix}$, we have
\begin{align*}
\MoveEqLeft
T_t\left(\begin{bmatrix}
    a & b \\
    c & d
\end{bmatrix}\right)         
=(1-t)\begin{bmatrix}
    a & b \\
    c & d
\end{bmatrix}+t\begin{bmatrix}
   1  &  0 \\
   0  &  -1 \\
\end{bmatrix}\begin{bmatrix}
    a & b \\
    c & d
\end{bmatrix}\begin{bmatrix}
   1  &  0 \\
   0  &  -1 \\
\end{bmatrix} \\
&=(1-t)\begin{bmatrix}
    a & b \\
    c & d
\end{bmatrix}+t\begin{bmatrix}
   1  &  0 \\
   0  &  -1 \\
\end{bmatrix}\begin{bmatrix}
    a & -b \\
    c & -d
\end{bmatrix} \\
&=(1-t)\begin{bmatrix}
    a & b \\
    c & d
\end{bmatrix}+t\begin{bmatrix}
    a & -b \\
    -c & d
\end{bmatrix}
=\begin{bmatrix}
    a & (1-2t)b \\
    (1-2t)c & d
\end{bmatrix} .
\end{align*}
\end{example}

\begin{remark} \normalfont
If we replace the Pauli matrices $\begin{bmatrix}
      0&1\\
      1&0
    \end{bmatrix} $ and $
    \begin{bmatrix}
      0& -\i \\
      \i&0
    \end{bmatrix}$ by $\begin{bmatrix}
   1  & 0  \\
   0  & -1  \\
\end{bmatrix}$ and $
\begin{bmatrix} 
0&1\\ 
1&0
\end{bmatrix}$, it seems to us that we obtain different channels.
\end{remark}

\begin{remark} \normalfont
It would be interesting to look at the higher dimensions with the help of a computer.
\end{remark}

\begin{example} \normalfont
\label{Ex-Ornstein}
Let $H$ be a separable real Hilbert space with dimension $n \in \{2,\ldots,\infty\}$. The fermionic Ornstein-Uhlenbeck semigroup $(T_t)_{t \geq 0}$ is a Markov semigroup of operators acting on the von Neumann algebra $\scr{C}(H)$. Each contraction $T_t \co \L^\infty(\scr{C}(H)) \to \L^\infty(\scr{C}(H))$ is defined by
$$
T_t(s_A)
= \e^{-t\vert A\vert} s_A, \quad A \subset \{1,\ldots,n\}.
$$
So each operator $T_t$ is a radial multiplier with associated function $\phi_t\co \{0,\ldots,n\} \to \R$ defined by $\phi_t(|A|)=\e^{-t|A|}$. We refer to \cite{CaL93} and \cite[Section 9.B]{JMX06} for more information.
\end{example}

%%%%%%%%%%%%%%%%%%%%%%%%%%%%%%%%%%%%%%%%%%%%%%%%%%%%%%%%%%%%%%%%%%%%%%%%%%%%%%%%%%%%%%%%%%%%%
\section{Ergodic actions of discrete hypercubes on fermion algebras}
\label{sec-ergodic-action}

\paragraph{Coproducts} We need some classical construction of quantum group theory, see \cite{Kus05} for more information on this theory. Let $G$ be a compact group equipped with its normalized Haar measure. For any continuous function $f \co G \to \mathbb{C}$, we define
\begin{equation}
\label{Delta_commutatif}
(\Delta f)(s,t)
\ov{\mathrm{def}}{=} f(st),\quad s,t\in G.
\end{equation} 
It is well-known \cite[p.~145]{Kus05} that $\Delta$ induces an integral preserving normal unital injective $*$-homomorphism $\Delta \co \L^\infty(G) \to \L^\infty(G) \otvn \L^\infty(G)=\L^\infty(G \times G)$, called coproduct, which extends by Lemma \ref{lemma-trace-preserving} to a complete isometry from the Banach space $\L^p(G)$ into the space $\L^p(G \times G)$ for any $1 \leq p < \infty$, also denoted by $\Delta \co \L^p(G) \to \L^p(G \times G)$. Note that $(\L^\infty(G),\Delta)$ is a compact quantum group.

%If $\mu_G$ is a left Haar measure on $G$, then a left Haar weight $\varphi$ on the abelian von Neumann algebra $\L^\infty(G)$ is given by $\varphi(f)=\int_G f \d \mu_G$ and similarly we obtain a right Haar weight $\psi$ with a right Haar measure. In this situation, $\L^\infty(G)$ equipped with $\Delta$, $\varphi$ and $\psi$ is a locally compact quantum group which is compact if and only if the group $G$ is compact. 

%The antipode $S$ is given by $(Sf)(s)=f(s^{-1})$ where $s \in G$. If $G$ is compact, for any $f \in \L^\infty(G)$, we can show that 
%$$
%\hat{f}(\pi)
%=\int_G \pi(s)^*f(s) \d\mu_G(s),\quad \pi \in \Irr(G).
%$$ 
%\begin{example} \normalfont
%\begin{equation}
%\label{coproduct-VNG}
%\Delta(w_A)
%\ov{\mathrm{def}}{=} w_A \ot w_A,\quad .
%\end{equation} 
\begin{example} \normalfont
Consider the case of the compact abelian group $\Omega_n$ for $n \in \{1,2,\ldots,\infty\}$ and the linear map $\Delta \co \L^p(\Omega_n) \to \L^p(\Omega_n \times \Omega_n)$ for some $1 \leq p \leq \infty$. For any finite subset $A$ of the set $\{1,\ldots,n\}$, we have
\begin{align*}
\MoveEqLeft
\Delta(w_A)(\omega_1,\omega_2)         
 \ov{\eqref{Delta_commutatif}}{=} w_A(\omega_1 \omega_2)  
=w_A(\omega_1) w_A(\omega_2)
=(w_A \ot w_A)(\omega_1,\omega_2), \quad \omega_1,\omega_2 \in \Omega_n,
\end{align*}
where we used the fact that each Walsh function $w_A$ is a character of the abelian group $\Omega_n$. We conclude that
\begin{equation}
\label{delta-epsi}
\Delta(w_A)
=w_A \ot w_A, \quad A \subset \{1,\ldots,n\}.
\end{equation}
\end{example}

\paragraph{Group actions} Recall the following notion of action of \cite[Definition I.1]{Eno77} and \cite[Definition 1.1]{Vae01}. We also refer to \cite{DeC17} for more information on actions of compact quantum groups. A right action of a compact group $G$ on a von Neumann algebra $\cal{M}$ is a normal unital injective $*$-homomorphism $\alpha \co \cal{M} \to \cal{M} \otvn \L^\infty(G)$ such that
\begin{equation}
\label{Def-corep}
(\alpha \ot \Id) \circ \alpha
=(\Id \ot \Delta) \circ \alpha.
\end{equation}
Similarly, a left action of $G$ on a von Neumann algebra $\cal{M}$ is a normal unital injective $*$-homomorphism $\beta \co \cal{M} \to \L^\infty(G) \otvn \cal{M}$ such that
\begin{equation}
\label{Def-corep-left}
(\Id \ot \beta) \circ \beta
=(\Delta \ot \Id) \circ \beta.
\end{equation}

\begin{remark} \normalfont
Let us recall that there is a simple passage from right to left actions (and conversely) if $G$ is abelian. Indeed, if $\flip \co \cal{M} \otvn \L^\infty(G) \to \L^\infty(G) \otvn \cal{M}$, $x \ot y \mapsto y \ot x$ denotes the flip map it is easy to check that a map $\alpha \co \cal{M} \to \cal{M} \otvn \L^\infty(G)$ is a right action of $G$ on $\cal{M}$ if and only if $\beta \ov{\mathrm{def}}{=} \flip \circ \alpha \co \cal{M} \to \L^\infty(G) \otvn \cal{M}$ is a left action of the group $G$.
 %\cite[Section 4]{KuV03}, i.e. the locally compact quantum group defined by $\L^\infty(G^\op) \ov{\mathrm{def}}{=} \L^\infty(G)$ and the reversed coproduct $\Delta_{G^\op} \ov{\mathrm{def}}{=} \flip \circ \Delta_G$ where $\flip$ denotes again a suitable flip map. 
This allows simple translation of many results about right actions to left actions and vice versa.
\end{remark}

\begin{remark} \normalfont
\label{Rem-classical-action}
Let $\cal{M}$ be a von Neumann algebra and $G$ be a compact group. Suppose that we have a group homomorphism $\Phi \co G \to \Aut(\cal{M})$ such that for each $x \in \cal{M}$ the function $G \to \cal{M}$, $s \mapsto \Phi_s(x)$ is continuous if $\cal{M}$ is endowed with the weak* topology. This homomorphism can be translated to a map $\alpha \co \cal{M} \to \cal{M} \otvn \L^\infty(G)=\L^\infty(G,\cal{M})$, $x \mapsto \alpha(x)$ where $\alpha(x)(s) \ov{\mathrm{def}}{=} \Phi_s(x)$. By \cite[p.~263]{Str81} or \cite[Proposition I.3]{Eno77}, the map $\alpha$ is a right action of the group $G$ on the von Neumann algebra $\cal{M}$. A reverse procedure is described in 
\cite[p.~265]{Str81} and \cite[Proposition I.3]{Eno77}. So there is a bijective correspondence between such homomorphisms and right actions. 
\end{remark}

\paragraph{Ergodic actions} For a right action $\alpha \co \cal{M} \to \cal{M} \otvn \L^\infty(G)$ of a compact group $G$ as in \eqref{Def-corep}, we introduce the fixed point algebra \cite[Definition 1.2]{Vae01}
$$
\cal{M}^\alpha 
\ov{\mathrm{def}}{=} \{x \in \cal{M} : \alpha(x) = x \ot 1\}.
$$
Note that this algebra is a unital sub-von Neumann algebra of the von Neumann algebra $\cal{M}$. If $\Phi \co G \to \Aut(\cal{M})$ is the associated group homomorphism, we have 
$$
\cal{M}^\alpha=\{x \in \cal{M}: \Phi_s(x)=x \text{ for any } s \in G\}.
$$ 
The right action $\alpha$ is said to be ergodic if the fixed point subalgebra $\cal{M}^\alpha$ equals $\mathbb{C}1$. 

If the compact group $G$ is equipped with its normalized Haar measure and if $\alpha \co \cal{M} \to \cal{M} \otvn \L^\infty(G)$ is an ergodic right action then, by adapting the reasoning found in \cite[pp.~94-95]{Boc95}, originally applied to $\mathrm{C}^*$-algebras, to the context of von Neumann algebras, the von Neumann algebra $\cal{M}$ admits a unique normal state $h_\cal{M}$ determined by the following condition
\begin{equation}
\label{Def-ergodic}
\bigg(\Id \ot \int_G\bigg) \circ \alpha(x)
=h_\cal{M}(x)1_{\cal{M}}, \quad x \in \cal{M}.
\end{equation} 
This means that for any $x \in \cal{M}$ the vector-valued integral $\int_G \alpha(x)(s) \d s$ (Gelfand integral) is a scalar operator $h_\cal{M}(x)1_{\cal{M}}$. Then $h_\cal{M}$ is the unique $G$-invariant normal state on the von Neumann algebra $\cal{M}$. Note that $h_\cal{M}$ is faithful by adapting \cite[Remark 2]{Boc95} to the context of von Neumann algebras.

The right action $\alpha \co \cal{M} \to \cal{M} \otvn \L^\infty(G)$ preserves the states, since for any $x \in \cal{M}$, we have
\begin{align*}
\MoveEqLeft
\bigg(h_\cal{M} \ot \int_G\bigg) \circ \alpha(x)
=h_\cal{M}\bigg(\bigg(\Id \ot \int_G\bigg) \circ \alpha(x)\bigg)
\ov{\eqref{Def-ergodic}}{=} h_\cal{M}\big(h_\cal{M}(x)1_{\cal{M}}\big) \\
&=h_\cal{M}(x)h_\cal{M}(1_{\cal{M}})
=h_\cal{M}(x).         
\end{align*}
If there exists a right ergodic action of a compact group $G$ on a von Neumann algebra $\cal{M}$ then it is known \cite[Corollary 4.2]{HLS81} that the von Neumann algebra $\cal{M}$ is necessarily finite and injective. Moreover, the proof of \cite[Corollary 4.2]{HLS81} shows that the state $h_\cal{M}$ is necessarily a trace. Finally, we refer to \cite{AHK80} and \cite{OPT80} for a complete classification of ergodic actions of \textit{abelian} compact groups\footnote{\thefootnote. Note that \cite[Lemma 1.24]{ChH1} says that abelian compact groups are ergodically rigid, i.e.~the only ergodic actions of abelian compact groups are on type I von Neumann algebras. Unfortunately, this contradicts the results of \cite{OPT80}. So the result \cite[Lemma 1.24]{ChH1} is false (confirmed by email).} and to the important papers \cite{Was88a}, \cite{Was88b} and \cite{Was89}.

We will use the the following result from \cite{Arh24b}.

\begin{prop}
\label{Prop-transfer-1}
Let $G$ be a compact group equipped with its normalized Haar measure and let $\cal{M}$ be a finite von Neumann algebra. Let $\alpha \co \cal{M} \to \cal{M} \otvn \L^\infty(G)$ be an ergodic right action. Consider $1 \leq p \leq \infty$. The map $\alpha$ induces a completely contractive map $\alpha \co \L^1(\cal{M}) \to \L^p(\cal{M},\L^1(G))$, where $\cal{M}$ is equipped with the unique $G$-invariant trace.
\end{prop}

%\begin{thm}
%\label{thm-Smin-transfert}
%Let $G$ be a compact group equipped with its normalized Haar measure and $\cal{M}$ be a von Neumann algebra. Let $\alpha \co \cal{M} \to \cal{M} \otvn \L^\infty(G)$ be an ergodic right action. Consider $1 \leq p <\infty$. Then for any $f \in \L^1(G)$ we have
%\begin{equation}
%\label{transference-norms}
%\norm{T_{f,\alpha}}_{\cb,\L^1(\cal{M}) \to \L^p(\cal{M})}
%\leq \norm{T_f}_{\cb,\L^1(G) \to \L^p(G)}.
%\end{equation}
%and
%\begin{equation}
%\label{belles-estimations}
%\H_{\cb,\min}\big(T_{f,\alpha}\big)
%\geq \H(f)%, \quad
%\quad \text{and} \quad
%\chi(T_{f,\alpha}) 
%\leq -\H(f).
%%\quad \text{and} \quad
%%-\H(f)
%%\leq \Q^1(T_{f,\alpha}) 
%\end{equation}
%%where we suppose $\cal{M}=\M_n$ for the last inequality.
%\end{thm}

\paragraph{Ergodic actions of the $n$-dimensional discrete hypercube}
Let $H$ be a finite-dimensional real Hilbert space with dimension $n \geq 2$. We will present a right ergodic action from the group $\Omega_n$ on the fermion algebra $\scr{C}(H)$. For that, we introduce the linear maps $\alpha \co \scr{C}(H) \to \scr{C}(H) \otvn \L^\infty(\Omega_n)$ and $\beta \co \scr{C}(H) \to \L^\infty(\Omega_n) \otvn \scr{C}(H)$ defined
by
\begin{equation}
\label{actions-alpha-beta}
\alpha(s_A)
\ov{\mathrm{def}}{=} s_A \ot w_A,\quad
\beta(s_A)
\ov{\mathrm{def}}{=} w_A \ot s_A,\quad A \subset \{1,\ldots,n\}.
\end{equation}
We have $\beta = \flip \circ \alpha$.

\begin{prop}
\label{prop-ergodic}
Let $H$ be a finite-dimensional real Hilbert space with dimension $n \geq 2$. The map $\alpha$ defines a right action from the compact abelian group $\Omega_n$ on the fermion algebra $\scr{C}(H)$ which is ergodic. Finally, the unique $\Omega_n$-invariant normal state on the fermion algebra $\scr{C}(H)$ is the trace $\tau$ defined by \eqref{trace-fer}
\end{prop}

\begin{proof}
For any subsets $A$ and $B$ of $\{1,\ldots,n\}$, we have
\begin{align*}
\MoveEqLeft
\alpha(s_As_B)
\ov{\eqref{product-sA}}{=} (-1)^{n(A,B)} \alpha(s_{A \Delta B})
\ov{\eqref{actions-alpha-beta}}{=} (-1)^{n(A,B)} s_{A \Delta B} \ot w_{A \Delta B}
= \\
&\ov{\eqref{product-sA}\eqref{prod-walsh} }{=} s_A s_B \ot w_A w_B
=(s_A \ot w_A)(s_B \ot w_B)
\ov{\eqref{actions-alpha-beta}}{=} \alpha(s_A)\alpha(s_B)
\end{align*}
and
$$
\alpha(s_A^*)
\ov{\eqref{adjoint-sA}}{=} (-1)^{n_A} \alpha(s_A)
\ov{\eqref{actions-alpha-beta}}{=}(-1)^{n_A}s_A \ot w_A
\ov{\eqref{adjoint-sA}}{=}(s_A \ot w_A)^*
\ov{\eqref{actions-alpha-beta}}{=} \alpha(s_A)^*.
$$
Consequently, the linear map $\alpha$ is a $*$-homomorphism, which is unital by construction. We also have
\begin{align*}
\MoveEqLeft
\bigg(\tau \ot \int_{\Omega_n}\bigg)(\alpha(s_A))
\ov{\eqref{actions-alpha-beta}}{=} \bigg(\tau \ot \int_{\Omega_n}\bigg)(s_A \ot w_A) 
=\tau(s_A) \int_{\Omega_n} w_A
\ov{\eqref{trace-fer} \eqref{int-walsh}}{=} 
\delta_{A,\emptyset}
\ov{\eqref{trace-fer}}{=} \tau(s_A).         
\end{align*}
Hence the map $\alpha$ is trace preserving. The injectivity is a consequence of \cite[Lemma 2.1 p.~2283]{Arh19} but can also be proved directly. For any finite subset $A$ of the set $\{1,\ldots,n\}$, we have
$$
(\alpha \ot \Id)\alpha(s_A)
\ov{\eqref{actions-alpha-beta}}{=} (\alpha \ot \Id)(s_A \ot w_A)
=\alpha(s_A) \ot w_A
\ov{\eqref{actions-alpha-beta}}{=} s_A \ot w_A \ot w_A
$$
and
$$
(\Id \ot \Delta)\alpha(s_A)
\ov{\eqref{actions-alpha-beta}}{=} (\Id \ot \Delta)(s_A \ot w_A)
=s_A \ot \Delta(w_A)
\ov{\eqref{delta-epsi}}{=} s_A \ot w_A \ot w_A.
$$
By linearity, we obtain the equality \eqref{Def-corep}. 

In order to prove the ergodicity of $\alpha$,  consider some element $x \in \scr{C}(H)$ such that $\alpha(x) = x \ot 1$. If $\sum_{A} \lambda_A s_A$ is the <<Fourier series>> of $x$ in the space $\L^2(\scr{C}(H))$, where each $\lambda_A$ is a complex number, we have $\alpha(\sum_{A} \lambda_A s_A) \ov{\eqref{actions-alpha-beta}}{=} \sum_{A} \lambda_A s_A \ot w_A$. The equality $\sum_{A} \lambda_A s_A \ot w_A=\big(\sum_{A} \lambda_A s_A\big) \ot 1$ is equivalent to $\sum_{A} \lambda_A s_A=1$. So the ergodicity is clear.
%2. For any finite subset $A$, we have
%\begin{align*}
%\MoveEqLeft
%R_{\phi,\sigma}(s_A)            
%=(f \ot \Id)\sigma(s_A)   
%=(f \ot \Id)(w_A \ot s_A)  \\
%&=\phi(|A|)s_A
%=M_{\phi}(s_A).
%\end{align*}
%3. Let $\E_\sigma \co \L^\infty(\{-1,1\}^N) \otvn \scr{C} \to \scr{C}$ be the canonical trace preserving conditional expectation associated with $\sigma$. We have $\E_\sigma\sigma=\Id$. Let $\E \co \M_K \to \scr{C}$ from the matrix algebra $\M_{\{-1,1\}^N}$ onto and $J \co \scr{C} \to \M_K$ the canonical embedding.
%
%Note $\Id_{\ell^\infty_I}$ is a Schur multiplier hence Hadamard hence strongly additive \cite{WiY1}
%
%\begin{align*}
%\MoveEqLeft
%\Q(M_{\phi})          
%=\Q(R_{f,\sigma}) 
%=\Q(\E_\sigma \sigma R_{f,\sigma})
%\ov{\eqref{commute-transference}}{=} 
%\Q\big(\E_\sigma(R_f \ot \Id_{\scr{C}}) \circ \sigma\big)\\
%&\ov{\eqref{capacity-composition}}{\leq} \Q(R_f \ot \Id_{\scr{C}})
%=
%?=?\Q(R_f)
%\ov{\eqref{C-leq-Q}}{\leq} \C(R_f)
%\ov{\eqref{CT-Fourier}}{\leq} -\H(f).
%\end{align*} 
%
%By \cite[Theorem 5.1.5]{ER} and Definition \ref{Def-PPT}, note that $\Id_{\L^\infty} \co \L^\infty(\{-1,1\}^N) \to \L^\infty(\{-1,1\}^N)$ is entangling breaking. 
%Using \cite[Exercise 8.2 page 556]{Wat1} \cite[]{Wil17} in the fourth equality, we deduce that

Finally, for any finite subset $A$ of the set $\{1,\ldots,n\}$ we have
\begin{align*}
\MoveEqLeft
\bigg(\Id \ot \int_{\Omega_n}\bigg) \circ \alpha(s_A)
\ov{\eqref{actions-alpha-beta}}{=} \bigg(\Id \ot \int_{\Omega_n}\bigg) (s_A \ot w_A)
=s_A\int_{\Omega_n} w_A \\
&\ov{\eqref{int-walsh}}{=} \delta_{A,\emptyset} s_A 
=\delta_{A,\emptyset} 1_{\scr{C}(H)} 
\ov{\eqref{trace-fer}}{=} \tau(s_A)1_{\scr{C}(H)}.
\end{align*}
\end{proof}

\begin{prop}
\label{prop-autre-extension}
Let $H$ be a finite-dimensional real Hilbert space with dimension $n \geq 2$. Suppose that $1 \leq p \leq \infty$. The map $\beta \co \scr{C}(H) \to \L^\infty(\Omega_n) \otvn \scr{C}(H)$ is a trace preserving unital injective $*$-homomorphism and induces a complete contraction $\beta \co \L^1(\scr{C}(H)) \to \L^p(\Omega_n,\L^1(\scr{C}(H)))$.
\end{prop}

\begin{proof}
Since $\beta = \flip \circ \alpha$ and since $\alpha$ is a trace preserving unital injective $*$-homomorphism by Proposition \ref{prop-ergodic}, we see that $\beta$ is trace preserving unital injective $*$-homomorphism. It is not difficult to prove that
\begin{equation}
\label{ergo-biz-biss}
(\Id_{\L^\infty(\Omega_n)} \ot \tau) \circ \beta(x)
=\tau(x) 1_{\L^\infty(\Omega_n)}, \quad x \in \scr{C}(H).
\end{equation}
Indeed, for any subset $A$ of $\{1,\ldots,n\}$, we have
\begin{align*}
\MoveEqLeft
(\Id_{\L^\infty(\Omega_n)} \ot \tau) \circ \beta(s_A)
\ov{\eqref{actions-alpha-beta}}{=} (\Id_{\L^\infty(\Omega_n)} \ot \tau) (w_A \ot s_A)
=\tau(s_A) w_A.
\end{align*}
In the case of a \textit{positive} element $x$ in the space $S^\infty_k(\L^1(\scr{C}(H)))$, we can write
\begin{align}
\MoveEqLeft
\label{divers-566789-bis}
\bnorm{(\Id_{S^\infty_k} \ot \beta)(x)}_{S^\infty_k(\L^\infty(\Omega_n,\L^1(\scr{C}(H))))}              
\ov{\eqref{Norm-LpL1-pos}}{=} \bnorm{(\Id_{S^\infty_k(\L^\infty(\Omega_n))} \ot \tau)((\Id_{S^\infty_k} \ot \beta)(x))}_{S^\infty_k(\L^\infty(\Omega_n))} \\
&= \bnorm{\big[\Id_{S^\infty_k} \ot \big(\Id_{\L^\infty(\Omega_n)} \ot \tau \big) \circ \beta\big](x) }_{S^\infty_n(\scr{C}(H))} 
\ov{\eqref{ergo-biz-biss}}{=} \bnorm{\big(\Id_{S^\infty_k} \ot \tau \big)(x) \ot 1_{\L^\infty(\Omega_n)}}_{S^\infty_k(\scr{C}(H))} \nonumber \\
&=\bnorm{\big(\Id_{S^\infty_k} \ot \tau \big)(x)}_{S^\infty_k} 
\ov{\eqref{Norm-LpL1-pos}}{=}  \norm{x}_{S^\infty_k(\L^1(\scr{C}(H)))}\nonumber. 
\end{align} 
On the other hand, consider an \textit{arbitrary} element $x$ in the space $S^\infty_k(\L^1(\scr{C}(H)))$. Let $\epsi >0$. According to \eqref{formula-Junge}, there exists element $y$ and $z$ such that $x=yz$ and 
\begin{equation}
\label{fin1}
\norm{yy^*}_{S^\infty_k(\L^1(\scr{C}(H))}^{\frac{1}{2}} \norm{z^*z}_{S^\infty_k(\L^1(\scr{C}(H)))}^{\frac{1}{2}} \leq \norm{x}_{S^\infty_k(\L^1(\scr{C}(H)))} +\epsi.
\end{equation}
Now, we have
\begin{align*}
\MoveEqLeft
\norm{(\Id \ot \beta)(y)((\Id \ot \beta)(y))^*}_{S^\infty_k(\L^\infty(\Omega_n,\L^1(\scr{C}(H))))}^{\frac{1}{2}} 
\norm{((\Id \ot \beta)(z))^*(\Id \ot \beta)(z)}_{S^\infty_k(\L^\infty(\Omega_n,(\L^1(\scr{C}(H))))}^{\frac{1}{2}} \\
&=\norm{(\Id \ot \beta)(yy^*)}_{S^\infty_k(\L^\infty(\Omega_n,\L^1(\scr{C}(H))))}^{\frac{1}{2}} 
\norm{(\Id \ot \beta)(z^*z)}_{S^\infty_k(\L^\infty(\Omega_n,\L^1(\scr{C}(H))))}^{\frac{1}{2}} \\      
&\ov{\eqref{divers-566789-bis}}{=} \norm{yy^*}_{S^\infty_k(\L^1(\scr{C}(H)))}^{\frac{1}{2}}\norm{z^*z}_{S^\infty_k(\L^1(\scr{C}(H)))}^{\frac{1}{2}} 
\ov{\eqref{fin1}}{\leq} \norm{x}_{S^\infty_k(\L^1(\scr{C}(H)))} +\epsi.
\end{align*}  
Moreover, it is immediate that 
$$
(\Id \ot \beta)(x)
=(\Id \ot \beta)(yz)
=(\Id \ot \beta)(y)(\Id \ot \beta)(z).
$$ 
From \eqref{formula-Junge}, we infer that $\norm{(\Id \ot \beta)(x)}_{S^\infty_k(\L^\infty(\Omega_n),\L^1(\scr{C}(H))))} \leq \norm{x}_{S^\infty_k(\L^1(\scr{C}(H))} +\epsi$. Since $\epsi >0$ is arbitrary, we obtain that the map $\beta \co \L^1(\scr{C}(H)) \to \L^\infty(\Omega_n,\L^1(\scr{C}(H)))$ is a complete contraction. Now, by Lemma \ref{lemma-trace-preserving} we have a complete isometry $\beta \co \L^1(\scr{C}(H)) \to \L^1(\Omega_n,\L^1(\scr{C}(H)))$. We conclude by interpolation that the map $\beta$ induces a complete contraction $\beta \co \L^1(\scr{C}(H)) \to \L^p(\Omega_n,\L^1(\scr{C}(H)))$. 
%The case of an \textit{arbitrary} element $x$ of $S^\infty_k(\L^1(\Omega_n))$ follows from \eqref{formula-Junge}.
%Finally,  we have for any $x \in \L^\infty(\scr{C}(H))$
%\begin{equation}
%\label{Eq-sans-fin-34}
%\norm{\beta(x)}_{\L^\infty(\Omega_n,\L^p(\scr{C}(H)))}
%\geq \norm{\beta(x)}_{\L^p(\Omega_n,\L^p(\scr{C}(H)))}
%=\norm{x}_{\L^p(\scr{C}(H))}.
%\end{equation}
\end{proof}

Recall that the radial multiplier $R_\phi \co \L^\infty(\scr{C}(H)) \to \L^\infty(\scr{C}(H))$ is defined in \eqref{def-rad-mult} by $R_\phi(s_A) = \phi(\vert A\vert) s_A$, where $A$ is any subset of the set $\{1,\ldots,n\}$. Now, we provide two useful intertwining relations. 

\begin{lemma}
Let $H$ be a finite-dimensional real Hilbert space with dimension $n \geq 2$. For any complex function $\phi \co \{0,\ldots,n\} \to \mathbb{C}$, we have
\begin{equation}
\label{Equation-Delta}
\beta \circ R_\phi
=\big(\Id_{\L^\infty(\Omega_n)} \ot R_\phi\big) \circ  \beta.
\end{equation}

\end{lemma}

\begin{proof}
For any subset $A$ of $\{1,\ldots,n\}$, we have
\begin{align*}
\MoveEqLeft
\beta R_\phi(s_A)         
\ov{\eqref{def-rad-mult}}{=} \phi(\vert A\vert)\beta(s_A) 
\ov{\eqref{actions-alpha-beta}}{=} \phi(\vert A\vert)(w_A \ot s_A)
\end{align*}
and
\begin{align*}
\MoveEqLeft
(\Id_{} \ot R_\phi) \beta(s_A)         
\ov{\eqref{actions-alpha-beta}}{=} (\Id_{} \ot R_\phi) (w_A \ot s_A) 
=w_A \ot R_\phi(s_A)
\ov{\eqref{def-rad-mult}}{=} \phi(\vert A\vert)(w_A \ot s_A).
\end{align*}
\end{proof}

The second intertwining relation uses the definition of a radial multiplier $T_\phi \co \L^\infty(\Omega_n) \to \L^\infty(\Omega_n)$, $w_A \mapsto \phi(|A|) w_A$ given in \eqref{def-mult-rad}.

\begin{lemma}
\label{lemma-commute-mult-Walsh}
Let $H$ be a finite-dimensional real Hilbert space with dimension $n \geq 2$. For any complex function $\phi \co \{0,\ldots,n\} \to \mathbb{C}$, we have
\begin{equation}
\label{Equation-Delta-op}
\alpha \circ R_\phi
=\big(\Id_{\scr{C}(H)} \ot T_\phi\big) \circ \alpha.
\end{equation}
\end{lemma}

\begin{proof}
For any subset $A$ of the set $\{1,\ldots,n\}$, we have
\begin{align*}
\MoveEqLeft
\alpha \circ R_\phi(s_A)         
\ov{\eqref{def-rad-mult}}{=} \phi(\vert A\vert)\alpha(s_A) 
 \ov{\eqref{actions-alpha-beta}}{=} \phi(\vert A\vert)(s_A \ot w_A)
\end{align*}
and
\begin{align*}
\MoveEqLeft
(\Id_{} \ot T_\phi) \alpha(s_A)         
\ov{\eqref{actions-alpha-beta}}{=} (\Id_{} \ot T_\phi) (s_A \ot w_A) 
=s_A \ot T_\phi(w_A)
\ov{\eqref{def-mult-rad}}{=} \phi(\vert A\vert)(s_A \ot w_A).
\end{align*}
\end{proof}

\paragraph{A useful $*$-homomorphism} Let $H$ be a finite-dimensional real Hilbert space with dimension $n \geq 2$. We equally consider the linear map $\eta \co \L^\infty(\Omega_n) \to \scr{C}(H) \otvn \scr{C}(H)$ defined by
\begin{equation}
\label{autre-eta}
\eta(w_A)
\ov{\mathrm{def}}{=} s_A \ot s_A,\quad A \subset \{1,\ldots,n\}.
\end{equation}
As the following result shows, this map allows the intertwining of radial multipliers on the $n$-dimensional discrete hypercubes and the radial multipliers on the fermion algebras.  

\begin{prop}
\label{Prop-eta}
Let $H$ be a finite-dimensional real Hilbert space with dimension $n \geq 2$. Suppose that $1 \leq p \leq \infty$. The map $\eta$ is a trace preserving unital injective $*$-homomorphism and induces a complete contraction $\eta \co \L^1(\Omega_n) \to \L^p(\scr{C}(H),\L^1(\scr{C}(H)))$. Moreover, for any complex function $\phi \co \{0,\ldots,n\} \to \mathbb{C}$ we have
\begin{equation}
\label{commute-bis}
(\Id_{\scr{C}(H)} \ot R_\phi) \circ \eta
=\eta \circ T_\phi.
\end{equation}
\end{prop}

\begin{proof}
For any subsets $A$ and $B$ of the set $\{1,\ldots,n\}$, we have
\begin{align*}
\MoveEqLeft
\eta(w_Aw_B)
\ov{\eqref{prod-walsh}}{=}  \eta(w_{A \Delta B})
\ov{\eqref{autre-eta}}{=}  s_{A \Delta B} \ot s_{A \Delta B}
\ov{\eqref{product-sA}}{=} s_A s_B \ot s_A s_B \\
&=(s_A \ot s_A)(s_B \ot s_B)
\ov{\eqref{autre-eta}}{=} \eta(w_A)\eta(w_B)
\end{align*}
and
$$
\eta(w_A^*)
= \eta(w_A)
\ov{\eqref{autre-eta}}{=} s_A \ot s_A
\ov{\eqref{adjoint-sA}}{=} (s_A \ot s_A)^*
\ov{\eqref{autre-eta}}{=} \eta(w_A)^*.
$$
We deduce that $\eta$ is a $*$-homomorphism, which is clearly unital. We also have
\begin{align*}
\MoveEqLeft
(\tau \ot \tau)(\eta(w_A))
\ov{\eqref{autre-eta}}{=} (\tau \ot \tau)(s_A \ot s_A) 
=\tau(s_A)^2
\ov{\eqref{trace-fer}}{=} \delta_{A,\emptyset}
\ov{\eqref{int-walsh}}{=} \int_{\Omega_n} w_A.         
\end{align*}
Consequently, the map $\eta$ is trace preserving. The injectivity follows from \cite[Lemma 2.1 p.~2283]{Arh19}, but it is also possible to establish it through a direct proof. 

It is proved in \cite[1.2.6 p.~6]{BLM04} that any bounded linear map from an operator space into a \textit{commutative} $\C^*$-algebra is completely bounded and that its completely bounded norm coincides with its norm. By duality, it is therefore sufficient to prove that the map $\eta \co \L^1(\Omega_n) \to \L^p(\scr{C}(H),\L^1(\scr{C}(H)))$ is a contraction in order to obtain complete contractivity. Now, we use a proof similar to the one of Proposition \ref{prop-autre-extension}. It is easy to check that
\begin{equation}
\label{ergo-biz}
(\Id_{\scr{C}(H)} \ot \tau) \circ \eta(f)
=\bigg(\int_{\Omega_n} f \bigg) 1_{\scr{C}(H)}, \quad f \in \L^\infty(\Omega_n).
\end{equation}
Indeed, for any subset $A$ of $\{1,\ldots,n\}$, we have
\begin{align*}
\MoveEqLeft
(\Id_{\scr{C}(H)} \ot \tau) \circ \eta(w_A)
\ov{\eqref{autre-eta}}{=}(\Id_{\scr{C}(H)} \ot \tau) (s_A \ot s_A) 
=\tau(s_A) s_A \\
&\ov{\eqref{trace-fer}}{=} \delta_{A,\emptyset}s_A 
=\delta_{A,\emptyset} 1_{\scr{C}(H)}
\ov{\eqref{int-walsh}}{=} \bigg(\int_{\Omega_n} \omega_A \bigg) 1_{\scr{C}(H)}.
\end{align*}
In the case of a \textit{positive} function $f$ in the Banach space $\L^1(\Omega_n)$, we can write
\begin{align}
\MoveEqLeft
\label{divers-566789}
\norm{\eta(f)}_{\L^\infty(\scr{C}(H),\L^1(\scr{C}(H)))}              
\ov{\eqref{Norm-LpL1-pos}}{=} \bnorm{(\Id_{\L^\infty(\scr{C}(H))} \ot \tau)(\eta(f))}_{\L^\infty(\scr{C}(H))} \\
%&\ov{\eqref{ergo-biz}}{=} \bnorm{( (\smallint_{} \cdot)1_{\scr{C}(H)})(x) }_{\scr{C}(H)} 
&\ov{\eqref{ergo-biz}}{=} \norm{\bigg(\int_{\Omega_n}f\bigg) 1_{\scr{C}(H)}}_{\scr{C}(H)} 
=\int_{\Omega_n} f
= \norm{f}_{\L^1(\Omega_n)}\nonumber. 
\end{align}
Now, consider an \textit{arbitrary} function $f$ in the Banach space $\L^1(\Omega_n)$. Using \eqref{inverse-Holder}, there exist functions $g$ and $h$ such that $f=gh$ and 
\begin{equation}
\label{sansfin-2}
\norm{g}_{\L^2(\Omega_n)}\norm{h}_{\L^2(\Omega_n)}
= \norm{f}_{\L^1(\Omega_n)}.
\end{equation}
We obtain
\begin{align*}
\MoveEqLeft
\norm{\eta(g)(\eta(g))^*}_{\L^\infty(\scr{C}(H),\L^1(\scr{C}(H)))}^{\frac{1}{2}} 
\norm{(\eta(h))^* \eta(h)}_{\L^\infty(\scr{C}(H),(\L^1(\scr{C}(H))))}^{\frac{1}{2}} \\
&=\norm{\eta(gg^*)}_{\L^\infty(\scr{C}(H),\L^1(\scr{C}(H)))}^{\frac{1}{2}} 
\norm{\eta(h^*h)}_{\L^\infty(\scr{C}(H),\L^1(\scr{C}(H)))}^{\frac{1}{2}} \\      
&\ov{\eqref{divers-566789}}{=} \norm{gg^*}_{\L^1(\Omega_n)}^{\frac{1}{2}}\norm{h^* h}_{\L^1(\Omega_n)}^{\frac{1}{2}} 
\ov{\eqref{sansfin-2}}{=} \norm{x}_{\L^1(\Omega_n)}.
\end{align*}  
Moreover, we have the factorization $\eta(f)=\eta(gh)=\eta(g)\eta(h)$. From \eqref{formula-Junge}, we obtain the inequality $\norm{\eta(f)}_{\L^\infty(\scr{C}(H),\L^1(\scr{C}(H)))} \leq \norm{f}_{\L^1(\Omega_n)}$. Observe that by Lemma \ref{lemma-trace-preserving} we have a (complete) isometry $\eta \co \L^1(\Omega_n) \to \L^1(\scr{C}(H),\L^1(\scr{C}(H)))$. We conclude by interpolation that the map $\eta$ induces a  contraction $\eta \co \L^1(\Omega_n) \to \L^p(\scr{C}(H),\L^1(\scr{C}(H)))$. 

Finally, we prove the equality \eqref{commute-bis}. On the one hand, for any finite subset $A$ of the set $\{1,\ldots,n\}$, we have
\begin{align*}
\MoveEqLeft
(\Id \ot R_\phi)\eta(w_A)         
\ov{\eqref{autre-eta}}{=} (\Id \ot R_\phi)(s_A \ot s_A)
=s_A \ot R_\phi(s_A)
\ov{\eqref{def-rad-mult}}{=} \phi(|A|)s_A \ot s_A.
\end{align*}
On the other hand, we have
\begin{align*}
\MoveEqLeft
\eta T_\phi(w_A)         
\ov{\eqref{def-mult-rad}}{=} \phi(|A|)\eta(w_A) 
\ov{\eqref{autre-eta}}{=} \phi(|A|)s_A \ot s_A.
\end{align*}
\end{proof}

\begin{remark} \normalfont
The relatively straightforward tricks outlined in this section will also find applications in the realm of semigroups of operators, particularly in the quest to determine the optimal angles of the functional calculus of generators, see the forthcoming preprint and companion paper \cite{ArK24} where we will adapt some of the previous statements to the infinite-dimensional case $n=\infty$. Furthermore, these methods are expected to be valuable in various other contexts. Further connections between the Walsh system and the fermion system have been explored within the context of Riesz transforms, as detailed in the classical paper \cite{Lus98} of Lust-Piquard.
\end{remark}

%%%%%%%%%%%%%%%%%%%%%%%%%%%%%%%%%%%%%%%%%%%%%%%%%%%%%%%%%%%%%%%%%%%%%%%%%%%%%%%%%%%%%%%%%%%%%
\section{Completely bounded norms $\L^1$ to $\L^p$ of radial multipliers on fermion algebras}
\label{sec-completely-bounded-multipliers}

%Let $G$ be a locally compact group equipped with a left Haar measure. Recall that a continuous complex function $\varphi \co G \to \mathbb{C}$ is said to be positive definite if for any elements $s_1,\ldots,s_n$ of $G$, the matrix $[\varphi(s_i^{-1}s_j)]_{1 \leq i,j \leq n}$ is positive hermitian, see \cite[Definition 13.4.1 p.~286]{Dix77}.
%%for any continuous function with compact support $\xi \co G \to \mathbb{C}$ we have $\int_G\int_G\varphi(t^{-1}s) \ovl{\xi(t)} \xi(s) \geq 0$, see \cite[Proposition 13.4.4 p.~286]{Dix77}. 
%Such a function satisfies $\varphi(s^{-1})=\ovl{\varphi(s)}$ and $|\varphi(s)| \leq \varphi(e)$ for any $s \in G$. 

In the next result, we characterize completely positive radial multipliers on fermion algebras.

\begin{prop}
\label{Prop-cp}
Let $H$ be a finite-dimensional real Hilbert space with dimension $n \geq 2$ and let $\phi \co \{0,\ldots,n\} \to \mathbb{C}$ be a complex function. Then the following properties are equivalent.
\begin{enumerate}
	\item The radial multiplier $R_\phi \co \scr{C}(H) \to \scr{C}(H)$ is completely positive.
	%\item The radial multiplier $R_\phi \co \L^1(\scr{C}(H)) \to \L^1(\scr{C}(H))$ is positive.
	\item The complex function $f_\phi \co \Omega_n \to \mathbb{C}$ of \eqref{def-hg} is positive on the compact abelian group $\Omega_n$.
\end{enumerate} 
\end{prop}

\begin{proof}
%We only prove the case where $n <\infty$. The case $n=\infty$ is true with minor modifications.
1 $\Rightarrow$ 2: Suppose that the radial multiplier $R_\phi \co \scr{C}(H) \to \scr{C}(H)$ is completely positive. We denote by $\E \co \scr{C}(H) \otvn \scr{C}(H) \to \L^\infty(\Omega_n)$ the canonical trace preserving faithful conditional expectation associated to the trace preserving injective unital $*$-homomorphism $\eta \co \L^\infty(\Omega_n) \to \scr{C}(H) \otvn \scr{C}(H)$, provided by Proposition \ref{prop-existence-conditional-expectation}. Using the equality $\E \circ \eta=\Id$, we see that
$$
T_\phi
=\E \circ \eta \circ T_\phi
\ov{\eqref{commute-bis}}{=} \E \circ (\Id \ot R_\phi) \circ \eta.
$$ 
As a composition of completely positive maps, we deduce that the radial multiplier $T_\phi \co \L^\infty(\Omega_n) \to \L^\infty(\Omega_n)$ is (completely) positive. Recall that this operator can be seen as the convolution operator by the measure $\nu \ov{\mathrm{def}}{=} f_\phi \cdot \mu$, where the function $f_\phi \co \Omega_n \to \mathbb{C}$ is defined in \eqref{def-hg}. By Theorem \ref{thm-positivity}, we infer that the measure $\nu$ is positive. Consequently, the function $f_\phi \co \Omega_n \to \mathbb{C}$ is also positive.

%\cite[Proposition 5.4.9 p.~184]{KaL18}

2 $\Rightarrow$ 1: Suppose that the complex function $f_\phi \co \Omega_n \to \mathbb{C}$ is positive. Recall again that the linear operator $T_\phi \co \L^\infty(\Omega_n) \to \L^\infty(\Omega_n)$ is the convolution operator by the measure $\nu \ov{\mathrm{def}}{=} f_\phi \cdot \mu$. Since this measure is positive, we deduce by Theorem \ref{thm-positivity} that the Fourier multiplier $T_\phi \co \L^\infty(\Omega_n) \to \L^\infty(\Omega_n)$ is positive and even completely positive by \cite[Corollary 3.5 p.~194]{Tak02}, since $\L^\infty(\Omega_n)$ is a \textit{commutative} $\C^*$-algebra. Now, we consider the canonical trace preserving faithful conditional expectation $\E \co \scr{C}(H) \otvn \L^\infty(\Omega_n) \to \scr{C}(H)$ associated to the trace preserving injective unital $*$-homomorphism $\alpha \co \scr{C}(H) \to \scr{C}(H) \otvn \L^\infty(\Omega_n)$, provided by Proposition \ref{prop-existence-conditional-expectation}. With the equality $\E \circ \alpha=\Id$, we infer that
$$
R_\phi
=\E \circ \alpha \circ R_\phi
\ov{\eqref{Equation-Delta-op}}{=} \E \circ \big(\Id_{\scr{C}(H)} \ot T_\phi\big) \circ \alpha.
$$
We obtain by composition that the radial multiplier $R_\phi \co \scr{C}(H)) \to \scr{C}(H)$ is completely positive.
\end{proof}

%In this case, the Fourier transform $\hat{\mu}$ is the unique function $\varphi \in \L^\infty(\hat{G})$ such that
%\begin{equation}
%\label{}
%\widehat{T(f)}
%=\varphi \hat{f}, \quad f \in \L^2(G) \cap \L^p(G).
%\end{equation} 
%holds.
%
 %A continuous function is positive definite if and only if there exists a bounded positive measure such that $f=\hat{\mu}$. In this case the measure is unique. Bochner's theorem

\begin{thm}
\label{th-norm-cb-rad}
Let $H$ be a finite-dimensional real Hilbert space with dimension $n \geq 2$. Suppose that $1 < p \leq \infty$. For any function $\phi \co \{0,\ldots,n\} \to \mathbb{C}$, we have
\begin{equation}
\label{norm-cb}
\norm{R_\phi}_{\cb, \L^1(\scr{C}(H)) \to \L^p(\scr{C}(H))}
= \norm{f_\phi}_{\L^p(\Omega_n)}
\end{equation}
and
\begin{equation}
\label{norm-cb-2}
\norm{R_\phi^*}_{\cb, \L^{p^*}(\scr{C}(H)) \to \L^\infty(\scr{C}(H))}
=\norm{f_\phi}_{\L^{p}(\Omega_n)}.
\end{equation}
\end{thm}

\begin{proof}
The second formula is an immediate consequence of the first formula by duality. So, it suffices to prove the first one. Note that by Proposition \ref{Prop-transfer-1}, we have
\begin{equation}
\label{Ine-div-4457}
\norm{\alpha}_{\cb, \L^1(\scr{C}(H)) \to \L^p(\scr{C}(H),\L^1(\Omega_n))}
\leq 1.
\end{equation}
Recall that the linear map $\alpha$ is a trace preserving unital $*$-homomorphism. We denote by $\E \co \scr{C}(H) \otvn \L^\infty(\Omega_n) \to \scr{C}(H)$ the associated canonical trace preserving faithful conditional expectation, provided by Proposition \ref{prop-existence-conditional-expectation}. Thus, by Lemma \ref{lemma-trace-preserving} the linear map $\alpha$ induces a complete isometry $\alpha \co \L^p(\scr{C}(H)) \to \L^p(\scr{C}(H) \otvn \L^\infty(\Omega_n))$ and we have a complete contraction $\E \co \L^p(\scr{C}(H) \otvn \L^\infty(\Omega_n)) \to \L^p(\scr{C}(H))$. On the one hand, using $\E \alpha=\Id_{\L^p(\scr{C}(H))}$ in the first equality, we deduce that
\begin{align*}
\MoveEqLeft
\norm{R_\phi}_{\cb,\L^1(\scr{C}(H)) \to \L^p(\scr{C}(H))}
=\norm{\alpha R_\phi}_{\cb,\L^1(\scr{C}(H)) \to \L^p(\scr{C}(H) \otvn \L^\infty(\Omega_n))} \\
&\ov{\eqref{Equation-Delta-op}}{=} \norm{(\Id_{} \ot T_\phi ) \alpha}_{\cb,\L^1(\scr{C}(H)) \to \L^p(\scr{C}(H) \otvn \L^\infty(\Omega_n))} \\     
&\leq \norm{\Id \ot T_\phi}_{\cb,\L^p(\scr{C}(H),\L^1(\Omega_n)) \to \L^p(\scr{C}(H) \otvn \L^\infty(\Omega_n))} \norm{\alpha}_{\cb,\L^1(\scr{C}(H)) \to \L^p(\scr{C}(H),\L^1(\Omega_n))} \\
&\ov{\eqref{ine-tensorisation-os}\eqref{Ine-div-4457}}{\leq} \norm{T_\phi}_{\cb, \L^1(\Omega_n) \to \L^p(\Omega_n)}
\ov{\eqref{norm-equality-QG}}{=} \norm{f_\phi}_{\L^p(\Omega_n)}.
\end{align*}
On the other hand, using the trace preserving unital injective $*$-homomorphism $\eta \co \L^\infty(\Omega_n) \to \scr{C}(H) \otvn \scr{C}(H)$ defined in \eqref{autre-eta} and a similar reasoning using its associated canonical trace preserving faithful conditional expectation $\E \co \scr{C}(H) \otvn \scr{C}(H) \to \L^\infty(\Omega_n)$, which satisfies $\E\eta=\Id$, we see that
\begin{align*}
\MoveEqLeft
 \norm{f_\phi}_{\L^p(\Omega_n)}
\ov{\eqref{norm-equality-QG}}{=}\norm{T_\phi}_{\L^1(\Omega_n) \to \L^p(\Omega_n)}         
=\norm{\eta T_\phi}_{\L^1(\Omega_n) \to \L^p(\scr{C}(H) \otvn \scr{C}(H))}      \\
&\ov{\eqref{commute-bis}}{=} \norm{(\Id \ot R_\phi)\eta}_{\L^1(\Omega_n) \to \L^p(\scr{C}(H) \otvn \scr{C}(H))} \\
&\leq \norm{\Id_{} \ot R_\phi}_{\L^p(\scr{C}(H),\L^1(\scr{C}(H))) \to \L^p(\scr{C}(H) \otvn \scr{C}(H))} \norm{\eta}_{\L^1(\Omega_n) \to \L^p(\scr{C}(H),\L^1(\scr{C}(H)))} \\
&\ov{\eqref{ine-tensorisation-os}}{\leq}  \norm{R_\phi}_{\cb, \L^1(\scr{C}(H)) \to \L^p(\scr{C}(H))},
\end{align*}
where we used Proposition \ref{Prop-eta} in the last inequality.
\end{proof}

\begin{remark} \normalfont
Let $H$ be a finite-dimensional real Hilbert space with \textit{even} dimension $n=2k \geq 2$. If we consider the constant function $\phi \co \{0,\ldots,n\} \to \mathbb{C}$, $k \mapsto 1$, the associated radial multiplier is the operator identity $\Id \co \L^1(\scr{C}(H)) \to \L^1(\scr{C}(H))$, which identifies to the identity map $\Id_{S^1_{N}} \co S^{1}_N \to S^1_N$ with $N \ov{\mathrm{def}}{=} 2^k$. Suppose that $1 \leq p \leq \infty$. It is easy to check that the identity map $\Id \co S^1_N \to S^p_N$ satisfies 
$$
\norm{\Id}_{\cb,S^1_N \to S^p_N} 
= N^{1-\frac{1}{p}} 
\not= 1
= \norm{\Id}_{S^1_N \to S^p_N}.
$$
We may note along the way that $\H_{\cb,\min,\tr}(\Id)=-\log_2(N)$ and $\H_{\min,\tr}(\Id)=0$. Consequently, for a radial multiplier $R_\phi$ we have 
$$
\norm{R_\phi}_{\cb, \L^1(\scr{C}(H)) \to \L^p(\scr{C}(H))} \not= \norm{R_\phi}_{ \L^1(\scr{C}(H)) \to \L^p(\scr{C}(H))}
$$ 
in general.
\end{remark}

%%%%%%%%%%%%%%%%%%%%%%%%%%%%%%%%%%%%%%%%%%%%%%%%%%%%%%%%%%%%%%%%%%%%%%%%%%%%%%%%%%%%%%%%%%%%%
\section{Completely bounded minimal output entropy of the radial multipliers}
\label{entropy-cb-radial-multipliers}

Now, we can compute the completely bounded minimal output entropy of any radial multiplier which defines a quantum channel. By remark \ref{remark-trace-preserving} and Proposition \ref{Prop-cp}, this assumption means that $\phi(0)=1$ and that the complex function $f_\phi \co \Omega_n \to \mathbb{C}$ is positive.

\begin{cor}
\label{Cor-ent-cb}
Let $H$ be a finite-dimensional real Hilbert space with dimension $n \geq 2$. Consider a function $\phi \co \{0,\ldots,n\} \to \mathbb{C}$. Suppose that the radial multiplier $R_\phi \co \scr{C}(H) \to \scr{C}(H)$ is a quantum channel. We have
\begin{equation}
\label{entropy-cb}
\H_{\cb,\min,\tau}(R_\phi)
=\H(f_\phi),
\end{equation}
where $\H(f_\phi)=-\int_{\Omega_n} f_\phi\log_2 f_\phi$ is the Segal entropy of $f_\phi$ defined in \eqref{Segal-entropy} with the normalized integral $\int_{\Omega_n}$ and where $\tau$ is the normalized trace defined in \eqref{trace-fer}.
\end{cor}

\begin{proof}
First note that $f_\phi \geq 0$ by Proposition \ref{Prop-cp} and consequently
\begin{align}
\MoveEqLeft
\label{norm-hg}
\norm{f_\phi}_{\L^1(\Omega_n)}
= \int_{\Omega_n} f_\phi      
\ov{\eqref{def-hg}}{=} \int_{\Omega_n} \sum_A \phi(|A|) w_A 
\ov{\eqref{def-hg}}{=} \sum_A \phi(|A|) \int_{\Omega_n} w_A
=\phi(0)
=1.
\end{align}
We will use Theorem \ref{th-norm-cb-rad}. Using the equality $\norm{f_\phi}_{\L^1(\Omega_n)}=1$ in the third equality, we obtain
\begin{align*}
\MoveEqLeft
\H_{\cb,\min,\tau}(R_\phi)
\ov{\eqref{Def-intro-Scb-min}}{=} -\frac{1}{\log 2}\frac{\d}{\d p} \big[\norm{R_\phi}_{\cb,\L^1(\scr{C}(H)) \to \L^p(\scr{C}(H))}\big]|_{p=1} \\
&\ov{\eqref{norm-cb}}{=} -\frac{1}{\log 2}\frac{\d}{\d p} \big[\norm{f_\phi}_{\L^p(\Omega_n)}\big] |_{p=1} 
\ov{\eqref{deriv-norm-p}}{=} \H(f_\phi).%-\int_{\Omega_n} f_\phi\log_2 h_\phi
%\ov{\eqref{Segal-entropy}}{=}
\end{align*}
\end{proof}

\begin{remark} \normalfont
\label{rem-cbmin-est-max}
We warn the reader that the trace used on the algebra $\scr{C}(H)$ in the previous result for defining the completely bounded minimal output entropy is the \textit{normalized} one, defined in \eqref{trace-fer}. In particular, we have the inequality $\H_{\cb,\min,\tau}(R_\phi) \leq 0$ by \eqref{cbminless0}. If $n=2k$ is even and if we identify the fermion algebra $\scr{C}(H)$ with the matrix algebra $\M_{N}$ with $N=2^k$, we replace the \textit{normalized} trace $\tau$ of \eqref{trace-fer} by the \textit{non-normalized} trace $\tr$. So the completely bounded minimal output entropy of the quantum channel $R_\phi \co \M_N \to \M_N$ becomes 
\begin{equation}
\label{Hcb-non-normalized}
\H_{\cb,\min,\tr}(R_\phi)
=\H(f_\phi)+\log_2 {N}
\end{equation}
for this new trace. Moreover, by \cite[pp.~10-11]{AlC18} combined with \eqref{changement-de-trace} we have $$
-2\log_2N
=-\log_2 2^n \leq \H(f_\phi) 
\leq 0.
$$ 
We recover the inequalities
$$
-\log_2 {N}
\leq \H_{\cb,\min,\tr}(R_\phi) 
\leq \log_2 {N}.
$$
\end{remark}

\begin{example} \normalfont
If $n=2k$ is even and $N=2^k$, then the completely bounded minimal output entropy is minimal, i.e.~we have the equality $\H_{\cb,\min,\tr}(R_\phi)=-\log_2 {N}$, if and only if $\H(f_\phi)=-2\log_2 {N}$. By \cite[Lemma 2.5 p.~10]{AlC18} and \eqref{changement-de-trace}, this amounts to saying that the function $f_\phi \co \Omega_n \to \mathbb{C}$ is constant equal to 1 since $\phi(0)=1$ by Remark \ref{remark-trace-preserving}. This is equivalent to $R_\phi$ being the identity channel.
\end{example}

\begin{example} \normalfont
By the equivalence \eqref{carac-entropy-max} and the previous equality \eqref{entropy-cb}, the completely bounded minimal output entropy is maximal, i.e.~$\H_{\cb,\min,\tau}(R_\phi)=0$, if and only if $f_\phi=1$. This is equivalent to $\phi=\delta_0$, i.e.~$\phi(0)=1$ and $\phi(k)=0$ if $k \geq 1$. In this case, the associated quantum channel is the completely noisy channel.
\end{example}

\begin{remark} \normalfont
Note that we have a trace preserving $*$-isomorphism $\scr{C}(H) \otvn \scr{C}(H)=\scr{C}(H \ot_2 H)$ and that a tensor product of radial multipliers identifies to a radial multiplier acting on $\scr{C}(H \ot_2 H)$. This allows us the recover with formula \eqref{entropy-cb} the additivity \eqref{additivity-cb} of the completely bounded minimal output entropy on radial multipliers since the Segal entropy $\H$ is additive by \eqref{entropy-tensor-product}.
\end{remark}

\begin{remark} \normalfont%des infos dans Mokshay Madiman Entropy and the Additive Combinatorics of Probability Densities on LCA groups,SHANNON’S THEOREM FOR LOCALLY COMPACT GROUPS;Derriennic. Entropie, th´eor`emes limite et marches al´eatoires
%using the normalized Haar measure $\mu_{\Omega_n}$ on the abelian finite group $\Omega_n$.
%Maximum Entropy on Compact Groups Peter Harremo¨es
%(Kullback-Leibler divergence)
With \cite[Proposition 2.3]{LoW22}, we can express the entropy $\H(f_\phi)$ with the help of the relative entropy. We obtain the formula
$$
\H_{\cb,\min,\tau}(R_\phi)
=-\H(f_\phi\mu \| \mu),
$$  
where $f_\phi\mu$ is the probability measure on the abelian group $\Omega_n$ defined by the density $f_\phi$.
%Let $G$ be a compact abelian group equipped with its normalized Haar measure $\mu_G$. Recall that the entropy of a probability measure $\mu$ on $G$ which is absolutely continuous with respect to the Haar measure $\mu_G$ with density $f$ is defined by 
%$$
%h(\mu) \ov{\mathrm{def}}{=} -\int_G f(s)\log_2 f(s) \d\mu_G(s),
%$$ 
%see e.g.~\cite[p.~13]{Der85} for the more general case of second countable locally compact groups and \cite[Definition 1.1 p.~463]{KaV83} for the case of countable discrete groups.
%If $\mathrm{D}$ denotes the relative entropy (Kullback-Leibler divergence) \cite[p.~19]{CoT06} or \cite[Definition 2.29 p.~26]{AlC18}, a simple observation \cite[(2.93) p.~29]{CoT06} combined with \eqref{changement-de-trace} imply that
%$$
%\H(\mu)
%=-\mathrm{D}(\nu \| \mu_G) \leq 0.
%$$ 
%The Segal entropy $\H(h_g)$ can be seen as the entropy of the probability measure $h_g\mu_{\Omega_n}$ on the abelian group $\Omega_n$ and we obtain the formula
%$$
%\H_{\cb,\min,\tau}(R_g)
%=-\mathrm{D}(h_g\mu_{\Omega_n} \| \mu_{\Omega_n})
%$$  
%which is quite intuitive.
%See e.g.~\cite[Definition 1.1 p.~463]{KaV83} for this notion defined for at most countable discrete groups.
\end{remark}
%$\H(\mu)$ may or may not exist; if it does, it takes values in the extended real line $[-\infty,+\infty]$.

Our result allows us to obtain a lower bound of the channel coherent information defined in \eqref{def-Q1}.

\begin{cor}
\label{Cor-Q1}
Let $H$ be a finite-dimensional real Hilbert space with even dimension $n=2k \geq 2$. Consider a function $\phi \co \{0,\ldots,n\} \to \mathbb{C}$. Suppose that the radial multiplier $R_\phi \co \scr{C}(H) \to \scr{C}(H)$ is a quantum channel. We have
\begin{equation}
\label{mino-Q1}
\Q^{(1)}(R_\phi)
\geq \max\{-\log_2 {N}-\H(f_\phi),0\}.
\end{equation}
where $N \ov{\mathrm{def}}{=} 2^k$.
\end{cor}

\begin{proof}
Recall that we have $\Q^{(1)}(T) \geq -\H_{\cb,\min,\tr}(T)$ for any \textit{unital} quantum channel $T \co S^1_N \to S^1_N$. Hence the result is a consequence of the equality \eqref{Hcb-non-normalized}.
\end{proof}

\begin{example} \normalfont
\label{Example-suite}
Consider the case $n=2$ and the situation of Example \ref{Ex-multiplier-Fermions}. Here $N=2$. For any $0 \leq t \leq 1$, consider the function $\phi \co \{0,1,2\} \to \mathbb{C}$ defined by $\phi(0)=1$, $\phi(1)=1-2t$ and $\phi(2)=1$.
We have seen that the radial multiplier $R_\phi$ identifies to the dephasing channel $T_t \co \M_2 \to \M_2$, $x \mapsto (1-t)x+tZxZ$. We have
$$
f_\phi
\ov{\eqref{def-hg}}{=} 1+(1-2t)\epsi_1+(1-2t)\epsi_2+\epsi_{\{1,2\}}.
$$
In particular, we have $f_\phi(-1,-1)=4t$, $f_\phi(-1,1)=0$, $f_\phi(1,-1)=0$, $f_\phi(1,1)=4-4t$. We deduce that
\begin{align}
\MoveEqLeft
\label{sans-fin-5}
-\H(f_\phi)
\ov{\eqref{Segal-entropy}}{=} \int_{\Omega_2} f_\phi\log_2 f_\phi
\ov{\eqref{measure-hypercube}}{=} \frac{1}{4}\sum_{\epsi \in \{-1,1\}^2} f_\phi(\epsi)\log_2 f_\phi(\epsi) \\
&=\frac{1}{4}\big[  4t\log_2(4t) +  (4-4t)\log_2(4-4t) \big]
=t\log_2(4t) +  (1-t)\log_2(4-4t) \nonumber \\
&=2+t\log_2 t+(1-t)\log_2(1-t). \nonumber
\end{align}
%If we consider the radial multiplier $R_g$ as a quantum channel $R_g \co S^1_2 \to S^1_2$, we have to replace the normalized trace $\tau$ by the classical trace $\tr$ on the matrix algebra $\M_2$. 
We deduce the formula
\begin{equation}
\label{Hcbmin-dephasing}
-\H_{\cb,\min,\tr}(T_t)
\ov{\eqref{Hcb-non-normalized}}{=} -\H(f_\phi)-\log_2 N
=1+t\log_2 t+( 1-t)\log_2(1-t).
\end{equation}
We can recover this equation from the formula
\begin{equation}
\label{Hcbmin-HS}
-\H_{\cb,\min,\tr}\big(M_\varphi^\HS\big)
=\log_2 |G|-\H_{\tr}\bigg(\frac{1}{|G|}C_\varphi\bigg) 
\geq 0
\end{equation}  
of the paper \cite[p.~114]{GJL18a} applicable to all Herz-Schur multiplier 
$$
M_\varphi^\HS \co S^1_G \to S^1_G, x \mapsto [\varphi(sr^{-1})x_{sr}]_{s,r \in G},
$$ 
which are quantum channels. Here $G$ is a finite group with cardinal $|G|$, $\varphi \co G \to \mathbb{C}$ is a positive definite function with $\varphi(e)=1$ and $C_\varphi$ is the matrix $[\varphi(sr^{-1})]_{s,r \in G}$.

With $G=\Z/2\Z=\{\ovl{0},\ovl{1}\}$, we have seen in Example \ref{Ex-multiplier-Fermions} that the dephasing channel $T_t \co \M_2 \to \M_2$ identifies to the Herz-Schur multiplier of symbol $C_\varphi=\begin{bmatrix}
   1  &  1-2t \\
   1-2t  &  1 \\
\end{bmatrix}$, where $\varphi \co G \to \mathbb{C}$ is defined by $\varphi(\ovl{0})=1$ and $\varphi(\ovl{1})=1-2t$. The \textit{opposite} of the von Neumann entropy of the matrix 
$$
\frac{1}{|G|}C_\varphi=\begin{bmatrix}
   \frac{1}{2}  &  \frac{1}{2}-t \\
   \frac{1}{2}-t  &  \frac{1}{2} \\
\end{bmatrix}
$$ 
is\footnote{\thefootnote. The eigenvalues of a matrix $\begin{bmatrix}
   b  &  a \\
   a &  b \\
\end{bmatrix}$ are $b-a$ and $b+a$.} $t\log_2 t+\big( 1-t\big)\log_2\big(1-t\big)$. So we recover the formula \eqref{Hcbmin-dephasing}. With Corollary \ref{Cor-Q1}, we obtain the inequality 
$$
\Q(T_t) \geq \Q^{(1)}(T_t) \geq 1+t\log_2 t+(1-t)\log_2(1-t).
$$ 
Indeed, it is known by \cite[Exercise 24.7.1 p.~666]{Wil17} that the quantum capacity of the quantum channel $T_t$ is exactly equal to
%%we are able with Theorem \ref{Th-CEA-Herz-Schur} to recover
\begin{equation}
\label{quantum-cap-dephasing}
\Q(T_t)
=1+t\log_2 t+(1-t)\log_2(1-t).
\end{equation}
\end{example} 

\begin{example} \normalfont
\label{Example-suite-2}
Consider the situation of Example \ref{Ex-Ornstein} with $n=2$. For any $t \geq 0$, we have
$$
f_{\phi_t}
\ov{\eqref{def-hg}}{=} 1+\e^{-t}\epsi_1+\e^{-t}\epsi_2+\e^{-2t}\epsi_{\{1,2\}}.
$$
We have $f_{\phi_t}(-1,-1)=1-2\e^{-t}+\e^{-2t}=(1-\e^{-t})^2$, $f_{\phi_t}(-1,1)=1-\e^{-2t}$, $f_{\phi_t}(1,-1)=1-\e^{-2t}$, $f_{\phi_t}(1,1)=1+2\e^{-t}+\e^{-2t}=(1+\e^{-t})^2$. We deduce that
\begin{align*}
\MoveEqLeft
%\label{sans-fin-99}
-\H(f_{\phi_t})
\ov{\eqref{Segal-entropy}}{=} \int_{\Omega_2} f_{\phi_t}\log h_{g_t}
\ov{\eqref{measure-hypercube}}{=} \frac{1}{4}\sum_{\epsi \in \{-1,1\}^2} f_{\phi_t}(\epsi)\log_2 f_{\phi_t}(\epsi) \\
&=\frac{1}{4}\big[ (1-\e^{-t})^2\log_2((1-\e^{-t})^2) +2(1-\e^{-2t})\log_2(1-\e^{-2t})+(1+\e^{-t})^2\log_2((1+\e^{-t})^2) \big] \\
&=\frac{1}{2}\big[ (1-\e^{-t})^2\log_2(1-\e^{-t}) +(1-\e^{-2t})\log_2(1-\e^{-2t})+(1+\e^{-t})^2\log_2(1+\e^{-t}) \big] .
\end{align*}
\end{example}

%%%%%%%%%%%%%%%%%%%%%%%%%%%%%%%%%%%%%%%%%%%%%%%%%%%%%%%%%%%%%%%%%%%%%%%%%%%%%%%%%%%%%%%%%%%%%
\section{Completely $p$-summing radial multipliers on fermion algebras}
\label{summing-radial-multipliers}

\paragraph{Completely $p$-summing maps}
We start by providing some background on the class of completely $p$-summing maps defined in \eqref{def-norrm-cb-p-summing} where $1 \leq p <\infty$.  We begin by recalling the classical concept of $p$-summing operator, which is useful for understanding the notion of completely $p$-summing operator. Following the book \cite[p.~197]{DJT95}, we say that a bounded operator $T \co X \to Y$ is $p$-summing if there exists a constant $C \geq 0$ such that
\begin{equation}
\label{pq-summing}
\bigg( \sum_{j=1}^n \norm{T(x_j)}_Y^p \bigg)^\frac1p 
\leq C \sup_{y \in B_{X^*}} \bigg( \sum_{j=1}^n |\la y,x_j\ra_{X^*,X}|^p \bigg)^\frac1p
\end{equation}
for any integer $n \geq 1$ and any $x_1, \ldots, x_n \in X$. The $p$-summing norm of $T$, defined as the infimum of the constants $C$ as in \eqref{pq-summing}, is denoted by $\norm{T}_{\pi_{p},X \to Y}$.%The space of $p$-summing linear operators from $X$ to $Y$ is denoted by $\Pi_{p}(X,Y)$. If $X=Y$, we write $\Pi_{p,q}(X)$ and $\Pi_{p}(X)$.

Recall that the class of $p$-summing operators on Banach spaces is a Banach operator ideal by \cite[p.~128]{DeF93}. In particular, we have the inequality 
$$
\norm{R T S}_{\pi_{p}, W \to Z} \leq \norm{R}_{Y \to Z} \norm{T}_{\pi_{p},X \to Y} \norm{S}_{W \to X},
$$ 
with obvious notations. The class of completely $p$-summing maps satisfies a similar property where the operator norm is replaced by the completely bounded norm as the following easy lemma shows. This result \cite[(5.1) p.~51]{Pis98} is in the spirit of the definition of an operator space mapping ideal of \cite[p.~210]{EfR00}.

\begin{lemma}
\label{Lemma-is-bounded}
Suppose that $1 \leq p <\infty$. Let $E,F,G,H$ be operator spaces. Let $T \co E \to F$ be a completely $p$-summing map. If the linear maps $R \co F \to G$ and $S \co H \to E$ are completely bounded, then the linear map $R T S \co H \to G$ is completely $p$-summing and we have
\begin{equation}
\label{Ideal-ellp}
\norm{R T S}_{\pi_{p}^\circ, H \to G}
\leq \norm{R}_{\cb, F \to G} \norm{T}_{\pi_{p}^\circ, E \to F} \norm{S}_{\cb,H \to E}.
\end{equation}
\end{lemma}

The following is a simple useful observation done in \cite{Arh24b} which is a noncommutative analogue of the classical result \cite[Example 2.9 (d) p.~40]{DJT95}, which says that if $\Omega$ is a \textit{finite} measure space then the canonical inclusion $i_p \co \L^\infty(\Omega) \hookrightarrow  \L^p(\Omega)$ is $p$-summing with 
$$
\norm{i_p}_{\pi_p,\L^\infty(\Omega) \to \L^p(\Omega)}=\mu(\Omega)^{\frac{1}{p}}.
$$

\begin{prop}
\label{Prop-inj-finite-avn}
Suppose that $1 \leq p < \infty$. If $\cal{M}$ is an approximately finite-dimensional finite von Neumann algebra equipped with a normal finite faithful trace $\tau$, then the canonical inclusion $i_p \co \cal{M} \hookrightarrow \L^p(\cal{M})$ is completely $p$-summing and
\begin{equation}
\label{norm-ip-bis}
\norm{i_p}_{\pi_{p}^\circ,\cal{M} \to \L^p(\cal{M})}
=(\tau(1))^{\frac{1}{p}}.
\end{equation}
\end{prop}

Let $\cal{M}$ be a von Neumann algebra equipped with a normal finite faithful trace. If $a,b \in \cal{M}$, we denote by $\Mult_{a,b} \co \cal{M} \to \L^p(\cal{M})$, $x \mapsto axb$ the two-sided multiplication map.
The following is a slight variation of \cite[Theorem 2.2 p.~360]{JuP15}, \cite[Theorem 2.4 p.~366]{JuP15} and \cite[Remark 2.2 p.~363]{JuP15}. It provides a nice characterization of the completely $p$-summing norm of linear maps acting on finite-dimensional von Neumann algebras. The proof is similar and we skip the details.
%Noter l'analogie avec \cite[Theorem 5.1]{Oik1}.
\begin{thm}
\label{thm-factorization-positivity}
Suppose that $1 \leq p < \infty$. Let $\cal{M}$ and $\cal{N}$ be finite-dimensional von Neumann algebras such that $\cal{M}$ is equipped with a faithful trace. Let $T \co \cal{M} \to \cal{N}$ be a linear map. Then, the following assertions are equivalent.
\begin{enumerate}
\item $\norm{T}_{\pi_{p}^\circ,\cal{M} \to \cal{N}} \leq C$.

%\item There exists elements $a,b \in \cal{M}$ satisfying $\norm{a}_{\L^{2p}(\cal{M})} \leq 1$, $\norm{b}_{\L^{2p}(\cal{M})} \leq 1$ such that for every $x \in S^p_d \ot_{\min} \cal{M}$ we have 
%\begin{equation}
%\label{Ine-lp-summing}
%\norm{\big(\Id_{S^p_d} \ot T\big)(x)}_{S^p(\cal{N})}
%\leq C\norm{(1 \ot a)x(1 \ot b)}_{S^p(\L^p(\cal{M}))}.
%\end{equation}

\item There exist elements $a,b \in \cal{M}$ satisfying $\norm{a}_{\L^{2p}(\cal{M})} \leq 1$, $\norm{b}_{\L^{2p}(\cal{M})} \leq 1$ and a linear map $\tilde{T} \co \L^p(\cal{M}) \to \cal{N}$ such that
\begin{equation}
\label{sans-fin-33}
T
=\tilde{T} \circ \Mult_{a,b}
\quad \text{and} \quad 
\norm{\tilde{T}}_{\cb, \L^p(\cal{M}) \to \cal{N}} 
\leq C.
\end{equation}
\end{enumerate}
Furthermore, $
\norm{T}_{\pi_{p}^\circ,\cal{M} \to \cal{N}}=\inf\big\{C : C \text{ satisfies the previous conditions}\big\}$.
\end{thm}

%\begin{thm}
%\label{Th-p-sum}
%Suppose that $1 \leq p < \infty$. Let $H$ be a separable Hilbert space with dimension $n \in \{2,\ldots,\infty\}$. A function $g \co \{0,\ldots,n\} \to \mathbb{C}$ induces a completely $p$-summing radial multiplier $R_{g} \co \scr{C}(H) \to \scr{C}(H)$ if and only if it induces a completely bounded radial multiplier $R_g \co \L^p(\scr{C}(H)) \to \scr{C}(H)$. In this case, we have
%\begin{equation}
%\label{norm-p-sum}
%\norm{R_g}_{\pi_p^\circ,\scr{C}(H) \to \scr{C}(H)}
%%=\norm{M_f}_{\cb,\L^p(\scr{C}(H)) \to \L^\infty(\scr{C}(H))}
%=\norm{R_g}_{\cb, \L^p(\scr{C}(H)) \to \scr{C}(H)}.
%\end{equation}
%\end{thm}

We will use the two following results.

\begin{lemma}
\label{magic-434}
For any elements $c$ and $d$ belonging to the unit ball of the Banach space $\L^4(\scr{C}(H))$, the linear map
\begin{equation}
\label{Theta-magic}
\L^2(\scr{C}(H)) \to \L^2(\Omega_n,\L^2(\scr{C}(H))),\,  
x \mapsto (1 \ot c) \beta(x) (1 \ot d)
\end{equation}
is a (complete) contraction, where the map $\beta \co \scr{C}(H) \to \L^\infty(\Omega_n) \otvn \scr{C}(H)$ is defined in \eqref{actions-alpha-beta}.
\end{lemma}

\begin{proof}
Observe that by \cite[p.~139]{Pis2}, the operator spaces $\L^2(\scr{C}(H))$ and $\L^2(\Omega_n,\L^2(\scr{C}(H)))$ are operator Hilbert spaces. Therefore, by \cite[Proposition 7.2 (iii) p.~127]{Pis2}, to establish complete contractivity, it is sufficient to prove that the map defined in \eqref{Theta-magic} is contractive.

Let $x \in \L^2(\scr{C}(H))$. We will use the $*$-homomorphism $\alpha \co \scr{C}(H) \to \scr{C}(H) \otvn  \L^\infty(\Omega_n)$ of \eqref{actions-alpha-beta} which defines an ergodic action by Proposition \ref{prop-ergodic}. Consequently, by Proposition \ref{Prop-transfer-1}, we have a contraction $\alpha \co \L^1(\scr{C}(H)) \to \L^\infty(\scr{C}(H),\L^1(\Omega_n))$. By interpolation with the compatible contractive map $\alpha \co \L^\infty(\scr{C}(H)) \to \L^\infty(\scr{C}(H)) \otvn \L^\infty(\Omega_n)$, we see that $\alpha$ induces a  contraction $\alpha \co \L^2(\scr{C}(H)) \to \L^\infty(\scr{C}(H),\L^2(\Omega_n))$. By \eqref{LinftyLp-norms} applied with $q=2$, we deduce that
\begin{equation}
\label{}
\bnorm{(c \ot 1)\alpha(x)(d \ot 1)}_{\L^2(\scr{C}(H),\L^2(\Omega_n))}   
\leq \norm{\alpha(x)}_{\L^\infty(\scr{C}(H),\L^2(\Omega_n))}
%\leq\norm{\alpha(x)}_{\L^\infty(\scr{C}(H),\L^2(\Omega_n))}
\leq \norm{x}_{\L^2(\scr{C}(H))}.
\end{equation}
This means that the map
\begin{equation*}
\L^2(\scr{C}(H)) \to \L^2(\scr{C}(H),\L^2(\Omega_n)),\,  
x \mapsto (c \ot 1) \alpha(x) (d \ot 1)
\end{equation*}
is a well-defined contraction. Since $\beta=\flip \circ \alpha$, the map defined in \eqref{Theta-magic} is a contraction by Fubini's theorem, described in \eqref{Fubini}.  
\end{proof}

\begin{prop}
Suppose that $1 \leq p \leq \infty$. Assume that $a$ and $b$ are positive elements in the unit ball of the Banach space $\L^{2p}(\scr{C}(H))$. Then the map
\begin{equation}
\label{Theta-bbis}
\Theta \co \L^p(\scr{C}(H)) \to \L^p(\Omega_n,\L^p(\scr{C}(H))),\,  
x \mapsto (1 \ot a) \beta(x) (1 \ot b)
\end{equation}
is a well-defined complete contraction, where the $*$-homomorphism $\beta \co \scr{C}(H) \to \L^\infty(\Omega_n) \otvn \scr{C}(H)$ is defined in \eqref{actions-alpha-beta}.
\end{prop} 

\begin{proof}
Suppose that $2 \leq p \leq \infty$. For any complex number $z$ in the closed strip $\ovl{S} \ov{\mathrm{def}}{=} \{z \in \mathbb{C} : 0 \leq \Re z \leq 1 \}$, we consider the linear map
\begin{equation}
\label{def-Tz}
T_z \co \scr{C}(H) \to \L^1(\Omega_n,\L^1(\scr{C}(H))),\quad x \mapsto \big(1 \ot a^{\frac{pz}{2}}\big) \beta(x) \big(1 \ot b^{\frac{pz}{2}}\big).
\end{equation}
%where $a^{\frac{pz}{2}}$ and $b^{\frac{pz}{2}}$ are defined as in \cite{StZ79}. \textbf{See \cite{StZ79} \cite{ABHN11} Nikolski p49} 
Since $a$ and $b$ are positive elements in the unit ball of $\L^{2p}(\scr{C}(H))$, the elements $a^{\frac{p}{2}}$ and $b^{\frac{p}{2}}$ belong to the unit ball of the Banach space $\L^4(\scr{C}(H))$. Note that
\begin{equation}
\label{unitaries}
\bnorm{1 \ot a^{\frac{\i t p}{2}}}_{\L^\infty(\Omega_n) \otvn \scr{C}(H)}
=\norm{1}_{\L^\infty(\Omega_n)} \bnorm{a^{\frac{\i t p}{2}}}_{\scr{C}(H)}
=1
\end{equation}
and similarly
\begin{equation}
\label{unitaries-2}
\bnorm{1 \ot b^{\frac{\i t p}{2}}}_{\L^\infty(\Omega_n) \otvn \scr{C}(H)}
=1.
\end{equation}
For any $t \in \R$ and any $x \in \scr{C}(H)$, we deduce with H\"older's inequality and Lemma \ref{magic-434} that
\begin{align*}
\MoveEqLeft
\norm{T_{1+\i t}(x)}_{\L^2(\Omega_n,\L^2(\scr{C}(H)))}            
\ov{\eqref{def-Tz}}{=} \bnorm{\big(1 \ot a^{\frac{p+\i t p}{2}}\big) \beta(x) \big(1 \ot b^{\frac{p+\i t p}{2}}\big)}_{\L^2(\Omega_n,\L^2(\scr{C}(H)))} \\
&\ov{\eqref{Holder}}{\leq} \bnorm{1 \ot a^{\frac{\i t p}{2}}}_{\L^\infty(\L^\infty)} \bnorm{\big(1 \ot a^{\frac{p}{2}}\big) \beta(x) \big(1 \ot b^{\frac{p}{2}}\big)}_{\L^2(\Omega_n,\L^2(\scr{C}(H)))} \bnorm{1 \ot b^{\frac{\i t p}{2}}}_{\L^\infty(\L^\infty)} \\
&\ov{\eqref{unitaries}\eqref{unitaries-2}}{=} \bnorm{\big(1 \ot a^{\frac{p}{2}}\big) \beta(x) \big(1 \ot b^{\frac{p}{2}}\big)}_{\L^2(\Omega,\L^2(\scr{C}(H)))}
\ov{\eqref{Theta-magic}}{\leq} \norm{x}_{\L^2(\scr{C}(H))}.
\end{align*}
We infer that the map $T_{1+\i t} \co \L^2(\scr{C}(H)) \to \L^2(\Omega_n,\L^2(\scr{C}(H)))$ is contractive. Note that by \cite[p.~139]{Pis2}, the operator space $\L^2(\scr{C}(H))$ is an operator Hilbert space. We conclude that this map is even completely contractive by \cite[Proposition 7.2 (iii) p.~127]{Pis2}.

For any $t \in \R$ and any $x \in \scr{C}(H)$, we have
\begin{align*}
\MoveEqLeft
T_{\i t}(x)         
\ov{\eqref{def-Tz}}{=} \big(1 \ot a^{\frac{\i t p}{2}}\big) \beta(x) \big(1 \ot b^{\frac{\i t p}{2}}\big).
%&=\norm{\beta(x)}_{\L^\infty(\Omega_n) \otvn \scr{C}(H)}.
\end{align*} 
Note that the $*$-homomorphism $\beta \co \scr{C}(H) \to \L^\infty(\Omega_n) \otvn \scr{C}(H)$ is completely contractive by \cite[Proposition 1.2.4 p.~5]{BLM04}. Moreover, it is elementary with the equalities \eqref{unitaries} and \eqref{unitaries-2} to check that the two-sided multiplication map 
$$
\L^\infty(\Omega_n) \otvn \scr{C}(H) \to \L^\infty(\Omega_n) \otvn \scr{C}(H),\, y \mapsto (1 \ot a^{\frac{\i t p}{2}})y(1 \ot b^{\frac{\i t p}{2}})
$$ 
is also completely contractive (and even contractively decomposable by \cite[Exercise 12.1 p.~251]{Pis2}).
 As a composition of the completely contracting maps, the operator $T_{\i t} \co \L^\infty(\scr{C}(H)) \to \L^\infty(\Omega_n) \otvn \scr{C}(H)$ is completely contractive. 

Moreover, for any $x \in \scr{C}(H)$, the function $\ovl{S} \to \L^1(\Omega_n,\L^1(\scr{C}(H)))$, $z \mapsto T_z(x)$ is continuous and bounded %(pas totalement triviale) 
on the closed strip $\ovl{S}$ and analytic on the open strip $S$. 
%Finally, for any $x \in \scr{C}(H)$ the functions $t \to T_{\i t}(x)$ and $t \mapsto T_{1+\i t}(x)$ take values in $Y_0\L^2(\Omega_n,\L^2(\scr{C}(H)))$ and $\L^\infty(\Omega_n) \otvn \scr{C}(H)$, respectively,
%the functions $\R \to \L^2(\Omega_n,\L^2(\scr{C}(H)))$, $t \mapsto T_{1+\i t}(x)$ and $\R \to \L^\infty(\Omega_n) \otvn \scr{C}(H)$, $t \mapsto T_{\i t}(x)$ are continuous. 
Observe that by the reiteration theorem \cite[Theorem 4.6.1 p.~101]{BeL76} we have
\begin{equation}
\label{ident5}
\L^p(\scr{C}(H))
=(\L^\infty(\scr{C}(H)), \L^2(\scr{C}(H)))_{\frac{2}{p}}
\end{equation}
and
\begin{equation}
\label{ident4}
\L^p(\Omega_n,\L^p(\scr{C}(H)))
=(\L^\infty(\Omega_n) \otvn \scr{C}(H),\L^2(\Omega_n,\L^2(\scr{C}(H))))_{\frac{2}{p}}.
\end{equation}
By Stein's interpolation theorem (Theorem \ref{thm-Stein}), we conclude by taking $z=\frac{2}{p}$ that the linear map defined in \eqref{Theta-bbis} is a contraction. Actually, it is a complete contraction if we combine Stein's interpolation theorem with \eqref{defnormecb} and replacing the interpolation formulas \eqref{ident5} and \eqref{ident4} by the interpolation formulas 
$$
S^p(\L^p(\scr{C}(H)))
=(S^\infty(\L^\infty(\scr{C}(H))), S^2(\L^2(\scr{C}(H))))_{\frac{2}{p}}
$$
and
$$
S^p(\L^p(\Omega_n,\L^p(\scr{C}(H))))
=(S^\infty(\L^\infty(\Omega_n) \otvn \scr{C}(H)),S^2(\L^2(\Omega_n,\L^2(\scr{C}(H)))))_{\frac{2}{p}}.
$$

The case $1 \leq p \leq 2$ is similar. 
\end{proof}

\paragraph{Description of the completely $p$-summing norm of radial multipliers}
Here, we prove that there exists a simple expression for the completely $p$-summing norm of any radial multiplier $R_\phi \co \scr{C}(H) \to \scr{C}(H)$ viewed as an operator acting on the von Neumann algebra $\scr{C}(H)$. It is equal to the completely bounded norm from $\L^p$ into $\L^\infty$. In the next result, we use the \textit{normalized} trace on the finite-dimensional von Neumann algebra $\scr{C}(H)$.

\begin{thm}
\label{Th-p-sum}
Suppose that $1 \leq p < \infty$. Let $H$ be a finite-dimensional real Hilbert space with dimension $n \geq 2$. Consider a function $\phi \co \{0,\ldots,n\} \to \mathbb{C}$. We have
\begin{equation}
\label{norm-p-sum}
\norm{R_\phi}_{\pi_p^\circ,\scr{C}(H) \to \scr{C}(H)}
=\norm{R_\phi}_{\cb, \L^p(\scr{C}(H)) \to \scr{C}(H)}.
\end{equation}
\end{thm}

\begin{proof}
%$\Leftarrow$: %Suppose that the radial multiplier $R_g \co \L^p(\scr{C}(H)) \to \scr{C}(H)$ is completely bounded. 
If we consider the canonical inclusion $i_p \co \scr{C}(H) \hookrightarrow \L^p(\scr{C}(H))$ then by Proposition \ref{Prop-inj-finite-avn}, %the canonical inclusion $i_p \co \scr{C}(H) \hookrightarrow \L^p(\scr{C}(H))$ is $\ell^p(S^p_d)$-summing 
we have the equality
\begin{equation}
\label{norm-ip}
\norm{i_p}_{\pi_{p}^\circ,\scr{C}(H) \to \L^p(\scr{C}(H))}
\ov{\eqref{norm-ip-bis}}{=} 1.
\end{equation} 
By Lemma \ref{Lemma-is-bounded}, we obtain by composition the estimate %that $R_{\phi} \co \scr{C}(H) \to \scr{C}(H)$ is completely $p$-summing and the inequality
\begin{align*}
\MoveEqLeft
\norm{R_\phi}_{\pi_p^\circ, \scr{C}(H) \to \scr{C}(H)}
\ov{\eqref{Ideal-ellp}}{\leq}
\norm{R_\phi}_{\cb, \L^p(\scr{C}(H)) \to \scr{C}(H)} \norm{i_p}_{\pi_{p}^\circ\pi_{p}^\circ,\scr{C}(H) \to \L^p(\scr{C}(H))} \\  
&\ov{\eqref{norm-ip}}{=} \norm{R_\phi}_{\cb, \L^p(\scr{C}(H)) \to \scr{C}(H)}.
\end{align*}

Now we prove the reverse inequality, which is more involved. Let $\epsi > 0$. 
%$\Rightarrow$: Suppose that we have a completely $p$-summing radial multiplier $R_{\phi} \co \L^\infty(\scr{C}(H)) \to \L^\infty(\scr{C}(H))$. 
By Theorem \ref{thm-factorization-positivity}, there exist elements $a,b \in \L^{2p}(\scr{C}(H))$ satisfying $\norm{a}_{\L^{2p}(\scr{C}(H))} \leq 1$ and $\norm{b}_{\L^{2p}(\scr{C}(H))} \leq 1$ and a linear map $\tilde{R}_{\phi} \co \L^p(\scr{C}(H)) \to \scr{C}(H)$ such that 
\begin{equation}
\label{Facto-Fourier-multiplier}
R_{\phi}
\ov{\eqref{sans-fin-33}}{=} \tilde{R}_{\phi} \circ \Mult_{a,b}
\quad \text{and} \quad 
\norm{\tilde{R}_{\phi}}_{\cb,\L^p(\scr{C}(H)) \to \scr{C}(H)} 
\leq \norm{R_\phi}_{\pi_p^\circ, \scr{C}(H) \to \scr{C}(H)} + \epsi,
\end{equation}
%The algebra $\mathbb{C}[G]$ of finite sums $\sum_s \alpha_s \lambda_s$ is dense in $\L^{p^*}(\VN(G))$ and $M_{\phi}^*(\C[G]) \subset \C[G]$. This shows that $T_{\phi}$ is normal. 
where $\Mult_{a,b} \co \scr{C}(H) \to \L^p(\scr{C}(H))$ is the two-sided multiplication map. Replacing the map $\tilde{R}_{\phi}$ by the map $x \mapsto\tilde{R}_{\phi}(u x v^*)$ where $a=u|a|$ and $b^*=v|b^*|$ are the polar decompositions of the elements $a$ and $b^*$, we can suppose that $a$ and $b$ are positive elements. Indeed, firstly for any $x \in \L^p(\scr{C}(H))$ we have 
$$
\tilde{R}_{\phi}\big(u \Mult_{|a|,|b^*|}(x) v^*\big)
=\tilde{R}_{\phi}\big(u |a|\,  x \, |b^*|v^*\big)
=\tilde{R}_{\phi}(a x b)
\ov{\eqref{Facto-Fourier-multiplier}}{=} R_\phi(x).
$$ 
Secondly, by \cite[Lemma 5.1]{Arh24a}, the two-sided multiplication map $\Mult_{u,v} \co \L^p(\scr{C}(H)) \to \L^p(\scr{C}(H))$, $x \mapsto u x v$ is decomposable with decomposable norm 
$$
\norm{\Mult_{u,v}}_{\dec,\L^p(\scr{C}(H)) \to \L^p(\scr{C}(H))} 
\leq \norm{u}_{\L^\infty(\scr{C}(H))} \norm{v}_{\L^\infty(\scr{C}(H))} 
=1,
$$ 
hence completely bounded by \cite[Proposition 3.30 p.~50]{ArK23} with completely bounded norm 
$$
\norm{\Mult_{u,v}}_{\cb,\L^p(\scr{C}(H)) \to \L^p(\scr{C}(H))}
\leq \norm{\Mult_{u,v}}_{\dec,\L^p(\scr{C}(H)) \to \L^p(\scr{C}(H))}
\leq 1
$$
(an elementary proof of this estimate is also possible). Since the linear map $\tilde{R}_{\phi} \co \L^p(\scr{C}(H)) \to \scr{C}(H)$ is completely bounded, we deduce the estimate 
%by \eqref{ine-tensorisation-os} that the map $\Id_{\L^p(\Omega_n)} \ot \tilde{R}_{g} \co \L^p(\Omega_n,\L^p(\scr{C}(H))) \to \L^p(\Omega_n,\scr{C}(H))$ is completely bounded and that we have
\begin{align}
\MoveEqLeft
\label{Equa-138676}
\norm{\Id_{\L^p(\Omega_n)} \ot \tilde{R}_{\phi}}_{\cb,\L^p(\Omega_n,\L^p(\scr{C}(H))) \to \L^p(\Omega_n,\scr{C}(H))}            
\ov{\eqref{ine-tensorisation-os}}{\leq} \norm{\tilde{R}_{\phi}}_{\cb,\L^p(\scr{C}(H)) \to \scr{C}(H)} \\
&\ov{\eqref{Facto-Fourier-multiplier}}{\leq}  \norm{R_{\phi}}_{\pi_p^\circ, \scr{C}(H) \to \scr{C}(H)} +\epsi. \nonumber
\end{align}

Let $\E \co \L^\infty(\Omega_n) \otvn \scr{C}(H) \to \scr{C}(H)$ be the canonical trace preserving  faithful conditional expectation associated with the trace preserving unital injective $*$-homomorphism $\beta \co \scr{C}(H) \to \L^\infty(\Omega_n) \otvn \scr{C}(H)$ defined in \eqref{actions-alpha-beta}, provided by Proposition \ref{prop-existence-conditional-expectation}. By Proposition \ref{prop-autre-extension}, we have a complete contraction $\beta \co \L^1(\scr{C}(H)) \to \L^{p^*}(\Omega_n,\L^1(\scr{C}(H)))$. Using \eqref{Esp-dual}, we see that the adjoint map of this map identifies to the linear map $\E \co \L^p(\Omega_n,\scr{C}(H)) \to \scr{C}(H)$, which is therefore a complete contraction. Now, we have
\begin{align*}
\MoveEqLeft
\beta R_{\phi}
\ov{\eqref{Equation-Delta}}{=} (\Id \ot R_{\phi}) \beta
\ov{\eqref{Facto-Fourier-multiplier}}{=}(\Id \ot \tilde{R}_{\phi} \Mult_{a,b}) \beta \\
&=(\Id \ot \tilde{R}_{\phi}) (\Id \ot \Mult_{a,b}) \beta
\ov{\eqref{Theta-bbis}}{=} (\Id \ot \tilde{R}_{\phi}) \Theta.           
\end{align*} 
Since $\E \beta=\Id_{\scr{C}(H)}$, we deduce the equality 
$$
R_{\phi}=\E \beta R_{\phi}=\E (\Id \ot \tilde{R}_{\phi}) \Theta.
$$ 
That means that we have the following commutative diagram.
$$
\xymatrix @R=1.5cm @C=2cm{
\L^p(\Omega_n,\L^p(\scr{C}(H))) \ar[r]^{\Id_{\L^p(\Omega_n)} \ot \tilde{R}_{\phi}}   & \ar[d]^{\E} \L^p(\Omega_n,\scr{C}(H)) \\
\L^p(\scr{C}(H)) \ar[r]_{R_\phi} \ar[u]^{\Theta}   & \scr{C}(H) \\
  }
$$
We infer that we have the estimate%the function $g \co \{0,\ldots,n\} \to \mathbb{C}$ induces a completely bounded multiplier $R_{g} \co \L^p(\scr{C}(H)) \to \scr{C}(H)$ and that
\begin{align*}
\MoveEqLeft
\norm{R_{\phi}}_{\cb,\L^p(\scr{C}(H)) \to \scr{C}(H)}
=\norm{\Theta(\Id_{\L^p} \ot \tilde{R}_{\phi})\E}_{\cb, \L^p(\scr{C}(H)) \to \scr{C}(H)} \\
&\leq \norm{\Theta}_{\cb,\L^p \to \L^p(\L^p)} \norm{\Id_{\L^p} \ot \tilde{R}_{\phi}}_{\cb, \L^p(\L^p) \to \L^p(\L^\infty)} \norm{\E}_{\cb,\L^p(\L^\infty) \to \scr{C}(H)} \\
&\ov{\eqref{Equa-138676}}{\leq} \norm{R_\phi}_{\pi_p^\circ, \scr{C}(H) \to \scr{C}(H)}+\epsi.
\end{align*}
Since $\epsi >0$ is arbitrary, we obtain the desired estimate 
$$
\norm{R_{\phi}}_{\cb,\L^p(\scr{C}(H)) \to \scr{C}(H)} 
\leq \norm{R_\phi}_{\pi_p^\circ, \scr{C}(H) \to \scr{C}(H)}.
$$
\end{proof}

Theorem \ref{Th-p-sum} combined with Theorem \ref{th-norm-cb-rad} allows us to use formula \eqref{CEA-formula} to obtain the exact value of the entanglement-assisted classical capacity of these multipliers. 

\begin{thm}
\label{th-capacity-assisted}
Let $H$ be a finite-dimensional real Hilbert space with even dimension $n=2k \geq 2$. Consider a function $\phi \co \{0,\ldots,n\} \to \mathbb{C}$. Suppose that the radial multiplier $R_\phi \co \scr{C}(H) \to \scr{C}(H)$ is a quantum channel. We have
$$
\C_{\EA}(R_\phi)
=-\H(f_\phi), 
$$ 
where $\H(f_\phi)=-\int_{\Omega_n} f_\phi\log_2 f_\phi$ is the Segal entropy of $f_\phi$ defined in \eqref{Segal-entropy} with the normalized integral $\int_{\Omega_n}$.
\end{thm}

\begin{proof}
Note that the adjoint of a radial multiplier is clearly a radial multiplier. Recall that we have $f_\phi \geq 0$ by Proposition \ref{Prop-cp} and that $\norm{f_\phi}_{\L^1(\Omega_n)}=1$ by \eqref{norm-hg}. 
Thus, we can compute in the following way
\begin{align*}
\MoveEqLeft
\C_{\EA}(R_\phi)
\ov{\eqref{CEA-formula}}{=} \frac{1}{\log 2} \frac{\d}{\d p} \big[\norm{R_\phi^*}_{\pi_{p^*}^\circ, \scr{C}(H) \to \scr{C}(H)} \big]|_{p=1}  \\
&\ov{\eqref{norm-p-sum}}{=} \frac{1}{\log 2} \frac{\d}{\d p} \big[\norm{R_\phi^*}_{\cb, \L^{p^*}(\scr{C}(H)) \to \L^\infty(\scr{C}(H))} \big]|_{p=1}       
\ov{\eqref{norm-cb-2}}{=} \frac{1}{\log 2} \frac{\d}{\d p} \big[\norm{f_\phi}_{\L^{p}(\Omega_n)} \big]|_{p=1} \\  
&\ov{\eqref{deriv-norm-p}}{=} -\H(f_\phi)
\ov{\eqref{Segal-entropy}}{=} \int_{\Omega_n} f_\phi\log_2 f_\phi.
\end{align*}
\end{proof}

\begin{remark} \normalfont
In particular, combined with \eqref{mino-Q1} we obtain that 
$$
\max\{-\log_2 {N}-\H(f_\phi),0\} 
\leq \Q(R_\phi) 
\leq \C(R_\phi) 
\leq -\H(f_\phi),
$$
where $\Q(R_\phi)$ and $\C(R_\phi)$ are the quantum capacity and the classical capacity of the quantum channel $R_\phi$. 
\end{remark}

\begin{example} \normalfont
Consider the situation of Example \ref{Ex-multiplier-Fermions} and Example \ref{Example-suite}. If $0 \leq t \leq 1$, we have examined the dephasing channel $T \co \M_2 \to \M_2$, $x \mapsto (1-t)x+tZxZ$ of \cite[p.~155]{Wil17} viewed as a radial multiplier. We have seen that
$$
-\H(f_\phi)
\ov{\eqref{sans-fin-5}}{=} 2+t\log_2 t+(1-t)\log_2(1-t).
$$
Our result for the entanglement-assisted capacity allows us to recover the formula
$$
\C_{\EA}(T)
=2+t\log_2 t+(1-t)\log_2(1-t)
$$
of \cite[Exercise 21.6.2 p.~597]{Wil17}.
\end{example}

%\begin{prop}
%We have
%$$
%\norm{M_g}_{\cb, \L^p(\scr{C}(H)) \to \L^\infty(\scr{C}(H))} 
%\leq \norm{M_g}_{\cb, \L^p(\Omega_n) \to \L^\infty(\Omega_n)}
%=\norm{g}_{\ell^p}
%$$
%\end{prop}
%
%\begin{proof}
%$$
%\xymatrix @R=1.5cm @C=2cm{
%\L^1(\scr{C}(H)),\L^1(\Omega_n)) \ar[r]^{\Id_{\L^p(\scr{C}(H))} \ot M_g}   & \ar[d]^{\E} \L^p(\scr{C}(H),\L^p(\Omega_n)) \\
%\L^1(\scr{C}(H)) \ar[r]_{M_g} \ar[u]^{\Delta}   & \L^p(\scr{C}(H)) \\
  %}
%$$
%\end{proof}

%%%%%%%%%%%%%%%%%%%%%%%%%%%%%%%%%%%%%%%%%%%%%%%%%
\section{Future directions and open questions}
\label{sec-future}
%%%%%%%%%%%%%%%%%%%%%%%%%%%%%%%%%%%%%%%%%%%%%%%%%

The exact computation of the classical capacity, the quantum capacity and the private capacity of radial multipliers is an interesting open question. A characterization of these quantum channels which are entanglement breaking channels or PPT is missing. Consequently, we intend in a sequel to delve into various properties of these quantum channels and explore variations and generalizations of these multipliers.

\paragraph{Declaration of interest} None.

\paragraph{Competing interests} The author declares that he has no competing interests.

\paragraph{Data availability} No data sets were generated during this study.
%s and Grants

\paragraph{Acknowledgment} I would like to thank Li Gao for a very brief yet enlightening conversation regarding the <<proof>> of \cite[Lemma 3.11 p.~3434]{GJL20}. The proof of Proposition \ref{prop-Gao} is inspired by this discussion. The author wishes to express profound gratitude to the referee whose meticulous and comprehensive review significantly contributed to the improvement of this paper. The detailed comments and constructive suggestions provided were invaluable in enhancing both the accuracy and clarity of the mathematical arguments presented herein. This thorough engagement has undoubtedly enriched the quality of the work.

%The author gratefully acknowledges the support from the French National Research Agency grant ANR-18-CE40-0021 (project HASCON). We are thankful to Shouhei Honda, Naotaka Kajino, Jun Kigami, Bogdan Nica, Arup Kumar Pal, Jiayin Pan, El Maati Ouhabaz, Adrián González-Pérez, Sang-Gyun Youn, and Melchior Wirth for their valuable insights during brief discussions. Our appreciation extends to Li Gao and Bogdan Nica for their feedback and corrections, and to Christopher Sogge for pointing us to the reference \cite{BSS21}. Finally, a special note of gratitude goes to Bruno Iochum for the numerous discussions we had between the inception of the first version of this paper and the release of its preprint \cite{IoZ23}.

%%%%%%%%%%%%%%%%%%%%%%%%%%%%%%%%%%%%%%%%%%%%%%%%%%%%%%%%%%%%%%%%%%%%%%%%%%%%%%%%%%%%%%%%%%%%%%%%%%%%%%%%%%%%%%%%%%%%%%%%%%%%%
%%%%%%%%%%%%%%%%%%%%%%%%%%%%%%%%%%%%%%%%%%%%%%%%%%%%%%%%%%%%%%%%%%%%%%%%%%%%%%%%%%%%%%%%%%%%%%%%%%%%%%%%%%%%%%%%%%%%%%%%%%%%%

{\footnotesize

\vspace{0.2cm}

\noindent C\'edric Arhancet\\ 
\noindent 6 rue Didier Daurat, 81000 Albi, France\\
URL: \href{http://sites.google.com/site/cedricarhancet}{https://sites.google.com/site/cedricarhancet}\\
cedric.arhancet@protonmail.com\\%cedric.arhancet@gmail.com 
ORCID: 0000-0002-5179-6972 
}

\end{document}